\documentclass[12pt]{article}

\usepackage{natbib}

\renewenvironment{thebibliography}[1]{%
\begin{oldthebibliography}{#1}
    \setlength{\parskip}{0.00mm}
    \setlength{\itemsep}{0.00mm}}
{\end{oldthebibliography}}
    
\usepackage[french,english]{babel} 
\usepackage[T1]{fontenc} 

\usepackage{libertine}

\usepackage{tcolorbox}
\tcbuselibrary{breakable}

\usepackage{float}  

\usepackage{tikz}
\usetikzlibrary{arrows.meta, decorations.pathreplacing}

\usepackage{longtable}
\usepackage{dcolumn}
\usepackage{rotating}

\usepackage{framed}

\usepackage{comment}

\usepackage{amsthm,amssymb,mathrsfs,bm}
\usepackage[tbtags]{amsmath}
\usepackage{subeqnarray}
\usepackage{eurosym}
\usepackage{cancel}

\usepackage{relsize,exscale}
\usepackage{tabularx}
\usepackage{multirow}
\usepackage{booktabs}
\usepackage{graphicx}
\usepackage{nicematrix}
\usepackage{setspace}
\usepackage{lscape}
\usepackage{color}

\usepackage[
  left=2cm,
  right=2cm,
  top=2.5cm,
  bottom=2.5cm
]{geometry}
\usepackage{mdframed} 

\definecolor{dark-blue}{RGB}{21,85,212}
\definecolor{light-blue}{RGB}{51,153,255}
\definecolor{dark-green}{RGB}{50,150,86}
\definecolor{dandelion}{RGB}{212,174,21}
\definecolor{bloud}{RGB}{212,21,21}
\definecolor{NewPurple}{RGB}{127,0,255}

\usepackage{url}
\usepackage{enumerate}
\usepackage{fancyhdr}

\usepackage{subcaption}

\usepackage{hyperref}
\hypersetup{
    colorlinks=true,
    linkcolor=blue,
    filecolor=magenta,      
    urlcolor=blue,
    citecolor=blue,
    pdftitle={Overleaf Example},
    pdfpagemode=FullScreen,
    }
    
\usepackage[hang,splitrule]{footmisc} 


\usepackage[sf,bf,tiny,center]{titlesec} 
\titlelabel{\thetitle.\enspace}
\usepackage{titletoc}

\usepackage{nicefrac}

\renewcommand{\thefootnote}{\fnsymbol{footnote}}
\usepackage{caption}
\newcommand{\parhead}[1]{\noindent\textbf{#1}.}

\newcommand{\states}{\mathcal S}
\newcommand{\pop}{\mathcal N}
\newcommand{\incomes}{\mathcal Y}
\newcommand{\A}{\mathcal A}
\newcommand{\SocietiesFixed}{\mathcal F_\states}

\newcommand{\labels}{\mathcal I}
\newcommand{\types}{\mathcal T}
\newcommand{\domain}{\mathcal D}
\newcommand{\permutations}{\mathcal P}

\newcommand{\Ninteger}{\mathbb N}
\newcommand{\real}{\mathbb R}

\DeclareMathOperator*{\Wpref}{\succcurlyeq}

\DeclareMathOperator*{\defn}{=}
\DeclareMathOperator*{\argmin}{arg\,min}

\DeclareMathOperator{\MIN}{min}
\DeclareMathOperator{\var}{var}

\DeclareMathOperator*{\expect}{\mathbb{E}}

\newtheorem{theorem}{Theorem}

\newtheorem{axiom}{Axiom}

\newenvironment{axiomp} 
{%
  \addtocounter{axiom}{-1}%
  \begin{axiom}
}
{%
  \end{axiom}
  
}

\newtheorem{definition}{Definition}

\newtheorem{property}{Property}

\title{Opportunity-Sensitive Social Welfare\,\footnotemark[1]}

\author{Wienand, Tim \textcircled{r} Brice Magdalou \textcircled{r} Richard Nock \textcircled{r} Paul Hufe}

\begin{document}
\onehalfspacing
\maketitle
\thispagestyle{empty}
\setlength{\headheight}{14.5pt}

\footnotetext[1]{\href{https://www.aeaweb.org/journals/policies/random-author-order/search?RandomAuthorsSearch\%5Bsearch\%5D=dKWV5sC\_Fdb6}{Random author order}. \href{mailto:paul.hufe@bristol.ac.uk}{Paul Hufe}: 
University of Bristol, CESifo, HCEO, IFS, and IZA, UK. \href{mailto:brice.magdalou@umontpellier.fr}{Brice Magdalou}: CEE-M, Univ. Montpellier, CNRS, INRAE, Institut Agro, France. \href{mailto:richardnock@google.com}{Richard Nock}: Google Research. \href{mailto:wienand@ese.eur.nl}{Tim Wienand}: Erasmus University Rotterdam, Tinbergen Institute, Netherlands. We acknowledge financial support from ANR (ANR-24-CE26-3823-02) and UKRI (MR/X033333/1). We are grateful to seminar and workshop audiences at the University of Bristol, the University of Montpellier, SOFI Stockholm, Tinbergen Institute, the French-Japanese Webinar in Economics (FJWE), the Normative Economics \& Economic Policy Webinar (NE\&EP), IIPF 2026,  and the FUR Conference 2026 for their constructive feedback and suggestions. This paper has benefited from discussions with Mohammed Abdellaoui, Gilles Dufrénot, Francisco Ferreira, Nicolas Gravel, Markus J\"antti, Paola Manzini, Juan Moreno-Ternero, Mauricio Ribeiro, and Peter Wakker. All remaining errors are our own.}


\renewcommand{\thefootnote}{\arabic{footnote}}

\pagestyle{fancy}
\fancyhf{}
\chead{\textsc{Opportunity-Sensitive Social Welfare}}
\renewcommand{\headrulewidth}{0pt}
\cfoot{\thepage} 

\begin{abstract}

\noindent

We develop a framework to evaluate income distributions from the perspective of an opportunity-egalitarian social planner. Through an original axiomatic route that bridges normative economics and decision theory under ambiguity, we characterize a class of social welfare functions governed by a single parameter capturing the planner's aversion to inequality of opportunity. 
We obtain a `mean-divergence' representation that separates economic efficiency from a penalty for inequality of opportunity, as well as a stochastic dominance criterion for robust welfare comparisons.  
We apply these theoretical results to analyze opportunity-sensitive social welfare in the United States across birth cohorts and regions. Mapping standard estimates of intergenerational mobility into our framework, we show that welfare rankings vary substantially with the assumed degree of inequality-of-opportunity aversion. We transparently gauge the extent of this normative ambiguity by providing bounds on the ranking of birth cohorts and regions, respectively.

\vspace{0.4cm}
\end{abstract}

\newpage
\setcounter{page}{1}


\section{Introduction}
\label{introduction}

Equality of opportunity is widely regarded as a fundamental ideal of distributive justice \citep{RT16}. While inequalities resulting from individual effort or well-informed choices may be morally acceptable, they become unfair when driven by circumstances beyond the individual's control, such as parental income, race, or biological sex. Assessing the extent of these unfair inequalities is therefore essential to evaluating a society's distributive performance and to guiding the design of fair public policies.

In this paper, we develop a new normative framework to evaluate income distributions from the perspective of an opportunity-egalitarian social planner. We characterize a novel class of opportunity-sensitive social welfare functions governed by a single, transparent parameter that captures the planner’s degree of aversion to inequality of opportunity. The resulting framework is both normatively transparent and easy to implement in empirical applications. The inequality-of-opportunity aversion parameter governs how rapidly the planner's concern shifts toward disadvantaged circumstance groups and directly translates into endogenous welfare weights for each of them. 

We take these measures to the data and evaluate income distributions in the United States across birth cohorts and regions. In line with the existing literature, we estimate intergenerational mobility and interpret the resulting estimates as proxy indicators of equality of opportunity. These estimates replicate well-established findings on mobility patterns in the US, both across birth cohorts and commuting zones. However, by interpreting these estimates through our framework, we also show that their welfare implications are ambiguous as they crucially depend on assumptions about the appropriate degree of inequality-of-opportunity aversion.

To fix ideas, consider a stylized example of two societies with identical marginal income distributions. In the first society, an individual's income prospects are entirely determined by circumstances beyond the individual's control, whereas in the second, they are independent of such factors. A utilitarian planner, who evaluates welfare solely on the basis of the marginal distribution of anonymous outcomes, would be strictly indifferent between the two. By contrast, an opportunity-egalitarian planner would regard the first society as substantially less desirable. We capture this intuition by partitioning the population into distinct `types' based on observable circumstances and by penalizing disparities in type-specific welfare derived from these conditional income prospects, according to the planner’s degree of inequality-of-opportunity aversion.

\parhead{Theoretical foundations} In the first part of the paper, we provide a complete axiomatic characterization of this class of social welfare functions. To cleanly isolate our core normative requirements, we begin by establishing a baseline representation under the assumption of a uniform demographic distribution across types. The resulting representation takes a double-expectation form, in which the welfare of each type, evaluated by the expected utility of its income distribution, is aggregated via a concave transformation that captures the social planner’s inequality-of-opportunity aversion. This aggregation is achieved by taking the expectation over the demographic distribution of types. By strengthening our independence axioms, we further show that this transformation reduces to a one-parameter exponential family. Moreover, under the usual monotonicity and scale-invariance requirements and the assumption that the social planner is averse to inequality of opportunity, the utility function is necessarily logarithmic in income. We then extend the analysis to societies with arbitrary, heterogeneous demographic weights. This realistic setting yields a highly tractable criterion. Formally, letting $q(s)$ denote the demographic weight of type $s$ and $U(\pi_s)$ the expected utility of its income distribution, social welfare  takes the following explicit form when the planner is averse to unequal opportunities ($\theta>0)$:
\begin{equation*}
    V_{\theta,u}(f,q) = -\frac{1}{\theta} \ln \left( \sum_{s \in \mathcal{S}} q(s) e^{-\theta U(\pi_s)} \right)\,.
\end{equation*}
This structure admits multiple insightful representations. Beyond this baseline functional form, we show that the planner's evaluation can also be expressed as a minimization problem involving the Kullback-Leibler divergence, which assigns endogenous normative weights to socially disadvantaged types. Furthermore, we establish an exact, closed-form `mean-divergence' representation. This specific formulation uses a Bregman divergence to structurally separate economic efficiency from the inequality-of-opportunity penalty, a trade-off that also naturally reduces to a mean-variance approximation. Finally, we exploit this rich analytical structure to establish a stochastic dominance criterion that allows for an unambiguous ranking of societies across a broad class of completely alternating transformations.

\parhead{Empirical analysis} 
In the second part of the paper, we use this framework to assess the welfare implications of intergenerational mobility estimates in the United States. In line with the empirical literature on intergenerational income mobility, we quantify the extent to which children's incomes are predicted by their parents' incomes and interpret these estimates as proxy indicators of equality of opportunity. Mapping these empirical estimates into our framework, we assess income distributions from the perspective of an opportunity-egalitarian social planner. Concretely, we use the Panel Study of Income Dynamics (PSID) to estimate mobility at the national level across the 1940s--1980s birth cohorts, and data from Opportunity Insights to estimate mobility differences across Commuting Zones (CZs) for children born around 1980. 
We then show that the associated welfare comparisons are subject to normative ambiguity, i.e., they crucially depend on assumptions about the appropriate degree of inequality-of-opportunity aversion. We quantify the extent of normative ambiguity in welfare comparisons by constructing dominance sets for the rankings of societies over plausible ranges of inequality-of-opportunity aversion. The resulting bounds provide a transparent tool for welfare comparisons across societies. For example, we show that the substantial drop in intergenerational mobility between the 1950s and the 1960s does not lead to a corresponding drop in social welfare when assuming a low but plausible extent of inequality-of-opportunity aversion. More generally, our analysis reveals substantial normative ambiguity both across birth cohorts and regions in the United States: Unambiguous welfare rankings are rare, even when restricting the analysis to a small subset of the full range of potential parameter values. 

\parhead{Related literature} Our paper contributes to two main strands of the economics literature: the theoretical literature on normative economics and the empirical literature on intergenerational mobility and equality of opportunity.

First, within normative economics, we contribute to the theoretical literature on equality of opportunity and the assessment of risky social situations. Firmly embedded within the ex-ante perspective, our approach evaluates societies by treating currently observed type-specific distributions as objective income prospects. A defining feature of our framework is that we do not impose expected utility representations or von Neumann-Morgenstern rationality \textit{a priori} on individuals, a standard premise in foundational works like \cite{Fleurbaey10} and \cite{GKPS10}. Instead, this expected utility structure emerges endogenously from the social planner's evaluation of aggregate risks. Furthermore, while our social welfare function shares the `generalized utilitarian' structure of \cite{GKPS10}, we reach it via a less demanding axiomatic route that isolates the core normative evaluation from population variations. This deliberate separation allows our framework to explicitly tackle the aggregation of unequally sized types, distinguishing our approach from recent theoretical contributions that assume symmetric groups or uniform demographic weights \citep{Moramarco2024, AFGK26}.\footnote{Although grounded in an ex-ante framework, our methodology establishes formal connections to ex-post egalitarianism \citep{GKPS12} and ex-post prioritarianism \citep{Adler25}. Finally, for a recent and insightful exploration of the boundaries between welfarism and equality of opportunity, see \cite{Carroll25}, who characterizes the conditions under which these two criteria coincide in resource allocation problems.} 


Second, we contribute to the empirical literature on intergenerational mobility and equality of opportunity \citep{Chetty2014, Chetty2026, davis2026decline, JKN25, Nybom2024, Mello2025, HKP22, RT16, Ramos2016}. The extent of intergenerational mobility is commonly interpreted as a proxy for equal opportunities in a society \citep{Mazumder2023, Chetty2014, Alesina2018}. Despite this normative motivation, much of the existing literature on intergenerational mobility sidesteps discussions about its normative underpinnings. We address this shortcoming by embedding standard intergenerational mobility estimates into our axiomatically characterized social welfare function. This mapping makes the normative foundations explicit and allows us to evaluate the implications of intergenerational mobility estimates for social welfare assessments. Relatedly, several recent papers have explored the connection between the literature on equality of opportunity and intergenerational mobility \citep{Mazumder2023, Jantti2020, Adermon2026}. These papers mostly analyze whether estimates from both branches of the literature yield similar empirical conclusions. Our approach differs by incorporating standard intergenerational mobility estimates into an axiomatically grounded analytical framework that represents the preferences of an opportunity-egalitarian social planner.

Our contribution falls within the field of normative economics, but the formal framework remains relevant to other areas of economic theory. In particular, it exhibits a structural connection to decision theory under ambiguity. A detailed mapping of our axioms and results to the ambiguity literature is provided in Appendix~\ref{app: risk}. 

\parhead{Outline} The remainder of the paper is organized as follows. Section~\ref{sec-setup} establishes the formal framework. Section~\ref{sec-preferences} introduces the baseline axioms governing the planner's preferences over uniform type partitions. Section~\ref{sec-representation-uniform} derives the main representation theorems for the social welfare function in this uniform setting. Section~\ref{sec-representation-extended} extends this characterization to the domain of arbitrary demographic distributions across types. Section~\ref{sec-representation-equivalent} explores equivalent representations of the welfare criterion and establishes a stochastic dominance criterion. Section~\ref{sec-properties} examines the basic properties of the welfare function. In Section \ref{sec-empirics}, we use our framework to evaluate opportunity-sensitive social welfare in the US across birth cohorts and commuting zones. Finally, Section~\ref{sec-conclusion} concludes the paper.

\section{Setup}
\label{sec-setup}
For a topological space $\A$, we indicate by $|\A|$ its cardinality whenever $\A$ is finite. We let $\Sigma_\A$ be the Borel $\sigma$-algebra on $\A$, and we denote by $\Delta(\A)$ the set of finitely supported probability measures on $(\A,\Sigma_\A)$, endowed with the weak topology. For $p \in \Delta(\A)$ we denote by $p(a)$, instead of $p(\{a\})$, the probability assigned to $a \in \A$. 

Consider a finite population $\pop$ consisting of at least two individuals. A \emph{type} $s \in \states$, with $\states$ a partition of~$\pop$, is a population subgroup sharing characteristics that give rise to unequal opportunities. We focus on subgroups defined by parental income; however, the framework is general enough to incorporate other sources of unequal opportunities, such as biological sex, ethnicity, or neighborhood quality during childhood. We do not impose an order relation over types. We denote by $q \in \Delta(\states)$ the demographic type weights, defined by the population shares $q(s) = |s|/|\pop|$ for all $s \in \states$, which represent the distribution of individuals over $\states$ and which is known by the social planner.

An \emph{income distribution} is denoted $\pi \in \Delta(\incomes)$, with $\incomes$ a compact interval in $(0,\infty)$. The mean income is $\mu(\pi) = \sum_{y \in \incomes} \pi(y) y$.\footnote{Note that, since any $\pi \in \Delta(\incomes)$ is finitely supported, the sum contains only a finite number of non-zero terms $\pi(y)$. It is therefore well-defined and is equivalent to a summation over the finite support of $\pi$.} Type-specific income distributions are denoted by $\pi_s \in \Delta(\incomes)$, for all $s \in \states$. 

An \emph{opportunity profile} $f \in \SocietiesFixed$ assigns a type-specific income distribution $f(s) = \pi_s \in \Delta(\incomes)$ to each $s\in \states$. The set of all opportunity profiles $\SocietiesFixed = \Delta(\incomes)^{|\states|}$, for a given partition~$\states$, is endowed with the corresponding product topology. Denoting by $\expect\!_q$ the expectation operator according to $q \in \Delta(\states)$, the income distribution across the entire population is $\expect_q[f] = \sum_{s \in \states} q(s) \pi_s$, and the mean income is $\mu \left( \expect_q[f] \right) = \sum_{s \in \states} q(s) \mu(\pi_s)$. An equal-opportunity profile that assigns the same income distribution $\pi \in \Delta(\incomes)$ to all types $s \in \states$, is denoted by $f_\pi \in \SocietiesFixed$. 

Finally, a \emph{society} is a pair $(f,q) \in \SocietiesFixed \times \Delta(\states)$.

\parhead{Example} Suppose there are two types in society: 20\% of the population have rich parents (R), and 80\% have poor parents (P). The set of types is therefore $\states$ = \{R, P\}. The set of observed incomes is \{\$1, \$2\}. The type-specific income distributions are $\pi_{\textnormal{R}}$ = (0.1, 0.9) and $\pi_\textnormal{P}$ = (0.6, 0.4). This society is characterized by the following opportunity profile $f$ and demographic type weights $q$:
\[
f\ \defn\ 
\begin{pNiceMatrix}[first-row, first-col]
               & \textnormal{\$1} & \textnormal{\$2} \\
\textnormal{R} & $0.1$         & $0.9$         \\
\textnormal{P} & $0.6$         & $0.4$         \\
\end{pNiceMatrix}\,,\quad \textnormal{and}\quad 
 q \defn\ \begin{pNiceMatrix}
  $0.2$ \\
  $0.8$ \\
\end{pNiceMatrix}\,.
\]
The overall income distribution is $\expect_q[f]$ = (0.5, 0.5) = 0.2 $\times$ (0.1, 0.9) + 0.8 $\times$ (0.6, 0.4), with mean income $\mu(\expect_q[f])$ = \$1.5. Moreover, the mean incomes of individuals with rich and poor parents are $\mu(\pi_{\textnormal{R}})$ = \$1.9 and $\mu(\pi_{\textnormal{P}})$ = \$1.4, respectively.

\section{Social planner's preferences}
\label{sec-preferences}
Initially, we assume that the demographic type weights $q \in \Delta(\states)$ are fixed and uniform across all societies, with $q(s) = 1 / |\states|$ for all $s \in \states$. This assumption allows us to separate considerations relating to equality of opportunity from those relating to variation in population composition. It will be relaxed in Section~\ref{sec-representation-extended}.

The social planner's preferences are represented by a binary relation $\Wpref$ on $\SocietiesFixed$: $f \Wpref f'$ means that opportunity profile $f \in \SocietiesFixed$ is weakly preferred to profile $f' \in \SocietiesFixed$, given the uniform distribution of types in $\states$. We denote by $\sim$ and $\succ$ the symmetric and asymmetric parts of $\Wpref$, respectively. We assume that $\Wpref$ is a weak order (transitive and complete), non-degenerate (the asymmetric part is non-empty), and continuous (the upper and lower contour sets are closed). Under these assumptions, there exists a continuous social welfare function $W: \SocietiesFixed \rightarrow \real$ which represents $\Wpref$ so that, for all $f, f' \in \SocietiesFixed$, we have $f \Wpref f'\ \Longleftrightarrow \ W(f) \geq W(f')$. 

The first axiom we impose on $\Wpref$ reflects unanimity in social improvement: 
If the social planner prefers all type-specific income distributions for one opportunity profile over those for another, then they must prefer the former profile to the latter. 

\begin{axiom}[\textsc{Strict Pareto}] \label{ax-Pareto}
For all $f = {\left( \pi_s \right)}_{s \in \states} \in \SocietiesFixed$ and all $f'={\left( \pi'_s \right)}_{s \in \states} \in \SocietiesFixed$, we have: $\forall s \in \states$, $f_{\pi_s} \Wpref f_{\pi'_s}\ \Rightarrow\ f \Wpref f'$. If, in addition, there exists at least one $s \in \states$ such that $f_{\pi_s} \succ f_{\pi'_s}$, then $f \succ f'$.
\end{axiom}

The second axiom is a `row separability condition'. It states that the ranking of opportunity profiles does not depend on types that have the same income distribution in both profiles. For any $f \in \SocietiesFixed$ and any $\pi \in \Delta(\incomes)$, let $(f_{-s},\pi) \in \SocietiesFixed$ be the opportunity profile such that $(f_{-s},\pi)(t) = \pi_t$ for all $t \neq s$ and $(f_{-s},\pi)(s) = \pi$.

\begin{axiom}[\textsc{Independence to common type-specific distributions}] \label{ax-ICCD}
For all $f,f'  \in \SocietiesFixed$, all $\pi, \pi' \in \Delta(\incomes)$ and all $s \in \states$,\quad $(f_{-s}, \textcolor{dark-blue}{\pi}) \Wpref\, (f'_{-s}, \textcolor{dark-blue}{\pi})\ \Rightarrow\ (f_{-s}, \textcolor{dark-green}{\pi'}) \Wpref\, (f'_{-s}, \textcolor{dark-green}{\pi'})$.
\end{axiom}
Consider the following example:
\[
\begin{pNiceMatrix}[first-row, first-col] 
               & \textnormal{\$1} & \textnormal{\$2} \\
\textnormal{R} & 0.2                        & 0.8 \\
\textnormal{P} & \textcolor{dark-blue}{0.4} & \textcolor{dark-blue}{0.6} \\
\end{pNiceMatrix}
\quad \Wpref\quad
\begin{pNiceMatrix}[first-row, first-col] 
               & \textnormal{\$1} & \textnormal{\$2} \\
\textnormal{R} & 0.5                        & 0.5 \\
\textnormal{P} & \textcolor{dark-blue}{0.4} & \textcolor{dark-blue}{0.6} \\
\end{pNiceMatrix}\quad
\Longrightarrow\quad
\begin{pNiceMatrix}[first-row, first-col] 
               & \textnormal{\$1} & \textnormal{\$2} \\
\textnormal{R} & 0.2                        & 0.8 \\
\textnormal{P} & \textcolor{dark-green}{0.3} & \textcolor{dark-green}{0.7} \\
\end{pNiceMatrix}
\quad \Wpref\quad
\begin{pNiceMatrix}[first-row, first-col] 
               & \textnormal{\$1} & \textnormal{\$2} \\
\textnormal{R} & 0.5                        & 0.5 \\
\textnormal{P} & \textcolor{dark-green}{0.3} & \textcolor{dark-green}{0.7} \\
\end{pNiceMatrix}\,.
\]

The third axiom specifies that type identities are irrelevant for the social planner's evaluations. Specifically, it states that permuting the rows of an opportunity profile does not affect social welfare. Let $\permutations$ be the set of all permutations $\sigma: \states \rightarrow \states$.
\begin{axiom}[\textsc{Anonymity}] \label{ax-anonymity}
For all $f, f' \in \SocietiesFixed$, if there exists $\sigma \in \permutations$ such that $\pi_{\sigma(s)} = \pi'_s$ for all $s \in \states$, then $f \sim f'$.
\end{axiom}
Consider the following example:
\begin{equation} \nonumber
f\ =\ 
\begin{pNiceMatrix}[first-row, first-col]
               & \textnormal{\$1} & \textnormal{\$2} \\
\textnormal{R} & 1         & 0         \\
\textnormal{P} & 0         & 1         \\
\end{pNiceMatrix}\quad
\sim\quad
\begin{pNiceMatrix}[first-row, first-col] 
               & \textnormal{\$1} & \textnormal{\$2} \\
\textnormal{R} & 0         & 1         \\
\textnormal{P} & 1         & 0         \\
\end{pNiceMatrix}\ =\ f'\,.
\end{equation}
This example illustrates that the anonymity axiom rules out some plausible egalitarian preferences. For example, one may argue that $f$ should be preferred to $f'$ because it reduces inequalities between dynasties by re-ranking the rich and the poor across generations. However, in this work, we evaluate individual opportunity sets rather than the fortunes of multigenerational dynasties. From this vantage point, type labels carry no moral weight beyond the opportunity set they are attached to.

The fourth axiom captures inequality-of-opportunity aversion. It states that any convergence of the income distributions of two types reduces inequality of opportunity.
\begin{axiom}[\textsc{Inequality of opportunity aversion}] \label{ax-aversion}
For all $f,f'\in \SocietiesFixed$, if there exist 
$s, s' \in \states$ such that $\pi_t = \pi'_t$ for all $t \neq s, s'$ and, furthermore, $\textcolor{bloud}{\pi_{s}} = \textcolor{bloud}{\pi_{s'}} = \left( \nicefrac{1}{2} \right) (\textcolor{dark-blue}{\pi'_s} + \textcolor{dark-green}{\pi'_{s'}})$, then $f \Wpref f'$.
\end{axiom}
%
Consider the following example:
\[
f\ =\ 
\begin{pNiceMatrix}[first-row, first-col]
               & \textnormal{\$1} & \textnormal{\$2} \\
\textnormal{R} & \textcolor{bloud}{0.3} & \textcolor{bloud}{0.7}  \\
\textnormal{P} & \textcolor{bloud}{0.3} & \textcolor{bloud}{0.7} \\
\end{pNiceMatrix}
\quad
\Wpref\quad
\begin{pNiceMatrix}[first-row, first-col]
               & \textnormal{\$1} & \textnormal{\$2} \\
\textnormal{R} & \textcolor{dark-blue}{0.2}  & \textcolor{dark-blue}{0.8}  \\
\textnormal{P} & \textcolor{dark-green}{0.4} & \textcolor{dark-green}{0.6} \\
\end{pNiceMatrix}\ =\ f'\,.
\]
This example also illustrates that our axiom on inequality-of-opportunity aversion has no immediate implications for overall inequality. To see this note that $\expect_q[f] = \expect_q[f'] = (0.3, 0.7)$, illustrating that inequality of opportunity is lower in $f$ than $f'$ while total inequality is the same.  

The fifth axiom is a `column separability condition'. It is analogous to the independence axiom in the expected utility model of \citet{Morgenstern1944}. Specifically, it states that mixing two equal-opportunity profiles with a common equal-opportunity profile does not alter the ranking of the first two profiles. 
\begin{axiom}[\textsc{Independence to common equal opportunities}] \label{ax-ICEO}
For all $\pi,\pi',\pi'' \in \Delta(\incomes)$ and all $\alpha \in (0,1)$,\quad $f_\pi \Wpref\, f_{\pi'}\ \Rightarrow\ \alpha f_\pi + (1-\alpha) {\textcolor{dark-blue}{f_{\pi''}}} \Wpref\, \alpha f_{\pi'} + (1-\alpha) {\textcolor{dark-blue}{f_{\pi''}}}$.
\end{axiom}
Consider the following example: 
\[
\begin{pNiceMatrix}
0.2 & 0.8\\
0.2 & 0.8\\
\end{pNiceMatrix}
\ \Wpref\
\begin{pNiceMatrix}
0.4 & 0.6\\
0.4 & 0.6\\
\end{pNiceMatrix}
\ \
\Longrightarrow\ \ \alpha 
\begin{pNiceMatrix}
0.2 & 0.8\\
0.2 & 0.8\\
\end{pNiceMatrix}
+(1-\alpha) 
\begin{pNiceMatrix}
\textcolor{dark-blue}{1} & \textcolor{dark-blue}{0} \\
\textcolor{dark-blue}{1} & \textcolor{dark-blue}{0} \\
\end{pNiceMatrix}
\ \Wpref\ \alpha
\begin{pNiceMatrix}
0.4 & 0.6\\
0.4 & 0.6\\
\end{pNiceMatrix}
+(1-\alpha) 
\begin{pNiceMatrix}
\textcolor{dark-blue}{1} & \textcolor{dark-blue}{0} \\
\textcolor{dark-blue}{1} & \textcolor{dark-blue}{0} \\
\end{pNiceMatrix}\,.
\]

%

The next axiom strengthens A\ref{ax-ICEO} by relaxing the requirement that the initial profiles must be equal-opportunity profiles. Because $\Wpref$ is assumed to be a continuous weak order, it implies A\ref{ax-ICEO}, but the converse does not hold.

\begin{axiomp}[\textsc{Strong independence to common equal opportunities}] \label{ax-ICEO2}
For all $f,f' \in \SocietiesFixed$, all $\pi, \pi' \in \Delta(\incomes)$ and all $\alpha \in (0,1)$,\quad $\alpha f + (1-\alpha) {\textcolor{dark-blue}{f_\pi}} \Wpref\, \alpha f' + (1-\alpha) {\textcolor{dark-blue}{f_\pi}} \Rightarrow \alpha f + (1-\alpha) {\textcolor{dark-green}{f_{\pi'}}} \Wpref\, \alpha f' + (1-\alpha) {\textcolor{dark-green}{f_{\pi'}}}$.
\end{axiomp}
Consider the following example:
\[
\begin{aligned}
& \alpha 
\begin{pNiceMatrix}
0.2 & 0.8\\
0.3 & 0.7 \\
\end{pNiceMatrix}
+(1-\alpha) 
    \begin{pNiceMatrix}
    \textcolor{dark-blue}{1} & \textcolor{dark-blue}{0} \\
    \textcolor{dark-blue}{1} & \textcolor{dark-blue}{0} \\
    \end{pNiceMatrix}
\ \Wpref\ \alpha
\begin{pNiceMatrix}
0.4 & 0.6\\
0.1 & 0.9 \\
\end{pNiceMatrix}
+(1-\alpha) 
    \begin{pNiceMatrix}
    \textcolor{dark-blue}{1} & \textcolor{dark-blue}{0} \\
    \textcolor{dark-blue}{1} & \textcolor{dark-blue}{0} \\
    \end{pNiceMatrix} \\
\Longrightarrow \quad & \alpha 
\begin{pNiceMatrix}
0.2 & 0.8\\
0.3 & 0.7 \\
\end{pNiceMatrix}
+(1-\alpha) 
    \begin{pNiceMatrix}
    \textcolor{dark-green}{0} & \textcolor{dark-green}{1} \\
    \textcolor{dark-green}{0} & \textcolor{dark-green}{1} \\
    \end{pNiceMatrix}
\ \Wpref\ \alpha
\begin{pNiceMatrix}
0.4 & 0.6\\
0.1 & 0.9 \\
\end{pNiceMatrix}
+(1-\alpha) 
    \begin{pNiceMatrix}
    \textcolor{dark-green}{0} & \textcolor{dark-green}{1} \\
    \textcolor{dark-green}{0} & \textcolor{dark-green}{1} \\
    \end{pNiceMatrix}\,.
\end{aligned}
\]
\section{Representation with a uniform type distribution}
\label{sec-representation-uniform}

We derive the main representation results in the baseline setting with fixed demographic type weights and relate them to the existing normative economics literature. All proofs of theorems and properties presented in the main text are provided in Appendix~\ref{app: proofs}. The connections to comparable results in decision theory under ambiguity are detailed in Appendix~\ref{app: risk}. 

\begin{theorem} \label{theo-general} 
Let $q(s) = 1/ |\states|$ for all $s \in \states$, and $|\states| \geq 3$.\footnote{Theorem~\ref{theo-general} and all other theorems in this section also work for $|\states| = 2$ when invoking an additional axiom called the `Thomsen condition' \cite[see][]{Wa88,De60}.} The social planner's preference relation $\Wpref$ satisfies A\ref{ax-Pareto}-A\ref{ax-ICEO} iff it can be represented by a function $W_{\phi, u, q} : \SocietiesFixed \rightarrow \real$, defined by:
    \begin{equation} \label{EDEU-general-representation}
        W_{\phi, u, q}(f) = \phi^{-1} \left(  \sum_{s \in \states} q(s) \phi \left( U(\pi_s) \right) \right)\,,
    \end{equation}
where $U(\pi_s) = \sum_{y \in \incomes} \pi_s(y) u(y)$ is the expected utility of type $s$ with a utility function $u : \incomes \rightarrow \real$, continuous and non-constant, and $\phi$ is a strictly increasing and weakly concave function. Moreover, two functions $W_{\phi, u, q}$ and $W_{\phi', u', q}$ represent the same preferences iff 
there exist $a, A > 0$ and $b, B \in \real$ such that $u' = a u + b$ and $\phi'(a v + b) = A \phi(v) + B$, for all $v$ in the range of $U$.
\end{theorem}
Theorem~\ref{theo-general} provides a social welfare function based on a two-stage evaluation: Expected utility is first computed within each type, and these type-specific welfare levels are then aggregated through an increasing and concave transformation capturing inequality-of-opportunity aversion. The first-stage utility function is left entirely unrestricted. Moreover, by treating the income distribution of their own type as the relevant income prospect of individuals, our framework naturally aligns with the ex-ante equality of opportunity literature \cite[][]{FP13}. 

The social welfare function presented in Theorem~\ref{theo-general} shares the ‘generalized utilitarian’ structure found in \cite{GKPS10}, while placing it on different normative foundations. \cite{GKPS10} operate within the Harsanyi tradition, focusing on an impartial observer’s preferences over hypothetical identity lotteries (‘accidents of birth’) versus objective risks (‘life chances’), while explicitly assuming that individuals are von Neumann-Morgenstern expected utility maximizers. In their framework, the observer’s preferences are defined over the entire product space of outcome lotteries and identity lotteries, and their identity-lottery independence axiom (Axiom 3, page 1947) requires the observer to evaluate probability mixtures of different identity distributions. By contrast, our framework abstracts from subjective individual preferences, and takes the perspective of an opportunity-egalitarian social planner who evaluates the objective opportunity structure of society.

Moreover, we fix the partition of types and their demographic weights $q \in \Delta(\states)$ at the outset of the axiomatic derivation. Because we do not require the planner to evaluate probability mixtures of different demographic distributions, our row separability condition (A\ref{ax-ICCD}) is structurally less demanding than the identity-lottery independence axiom used by \cite{GKPS10}. Our approach ensures that the initial axioms strictly capture the normative evaluation of conditional income distributions, unconfounded by variations in population structure. We address the distinct normative issue of comparing societies with different type partitions and demographic type weights separately in Section~\ref{sec-representation-extended}.

Along similar lines, our framework departs from the approach proposed by \cite{Fleurbaey10} for assessing risky social situations. While \cite{Fleurbaey10} introduces the elegant ‘expected equally distributed equivalent’ (EEDE) class of criteria, he assumes that individual ex-ante preferences are correctly measured by expected utility (as in \cite{GKPS10}, where the authors impose von Neumann-Morgenstern rationality on individuals from the outset). In contrast, our column separability condition (A\ref{ax-ICEO}) imposes independence directly on the social planner's preferences, but restricts it to the subset of equal-opportunity profiles. This ensures that the expected utility representation of a type emerges endogenously from the planner’s evaluation of aggregate social risk in the absence of inequality of opportunity, rather than being assumed as a behavioral premise.

Comparable structures for assessing inequality of opportunity have recently been explored by \cite{Moramarco2024} and \cite{AFGK26}. To cleanly isolate their respective normative criteria, these frameworks proceed by treating types symmetrically. For instance, \cite{Moramarco2024} propose a nested ‘between-group prioritarian’ social welfare function while assuming that types share the same size. Similarly, \cite{AFGK26} evaluate opportunity allocations using a uniform expected utility valuation, effectively weighting all types equally. By leaving the issue of arbitrary demographic type weights unaddressed in their core axiomatic characterizations, these approaches are limited in their direct applicability to societies with heterogeneous type distributions. 


Theorem~\ref{theo-general} provides a fairly general expression for the social welfare function in the double-expectation form. From a practical perspective, however, this expression is too general, as it assumes a rather arbitrary choice of functions $u$ and $\phi$. Replacing A\ref{ax-ICEO} by the stronger version A\ref{ax-ICEO2} restricts $\phi$ to a single-parameter function. 

\begin{theorem} \label{theo-main} Let $q(s) = 1/ |\states|$ for all $s \in \states$, and $|\states| \geq 3$. The social planner's preferences relation $\Wpref$ satisfies A\ref{ax-Pareto}-A\ref{ax-aversion} and A\ref{ax-ICEO2} iff it can be represented by a function $V_{\theta, u, q} \equiv W_{\phi_{\theta}, u, q}$, with $W_{\phi_\theta, u, q}$ as defined in Theorem~\ref{theo-general}, and with the following function $\phi_\theta$:
    \begin{equation} \label{expo-transfo} 
        \phi_\theta(t) = 
        \begin{cases} 
            t\,, & \text{for } \theta = 0\,, \\
            -\exp\left(-\theta t\right)\,, & \text{for } \theta > 0\,.
        \end{cases}
    \end{equation}
Moreover, two functions $V_{\theta, u, q}$ and $V_{\theta', u', q}$ represent the same preferences iff there exist $a > 0$ and $b \in \real$ such that $u' = a u + b$ and $\theta' = \theta / a$. \footnote{$\phi^{-1}_\theta(t) = - (\nicefrac{1}{\theta}) \ln (-t)$ if $\theta \neq 0$ and $\phi^{-1}_\theta(t) = t$ otherwise.} 
\end{theorem}
In Theorem~\ref{theo-main}, the function $\phi$ becomes translation-scale invariant and takes a single-parameter exponential form. The resulting structure is comparable to that identified by \cite{BD82}. However, their function $\phi_\theta$ depends on individual utilities, whereas in our case, it depends on the expected utilities of types.
We obtain the following representation:
\begin{equation} \label{V-second_order}
    V_{\theta, u, q}(f) = 
    \begin{cases} 
            \sum_{s \in \states} q(s) U(\pi_s)\,, & \text{for } \theta = 0\,, \\
            \left( - \nicefrac{1}{\theta} \right) \ln\left(  \sum_{s \in \states} q(s) \left(e^{- \theta U(\pi_s)} \right) \right)\,, & \text{for } \theta > 0\,.
    \end{cases}
\end{equation}

In this work, we abstract from individual preferences and focus exclusively on the social planner's normative judgments. As a result, the functions $u$ and $\phi$ are not completely independent. Theorem~\ref{theo-ln} in Appendix~\ref{app: log} shows that the usual additional requirements of monotonicity (A\ref{ax-mon}) and scale invariance (A\ref{ax-scale}) imply that the utility function $u$ must be logarithmic in the presence of inequality-of-opportunity aversion ($\theta>0$). These additional requirements are not particularly restrictive for applications that consider monetary outcomes. Therefore, we also use the representation of Theorem~\ref{theo-ln} when applying our framework to the intergenerational mobility estimates in Section \ref{sec-empirics}.


\section{Representation on the extended domain}
\label{sec-representation-extended}
We extend the framework to compare societies $(f,q) \in \SocietiesFixed \times \Delta(\states)$ with potentially different populations, type structures, and non-uniform demographic type weights~$q$. Hence, we lift the restriction that types are equally sized subgroups of the population, which may be restrictive for empirical applications.


Let $\labels$ be the (countably infinite) set of potential individuals. Any population is therefore a finite subset $\pop \subset \labels$. All possible sets of types can be written as $\types = \big\{ \states \mid \exists \pop \subset \labels\,:\, \states$ is a partition of $\pop \big\}$. Note that this setup encompasses both the case where types correspond to groups of individuals and the limiting case in which $\states$ uniquely identifies each individual. The set of all possible societies is denoted by:
\begin{equation}
     \domain = \bigcup_{\states \in \types} \left( \SocietiesFixed \times \Delta(\states) \right)\,.
\end{equation}

The social planner's preference $\Wpref^\star$ is now defined on the extended domain $\domain$, endowed with the disjoint union topology (where each component $\mathcal{F}_{\states} \times \Delta(\states)$ carries the product of their respective topologies), and $(f, q) \Wpref^\star (f', q')$ means that society $(f, q)$ is weakly preferred to society $(f', q')$. We denote by $\sim^*$ and $\succ^*$ its symmetric and asymmetric parts. As with the local preference $\succeq$ defined in Section~\ref{sec-preferences}, we assume that $\Wpref^\star$ is a continuous weak order on $\domain$, which directly implies continuity with respect to both opportunity profiles $f$ and demographic type weights $q$ within any fixed type partition $\states$. We introduce three additional axioms to characterize the social planner's preferences on the extended domain. To do so, we now formally treat the local preference introduced in Section~\ref{sec-preferences} as a family of relations $\{ \Wpref_\states \}_{\states \in \types}$, where each $\Wpref_\states$ evaluates profiles on a specific partition $\states$. 

The first of these axioms requires that, when two societies are based on the same type partition and have a uniform distribution across types, the preference on the extended domain coincides with the local preference. This axiom allows the locally defined axioms of Section~\ref{sec-preferences} to be imported into the general framework.
%
\begin{axiom}[\textsc{Consistency}] \label{ax-cons}
For all $f,f' \in \SocietiesFixed$ and all $\states \in \types$, if $q(s) = 1/|\states|$ for all $s \in \states$, then $(f,q) \Wpref^\star (f',q)\ \Leftrightarrow\ f \Wpref_\states f'$.
\end{axiom}

The second of these axioms is Dalton's principle of populations, which states that an identical replication of the population does not alter social welfare. Let $(f,q) \in \SocietiesFixed \times \Delta(\states)$ be a society defined on a finite population $\pop$, where $\states$ is a partition of $\pop$. We say that $(f',q') \in \mathcal F_{\states'} \times \Delta(\states')$ defined on a population $\pop'$ is a `k-fold replication' of $(f,q)$, with $k \geq 1$, if there exists a bijection $b : \states \rightarrow \states'$ such that for all $s \in \states$ we have $|b(s)| = k \cdot |s|$, $q'(b(s)) = q(s)$ and $f'(b(s)) = f(s)$. In other words, both societies have the same number of types, and each type in $\states$ can be uniquely mapped to a corresponding type in $\states'$ that contains exactly $k$ times as many individuals, while maintaining the exact same demographic weight and type-specific income distribution.

\begin{axiom}[\textsc{Principle of Populations}] \label{ax-pop}
For all $(f,q) \in \SocietiesFixed \times \Delta(\states)$ and all $(f',q') \in \mathcal F_{\states'} \times \Delta(\states')$, if $(f',q')$ is a k-fold replication of $(f,q)$, then $(f,q) \sim^\star (f',q')$.
\end{axiom}

The third of these axioms states that splitting a type into two subgroups with the same income distribution as the original type leaves social welfare unchanged.
%
\begin{axiom}[\textsc{Invariance to Type-Split}] \label{ax-split}
For all $(f,q) \in \SocietiesFixed \times \Delta(\states)$ and all $(f',q') \in \mathcal F_{\states'} \times \Delta(\states')$, if the following three conditions hold, then $(f,q) \sim^\star (f',q')$:
\begin{itemize}
    \item[(i)] $\states$ is obtained from $\states'$ by replacing one type $s_{ab} \in \states'$ with two distinct types $s_a, s_b \notin \states'$,
    \item[(ii)] For all types $s \in \states' \setminus \{s_{ab}\}$ we have $q(s) = q'(s)$, and $\pi_s = \pi'_s$\ ,
    \item[(iii)] $\textcolor{bloud}{q(s_{a})} + \textcolor{orange}{q(s_{b})} = \textcolor{violet}{q'(s_{ab})} > 0$, and $\textcolor{dark-green}{\pi_{s_{a}}} = \textcolor{dark-green}{\pi_{s_{b}}} = \textcolor{dark-blue}{\pi'_{s_{ab}}}$\ .
\end{itemize}
\end{axiom}
Consider the following example:
\[
(f, q)\ =\ 
\begin{pNiceMatrix}[first-row, first-col]
      & \textnormal{\$1} & \textnormal{\$2} \\
      & 0                    & 1                    \\
s_{a} & \textcolor{dark-green}{1} & \textcolor{dark-green}{0} \\
s_{b} & \textcolor{dark-green}{1} & \textcolor{dark-green}{0} \\
\end{pNiceMatrix}\ ,\
\begin{pNiceMatrix}
0.2 \\
\textcolor{bloud}{0.5} \\
\textcolor{orange}{0.3} \\
\end{pNiceMatrix}
\quad\ \sim^\star\quad\
\begin{pNiceMatrix}[first-row, first-col]
       & \textnormal{\$1}     & \textnormal{\$2}     \\
       & 0                        & 1                        \\
s_{ab} & \textcolor{dark-blue}{1} & \textcolor{dark-blue}{0} \\
\end{pNiceMatrix}\ ,\
\begin{pNiceMatrix}
0.2 \\
\textcolor{violet}{0.8} \\
\end{pNiceMatrix}
= (f', q')\ .
\]

These three axioms suffice to establish that the representation of local preferences over opportunity profiles $f \in \SocietiesFixed$ obtained in Section~\ref{sec-representation-uniform} can also represent preferences in the extended domain over societies $(f,q) \in \domain$. We add these axioms to those used in Theorem~\ref{theo-general} to obtain our next result.\footnote{Note that while the local representation in Theorem~\ref{theo-general} strictly requires at least three essential types to ensure additive separability, Theorem~\ref{theo-extended} applies to all societies in $\domain$, including those with only one or two types. This is a direct consequence of  A\ref{ax-split}: Any society with fewer than three types can be viewed as ordinally equivalent to a society with three or more identical sub-types, thereby inheriting the additive structure without additional axioms.}

 
\begin{theorem} \label{theo-extended} Assume that for each $\states \in \types$, the local preference $\Wpref_{\states}$ over uniform partitions satisfies A\ref{ax-Pareto}-A\ref{ax-ICEO}. The social planner's preferences relation~$\Wpref^\star$ satisfies A\ref{ax-cons}-A\ref{ax-split} iff, for all $(f,q) \in \domain$, it can be represented by the function $W_{\phi, u}(f,q) = \phi^{-1} \left(  \sum_{s \in \states} q(s) \phi \left( U(\pi_s) \right) \right)$, with $\phi$ and $u$ as characterized in Theorem~\ref{theo-general}.
\end{theorem}

Note that this result is based on Theorem~\ref{theo-general}, which is the most general. However, it also applies to the social welfare functions obtained in Theorem~\ref{theo-main} and Theorem~\ref{theo-ln} (Appendix~\ref{app: log}). We denote the social welfare functions characterized in Theorem~\ref{theo-main}, but applied to the extended domain, by $V_{\theta, u}(f,q)$.

\section{Equivalent representations and stochastic dominance}
\label{sec-representation-equivalent}

\subsection{Equivalent representations}
\label{subsec-equivalent}

We introduce two equivalent representations and an approximation of the social welfare function $V_{\theta, u}(f,q)$, to illustrate different perspectives on how aversion to inequality of opportunity is captured in our framework. To establish the first equivalent representation, we define Kullback-Leibler divergences.
\begin{definition} \label{def-KL}
    Consider two distributions $p, q \in \Delta(\states)$, with $p$ absolutely continuous to~$q$. The set of such distributions $p$ is denoted $\Delta^q(\states)$. A Kullback-Leibler divergence $D_{KL}$ is defined by:~\footnote{By absolute continuity, $q(s)=0$ implies $p(s)=0$. Moreover, in empirical analyses, it is generally the case that $q(s) > 0$ for all $s \in\states$. Because $\lim_{t \rightarrow 0^+} t \ln t = 0$, the contribution of state $s \in \states$ to the divergence $D_{KL}(p || q)$ is zero if $p(s)=0$.}
\[
D_{KL}(p || q) = \sum_{s \in \states} p(s) \ln{\frac{p(s)}{q(s)}}\,.
\]
\end{definition}
Kullback-Leibler divergences are a special case of a more general class called Bregman divergences \cite[see][and Definition~\ref{def-Bregman} below]{MN11}.\footnote{\cite{HKP22} apply them as an index of inequality of opportunity.} An important property of Kullback-Leibler divergences is that $D_{KL}(p || q) \geq 0$, with equality if and only if $p=q$. By letting $U = {(U(\pi_s))}_{s \in \states}$ be the list of expected utilities of all types in $f \in \SocietiesFixed$ distributed according to $q = {(q(s))}_{s \in \states}$, and recalling that $\expect\!_q [U] = \sum_{s \in \states} q(s) U(\pi_s)$, we obtain the following result.

\begin{theorem} \label{theo-main2} An equivalent representation of the social welfare function $V_{\theta, u}(f,q)$ is:
    \begin{equation} \label{HS-representation}
        V_{\theta, u}(f,q) =  
        \begin{cases} 
            \expect\!_q [U]\,, & \text{for } \theta = 0\,, \\
            \MIN_{p \in \Delta^q(\states)} \left[ \expect\!_p [U] + \frac{1}{\theta} D_{KL}(p || q) \right]\,, & \text{for } \theta > 0\,,
        \end{cases}
    \end{equation}
where $p \in \Delta^q(\states)$. The optimal weights $p^\star_{\theta, u}(f,q;\cdot)$ that solve this minimization problem are such that, for all $s \in \states$: 
    \begin{equation} \label{opt-p}
        p^\star_{\theta, u}(f,q;s) =\frac{q(s) e^{-\theta U(\pi_s)}}{\sum_{s \in \states} q(s) \left(e^{- \theta U(\pi_s)} \right)}\,. 
    \end{equation}
\end{theorem}

While $q(s)$ is the demographic weight of type $s$ in the population, $p(s)$ can be interpreted as a normative weight applied to type $s$ by the social planner. Equation \eqref{opt-p} furthermore shows that optimal type weights $p^\star_{\theta, u}(f,q;s)$ are negatively related to types' expected utilities $U(\pi_s)$ and that the strength of this negative relationship is governed by the inequality-of-opportunity aversion parameter $\theta$. Hence, while the representation in Theorem~\ref{theo-main} captures inequality-of-opportunity aversion by a concave transformation of types' expected utilities $U(\pi_s)$, the dual representation of Theorem~\ref {theo-main2} captures inequality-of-opportunity aversion through a multiplicative weighting of $U(\pi_s)$ with normative weights $p(s)$.

To establish the following equivalent representation of $V_{\theta, u}(f,q)$, we define Bregman divergences.
\begin{definition} \label{def-Bregman}
    Let $\varphi : \real \rightarrow \real$ be a differentiable and strictly convex function, with derivative $\varphi^{(1)}$. The Bregman divergence with generator $\varphi$, denoted $D_\varphi$, between two real numbers $x$ and $y$, is defined by $D_\varphi(x||y) = \varphi(x) - \varphi(y) -(x-y)\varphi^{(1)}(y)$.
\end{definition}
A Bregman divergence is a measure of `distance' between two quantities, called a divergence because it does not satisfy the triangular inequality. Like the Kullback-Leibler divergence, a Bregman divergence satisfies $D_\varphi(x||y) \geq 0$ for all $x,y \in \real$, and $D_\varphi(x||y)=0$ if and only if $x=y$. We have:\footnote{Note that if type-specific expected utilities are constant in Theorem~\ref{theo-main3}, the function $K_U$ becomes affine and the divergence trivially collapses to zero ($D_{K_U} = 0$). This reflects the economic intuition: The penalty vanishes whenever types face equivalent expected welfare. Otherwise, when $U$ is non-degenerate, $K_U$ is strictly convex and generates a proper Bregman divergence.}

\begin{theorem} \label{theo-main3} An equivalent representation of the social welfare function $V_{\theta, u}(f,q)$ is:
\begin{equation} \label{eq-Bregman-representation}
        V_{\theta, u}(f,q) =  
        \begin{cases} 
            \expect\!_q [U]\,, & \text{for } \theta = 0\,, \\
            \underbrace{\vphantom{\mathbb{E}_q\!\left[e^{-\theta U(\pi_s)}\right]}\mathbb{E}_q[U]}_{\text{Efficiency}} - \underbrace{\nicefrac{1}{\theta} D_{K_U}(-\theta||0)}_{\text{Inequality of opportunity}}\,, & \text{for } \theta > 0\,,
        \end{cases}
    \end{equation}
with $K_U(\tau) = \ln \left( \expect_q [e^{\tau U}] \right)$ for any $\tau \in \real$. Moreover, when $\theta > 0$ is close to zero, or when the dispersion of expected utilities across types is small, we have $V_{\theta, u}(f,q) \approx \expect\!_q [U] - \frac{\theta}{2} \var_q(U)$.
\end{theorem}

Theorem~\ref{theo-main3} provides a closed-form representation, which we call the `mean-divergence' criterion. In this representation, disparities in type-specific expected utilities are exactly captured by the Bregman divergence $D_{K_U}$. Given the well-established statistical properties of Bregman divergences \cite[][]{MN11}, this representation makes the planner's trade-off between efficiency and inequality of opportunity explicit. 

Theorem~\ref{theo-main3} also shows that $V_{\theta, u}(f,q)$ can be approximated by a `mean-variance' criterion. In this representation, the weight assigned to disparities across types, as measured by the variance $\var_q(U)$, is directly proportional to the planner’s degree of inequality-of-opportunity aversion. In this representation, the inequality-of-opportunity aversion parameter $\theta$ plays a role analogous to the Arrow-Pratt coefficient of absolute risk aversion in the decision theory literature. The closed form of the mean-divergence criterion based on the Bregman divergence goes beyond the local mean-variance approximation by accounting for all moments of the distribution of types' expected utilities.\footnote{The following expression highlights the analogy between the Bregman divergence and the variance.  For any $\theta > 0$, we have: $D_{K_U}(-\theta||0) = \ln \mathbb{E}_q\!\left[e^{-\theta U}\right] - \mathbb{E}_q[\ln e^{-\theta U}]$. Hence, the Bregman divergence uses the logarithm applied to an exponential transformation of utility, whereas the variance uses the square function.} 

\subsection{Stochastic dominance}
\label{subsec-stoch_dom}

We establish a stochastic dominance criterion that allows for unambiguous ranking of societies for a wide class of social planners' preferences over inequality of opportunity. For any given value of the inequality-of-opportunity aversion parameter $\theta$, two societies can always be ranked; however, the ranking may change as this parameter varies. Therefore, we are interested in the preorder relation which corresponds to the dominance of a society over another, for all admissible values of $\theta$. To establish our stochastic dominance result, the following definition is needed.

\begin{definition} \label{def-CM}
    A function $\varphi : (0,\infty) \rightarrow \real$ is said completely monotone if all the derivatives $\varphi^{(n)}$ of $\varphi$ exist, and they satisfy $(-1)^{(n)} \varphi^{(n)}(x) \geq 0$ for all integers $n \geq 0$, with $\varphi^{(0)} = \varphi$. A function  $\varphi : (0,\infty) \rightarrow \real$ such that $\varphi^{(1)}$ is completely monotone is said completely alternating.\footnote{Adding the restriction $\varphi \geq 0$, it is called a Bernstein function.}
\end{definition}
The notion of completely alternating functions has been investigated in the context of stochastic dominance for risky decisions \cite[see][]{Wh89}. In this literature, the utility function of an expected utility maximizer is typically assumed to be increasing and concave (positive first derivative and negative second derivative, respectively). The commonly used CARA and CRRA utility functions have these properties and, more generally, are completely alternating (positive third derivative, negative fourth derivative, and so on). 

For the social welfare function $V_{\theta, u}(f,q)$, the social planner transforms the expected utility $U(\pi_s)$ of each type $s \in \states$ using the function $\phi_\theta$ defined in Equation~(\ref{expo-transfo}), which is a particular completely alternating function. This raises the question of whether $\phi_\theta$ can be replaced with alternative completely alternating functions. In Definition~\ref{def-CM}, we restrict the domain of these transformations to $(0,\infty)$, which implicitly assumes that all expected utilities are strictly positive. This restriction is made without loss of generality, because the underlying utility function $u$ is defined up to an increasing affine transformation: As the income domain $\incomes$ is compact, any continuous utility function $u$ attains a finite global minimum on $\incomes$, allowing a shift of $u$ to ensure strict positivity.

\begin{theorem} \label{theo-stoch_dom} Let $(f,q), (f',q') \in \domain$ be two societies defined respectively on type partitions $\states$ and $\states'$. Assuming, without loss of generality, that the utility function $u$ takes strictly positive values, the following two statements are equivalent: 
\begin{itemize}
    \item[(i)] $\sum_{s \in \states} q(s) \varphi \left( U(\pi_s) \right) \geq \sum_{s \in \states'} q'(s) \varphi \left( U(\pi'_s) \right)$, for all $\varphi : (0,\infty) \rightarrow \real$ completely alternating\,,
    \item[(ii)] $V_{\theta, u}(f,q) \geq V_{\theta, u}(f',q')$, for all $\theta \geq 0$\,.
\end{itemize}
\end{theorem}
%

Theorem~\ref{theo-stoch_dom} establishes that the stochastic dominance of society $(f,q)$ over society $(f',q')$ according to $V_{\theta, u}$ with a fixed first-stage utility function $u$, extends to a broader class of social welfare functions than the one-parameter exponential family defined in Equation~(\ref{expo-transfo}).\footnote{In stochastic dominance theory, preorder (i) is called a ‘completion’ of preorder (ii).} Specifically, it implies dominance for any evaluation where expected utilities are transformed by a completely alternating function. This includes standard transformations such as CARA and CRRA (restricted to strictly positive aversion parameters), as well as the logarithmic function. While completely alternating functions form a proper subclass of the general increasing and weakly concave transformations characterized in Theorem~\ref{theo-general}, they still cover a wide range of normative judgments on how to incorporate inequality-of-opportunity aversion when comparing societies. Moreover, when adopting logarithmic utility (Theorem~\ref{theo-ln}), Statement (ii) of Theorem~\ref{theo-stoch_dom} depends on a single parameter, allowing for easy implementation in empirical applications.


\section{Properties}
\label{sec-properties}

We present the main properties satisfied by the social welfare function $V_{\theta, u}(f,q)$. 
%
\begin{property} \label{prop-V_positive}
If $u(y) \geq 0$ for all $y \in \incomes$, then $V_{\theta, u}(f,q) \geq 0$.
\end{property}
\noindent Property~\ref{prop-V_positive} shows that the non-negativity of the social welfare function follows from the simple requirement that $u$ be positive. Because $u$ is defined up to an increasing affine transformation, this requirement entails no loss of generality, as discussed in the previous section. 
\begin{property} \label{prop-derivative_EU}
$\frac{\partial V_{\theta, u}}{\partial U(\pi_s)}(f,q) = p^\star_{\theta, u}(f,q;s)$.
\end{property}
\noindent Property~\ref{prop-derivative_EU} establishes that the marginal effect of the expected utility of type $s \in \states$ equals $p^\star_{\theta, u}(f,q;s)$, i.e., the normative weight that the social planner assigns to that type. If the social planner is inequality-of-opportunity neutral ($\theta = 0$), then $p^\star_{\theta, u}(f,q;s) = q(s)$, meaning that types' expected utilities are weighted according to their demographic weights.

\begin{property} \label{prop-derivative_u}
$\frac{\partial V_{\theta, u}}{\partial u(y)}(f,q) = \pi^\star_{\theta, u}(f,q;y)$, where $\pi^\star_{\theta, u}(f,q;y) = \sum_{s \in \states} p^\star_{\theta, u}(f,q;s) \pi_s(y)$.
\end{property}
\noindent Property~\ref{prop-derivative_u} identifies the impact of the utility associated with income $y \in \incomes$ on social welfare. In the overall population, $\pi(y) = \sum_{s \in \states} q(s) \pi_s(y)$ represents the proportion of individuals earning income $y$. If the social planner is inequality-of-opportunity neutral ($\theta = 0$), this is exactly the marginal effect of $u(y)$ on social welfare, since $p^\star_{\theta, u}(f,q;s) = q(s)$. Otherwise, if the planner is inequality-of-opportunity averse ($\theta > 0$), the demographic weights $q(s)$ are replaced by the normative weights $p^\star_{\theta, u}(f,q;s)$, which transforms $\pi(y)$ into the weighted average $\pi^\star_{\theta, u}(f,q;y)$. 

\begin{property} \label{prop-extreme}
$V_{0, u}(f,q) = \sum_{s \in \states} q(s) U(\pi_s) = \sum_{y \in \incomes} \pi(y) u(y)$, and\\ $V_{\infty, u}(f,q) = \min_{s \in \states : q(s)>0} U(\pi_s)$.
\end{property}
\noindent Property~\ref{prop-extreme} illustrates the role of the inequality-of-opportunity aversion parameter through its polar cases. At $\theta=0$, the criterion collapses to standard utilitarianism: Social welfare reduces to the expected utility of the aggregate income distribution. Conversely, as $\theta \rightarrow \infty$, the planner adopts a strict Maximin perspective, focusing solely on the worst-off type. Assuming an affine~$u$, the function $V_{\infty, u}(f,q)$ recovers the opportunity-sensitive measure of \cite{vdG93}, namely $\min_{s \in \states} \mu(\pi_s)$.

The following result further clarifies the role of $\theta$ in our framework beyond the polar cases discussed above. For a given society $(f,q) \in \domain$, we simplify notation by writing $p^\star_{\theta,u}(f,q;s) = p^\star_{\theta,u}(s)$ in what follows.
\begin{theorem} \label{theo-comparative} Let $\theta, \theta' \geq 0$. The following three statements are equivalent:
\begin{itemize}
    \item[(i)] $\textcolor{dark-blue}{\theta} < \textcolor{dark-green}{\theta'}$,
    \item[(ii)] $V_{\textcolor{dark-blue}{\theta}, u}(f,q) > V_{\textcolor{dark-green}{\theta'}, u}(f,q)$, for all societies $(f,q) \in \domain$ exhibiting inequality of opportunity, i.e., where there exist $s, s' \in \states$ with $q(s), q(s') >0$ and $U(\pi_s) \neq U(\pi_{s'})$,
    \item[(iii)] For any two-type society with $q = (\nicefrac{1}{2},\nicefrac{1}{2})$ and an opportunity profile $f = (\pi_{s},\pi_{s'}) \in \SocietiesFixed$ such that $U(\pi_s) < U(\pi_{s'})$, we have $p^\star_{\textcolor{dark-green}{\theta'},u}(s) > p^\star_{\textcolor{dark-blue}{\theta},u}(s) \geq p^\star_{\textcolor{dark-blue}{\theta},u}(s') > p^\star_{\textcolor{dark-green}{\theta'},u}(s')$.
\end{itemize}
\end{theorem}
Theorem~\ref{theo-comparative} establishes that the higher the inequality-of-opportunity aversion parameter $\theta$, the lower the social welfare. This reflects a growing concern about disparities in expected utilities between types. Statement (iii) of the theorem concerns a simple society profile consisting of two types, $s$ and $s'$, represented with equal population shares. Recalling that $p^\star_{\theta,u}(s)$ measures the marginal impact of the expected utility of type $s$ (Property~\ref{prop-derivative_EU}), if $U(\pi_s) < U(\pi_{s'})$, then $p^\star_{\theta,u}(s) \geq p^\star_{\theta,u}(s')$, with equality only if $\theta=0$. If inequality-of-opportunity aversion increases from $\theta$ to $\theta'$, then the gap between normative weights increases. This result complements Property~\ref{prop-extreme}, which states that, if $\theta$ tends to infinity in such a society with two types, then $p^\star_{\theta,u}(s)$ tends to 1 and $p^\star_{\theta,u}(s')$ tends to 0.

Figure \ref{fig: theosum} summarizes the properties of our framework in graphical form using the mean-divergence representation from Theorem \ref{theo-main3}.\footnote{For the sake of illustration, we assume uniform demographic type weights $q(s)$, impose logarithmic utility, and order types $s$ by expected type utilities $U(\pi_s)$.} Panels (a)--(b) illustrate changes in $V_{\theta, u}(f,q)$ resulting from shifts in the normative parameter $\theta$. Panel (a) shows a baseline case without inequality-of-opportunity aversion: if $\theta=0$, normative weights $p^\star_{\theta,u}(s)$ correspond to demographic weights $q(s)$, and social welfare $V_{\theta, u}(f,q)$ is equal to the efficiency component $\mathbb{E}_q[U]$, i.e., the average of expected type utilities. Panel (b) introduces inequality of opportunity aversion: if $\theta>0$, normative weights $p^\star_{\theta,u}(s)$ diverge from demographic weights $q(s)$, with higher weights placed on types with lower expected type utilities. As a result, $V_{\theta, u}(f,q)$ is lower than $\mathbb{E}_q[U]$ with the difference equal to the inequality of opportunity component $\nicefrac{1}{\theta} D_{K_U}$, i.e., the Bregman divergence of expected type utilities.

Panels (c)--(d) illustrate changes in $V_{\theta, u}(f,q)$ resulting from shifts in empirical distributions. Panel (c) shows the effect of a parallel shift in type outcomes: for a given $\theta>0$, social welfare $V_{\theta, u}(f,q)$ increases, whereas inequality of opportunity $\nicefrac{1}{\theta} D_{K_U}$ remains unaffected. Panel (d) illustrates a mean-preserving spread in type outcomes: for a given $\theta>0$, social welfare $V_{\theta, u}(f,q)$ decreases due to an increase in inequality of opportunity $\nicefrac{1}{\theta} D_{K_U}$, whereas efficiency $\mathbb{E}_q[U]$ remains unaffected.

\begin{figure}[ht]
\caption{Comparative statics of $V_{\theta, u}(f,q)$}
\vspace{0.2cm} {\centering
\resizebox{0.9\textwidth}{!}{\begin{tikzpicture}[
  >=Stealth,
  font=\small,
]

\definecolor{cBase}{rgb}{0.35, 0.48, 0.62}      
\colorlet{cBaseL}{cBase!40!white}
\definecolor{cShift}{rgb}{0.82, 0.68, 0.22}      
\colorlet{cShiftL}{cShift!40!white}
\definecolor{cSpread}{rgb}{0.55, 0.55, 0.55}     
\colorlet{cSpreadL}{cSpread!40!white}

\def\W{5.2}
\def\H{4.8}
\def\HG{2.8}
\def\VG{2.2}

\def\MeanBase{2.60}
\def\MeanShift{3.40}   
\def\VBase{1.10}       
\def\VShift{1.90}      
\def\VSpread{0.55}     

\def\dotOuterR{6.0pt}
%

\def\dotCoords{0.5/0.75, 1.2/1.30, 2.2/2.50, 3.2/3.40, 4.5/4.05}

\newcommand{\yaxislabel}[3]{%
  \draw[thin, #2] (-0.12,#1) -- (0.04,#1);
  \node[left, font=\scriptsize, #2] at (-0.12,#1) {#3};
}
\newcommand{\yaxislabelghost}[2]{%
  \draw[thin, #1, opacity=0.40] (-0.12,#2) -- (0.04,#2);
}

\begin{scope}[shift={(0,0)}]
  \draw[->] (0,0) -- (\W+0.5,0) node[below] {$s$};
  \draw[->] (0,0) -- (0,\H+0.4) node[left] {$U(\pi_s)$};

  \draw[cBaseL, solid, semithick] (0,\MeanBase) -- (\W,\MeanBase);

  \draw[cBaseL, thick] plot[smooth, tension=0.55] coordinates {
    (0.5,0.75)(1.2,1.30)(2.2,2.50)(3.2,3.40)(4.5,4.05)};

  \foreach \px/\py in \dotCoords {
    \fill[cBaseL] (\px,\py) circle (\dotOuterR);
  }

  \draw[thin, cBaseL] (-0.12,\MeanBase) -- (0.04,\MeanBase);
  \node[left, font=\scriptsize, cBaseL, align=center] at (-0.12,\MeanBase)
    {$\mathbb{E}_q[U]$\\$=$\\$V_{\theta, u}(f,q)$};

  \node[font=\small\bfseries] at (\W/2, \H+1.1) {(a) $\theta = 0$};
\end{scope}

\begin{scope}[shift={(\W+\HG,0)}]
  \draw[->] (0,0) -- (\W+0.5,0) node[below] {$s$};
  \draw[->] (0,0) -- (0,\H+0.4) node[left] {$U(\pi_s)$};

  \draw[cBaseL, solid, semithick] (0,\MeanBase) -- (\W,\MeanBase);
  \draw[cBase,  dashed, semithick] (0,\VBase)    -- (\W,\VBase);

  \draw[cBase, thick] plot[smooth, tension=0.55] coordinates {
    (0.5,0.75)(1.2,1.30)(2.2,2.50)(3.2,3.40)(4.5,4.05)};

  \foreach \px/\py in \dotCoords {
    \fill[cBaseL] (\px,\py) circle (\dotOuterR);
  }
  \foreach \px/\py/\r in {0.5/0.75/9.8, 1.2/1.30/7.5, 2.2/2.50/4.1, 3.2/3.40/2.6, 4.5/4.05/1.9}{
    \fill[cBase] (\px,\py) circle (\r pt);
  }

  \yaxislabel{\MeanBase}{cBaseL}{$\mathbb{E}_q[U]$}
  \yaxislabel{\VBase}{cBase}{$V_{\theta, u}(f,q)$}

  \draw[<->, cBase, thick] (4.0,\VBase+0.06) -- (4.0,\MeanBase-0.06);
  \node[right, font=\tiny, cBase, fill=white, inner sep=2pt] at (4.05,1.85)
    {$\nicefrac{1}{\theta}D_{K_U}$};

  \node[font=\small\bfseries] at (\W/2, \H+1.1) {(b) $\theta > 0$};
\end{scope}

\begin{scope}[shift={(0,-\H-\VG)}]
  \draw[->] (0,0) -- (\W+0.5,0) node[below] {$s$};
  \draw[->] (0,0) -- (0,\H+0.4) node[left] {$U(\pi_s)$};

  \draw[cBaseL, solid, semithick] (0,\MeanBase) -- (\W,\MeanBase);
  \draw[cBase,  dashed, semithick] (0,\VBase)    -- (\W,\VBase);
  \draw[cShiftL, solid, semithick] (0,\MeanShift) -- (\W,\MeanShift);
  \draw[cShift,  dashed, semithick] (0,\VShift)    -- (\W,\VShift);

  \draw[cBase, thick, opacity=0.35] plot[smooth, tension=0.55] coordinates {
    (0.5,0.75)(1.2,1.30)(2.2,2.50)(3.2,3.40)(4.5,4.05)};
  \draw[cShift, thick] plot[smooth, tension=0.55] coordinates {
    (0.5,1.55)(1.2,2.10)(2.2,3.30)(3.2,4.20)(4.5,4.85)};

  \foreach \px/\py in \dotCoords {
    \fill[cBaseL] (\px,\py) circle (\dotOuterR);
  }
  \foreach \px/\py/\r in {0.5/0.75/9.8, 1.2/1.30/7.5, 2.2/2.50/4.1, 3.2/3.40/2.6, 4.5/4.05/1.9}{
    \fill[cBase] (\px,\py) circle (\r pt);
  }
  \foreach \px/\py in {0.5/1.55, 1.2/2.10, 2.2/3.30, 3.2/4.20, 4.5/4.85}{
    \fill[cShiftL] (\px,\py) circle (\dotOuterR);
  }
  \foreach \px/\py/\r in {0.5/1.55/9.8, 1.2/2.10/7.5, 2.2/3.30/4.1, 3.2/4.20/2.6, 4.5/4.85/1.9}{
    \fill[cShift] (\px,\py) circle (\r pt);
  }

  \yaxislabel{\MeanBase}{cBaseL}{$\mathbb{E}_q[U]$}
  \yaxislabel{\VBase}{cBase}{$V_{\theta, u}(f,q)$}
  \yaxislabel{\MeanShift}{cShiftL}{$\mathbb{E}_q[U]'$}
  \yaxislabel{\VShift}{cShift}{$V_{\theta, u}'(f,q)$}

  \draw[<->, cBase, thick] (4.1,\VBase+0.06) -- (4.1,\MeanBase-0.06);
  \node[left, font=\tiny, cBase, fill=white, inner sep=2pt] at (4.05,1.85)
    {$\nicefrac{1}{\theta}D_{K_U}$};

  \draw[<->, cShift, thick] (4.6,\VShift+0.06) -- (4.6,\MeanShift-0.06);
  \node[right, font=\tiny, cShift, fill=white, inner sep=2pt] at (4.65,2.65)
    {$\nicefrac{1}{\theta}D_{K_U}$};

  \node[font=\small\bfseries] at (\W/2, \H+1.1) {(c) Parallel shift};
\end{scope}

\begin{scope}[shift={(\W+\HG,-\H-\VG)}]
  \draw[->] (0,0) -- (\W+0.5,0) node[below] {$s$};
  \draw[->] (0,0) -- (0,\H+0.4) node[left] {$U(\pi_s)$};

  \draw[cBaseL, solid, semithick] (0,\MeanBase) -- (\W,\MeanBase);
  \draw[cBase,  dashed, opacity=0.45, semithick] (0,\VBase) -- (\W,\VBase);
  \draw[cSpread, dashed, semithick] (0,\VSpread) -- (\W,\VSpread);

  \draw[cBase, thick, opacity=0.35] plot[smooth, tension=0.55] coordinates {
    (0.5,0.75)(1.2,1.30)(2.2,2.50)(3.2,3.40)(4.5,4.05)};
  \draw[cSpread, thick] plot[smooth, tension=0.55] coordinates {
    (0.5,0.50)(1.2,1.10)(2.2,2.50)(3.2,3.60)(4.5,4.30)};

  \foreach \px/\py in \dotCoords {
    \fill[cBaseL] (\px,\py) circle (\dotOuterR);
  }
  \foreach \px/\py/\r in {0.5/0.75/9.8, 1.2/1.30/7.5, 2.2/2.50/4.1, 3.2/3.40/2.6, 4.5/4.05/1.9}{
    \fill[cBase] (\px,\py) circle (\r pt);
  }
  \foreach \px/\py in {0.5/0.50, 1.2/1.10, 2.2/2.50, 3.2/3.60, 4.5/4.30}{
    \fill[cSpreadL] (\px,\py) circle (\dotOuterR);
  }
  \foreach \px/\py/\r in {0.5/0.50/10.1, 1.2/1.10/7.5, 2.2/2.50/3.7, 3.2/3.60/2.1, 4.5/4.30/1.5}{
    \fill[cSpread] (\px,\py) circle (\r pt);
  }

  \draw[thin, cBaseL] (-0.12,\MeanBase) -- (0.04,\MeanBase);
  \node[left, font=\scriptsize, align=center] at (-0.12,\MeanBase)
    {{\color{cBaseL}$\mathbb{E}_q[U]$}\\$=$\\{\color{cSpread}$\mathbb{E}_q[U]'$}};
  \yaxislabel{\VBase}{cBase}{$V_{\theta, u}(f,q)$}
  \yaxislabel{\VSpread}{cSpread}{$V_{\theta, u}'(f,q)$}

  \draw[<->, cBase, thick] (4.1,\VBase+0.06) -- (4.1,\MeanBase-0.06);
  \node[left, font=\tiny, cBase, fill=white, inner sep=2pt] at (4.05,1.85)
    {$\nicefrac{1}{\theta}D_{K_U}$};

  \draw[<->, cSpread, thick] (4.6,\VSpread+0.06) -- (4.6,\MeanBase-0.06);
  \node[right, font=\tiny, cSpread, fill=white, inner sep=2pt] at (4.65,1.575)
    {$\nicefrac{1}{\theta}D'_{K_U}$};

  \node[font=\small\bfseries] at (\W/2, \H+1.1)
    {(d) Mean-preserving spread};
\end{scope}

\end{tikzpicture}}}
\begin{spacing}{0.9}
{\footnotesize
\textbf{Note:} This figure illustrates the properties of our measures based on the mean-divergence representation from Theorem \ref{theo-main3}: Social welfare $V_{\theta, u}(f,q)$ decomposes into efficiency, i.e., expected type utilities $\mathbb{E}_q[U]$ and inequality of opportunity, i.e., the $\theta$-weighted Bregman divergence in expected type utilities $\nicefrac{1}{\theta} D_{K_U}$.}
\end{spacing}
\label{fig: theosum}
\end{figure}

%
%
%

\section{Opportunity-sensitive social welfare in the United States}
\label{sec-empirics}

In this section, we take our axiomatic framework to the data and evaluate opportunity-sensitive social welfare in the United States. First, we assess national differences across time by comparing birth cohorts from the 1940s to the 1980s. Second, we examine regional differences for cohorts born around 1980 by comparing commuting zones (CZ). In both applications, we define types $s\in \states$ based on measures of parental income, thereby linking our analysis to the extensive literature on intergenerational income mobility in the US. Our analysis proceeds based on the additive decomposition of social welfare into efficiency $\mathbb{E}_q[U]$ and inequality of opportunity $\nicefrac{1}{\theta} D_{K_U}$ (Theorem \ref{theo-main3}). As illustrated above, the quantitative impact of inequality of opportunity on social welfare depends on the degree of inequality-of-opportunity aversion $\theta$. We therefore also develop a method for computing dominance sets to rank societies. This allows us to transparently illustrate which welfare rankings are robust to disagreement over the appropriate level of inequality-of-opportunity aversion, and which are subject to normative ambiguity.  

\subsection{Trend across birth cohorts (1940s--1980s)} \label{sec-trends}

\parhead{Data} The Panel Study of Income Dynamics (PSID) is the longest-running household survey in the world, collecting rich income and family background information for a nationally representative sample of US households since 1968. Children of sampled PSID households are invited to become separate sampling units once they move out of their parental household, enabling income data to be linked across generations. This feature makes the PSID one of the key data sources for analyses of intergenerational mobility in the US \citep[among others,][]{Solon1992, Mello2025, Hufe2025, Hartley2022, Vosters2017, Mazumder2018, Nybom2024}.

We use all available PSID waves (1968--2023) and process the data in accordance with standard practices in the recent intergenerational mobility literature. First, we calculate family incomes for each PSID household in each available year. The focus on family rather than individual income allows for the inclusion of population groups with historically lower labor force attachment rates in the analysis, such as ethnic minorities and women \citep{JKN25}. We measure family income as the sum of the household head's and their co-resident spouse's total labor income, deflated to 2019 USD. We exclude family members other than the household head and their spouse from the calculation of family income to avoid mechanical correlations in family incomes between co-resident parents and children. Second, we use the time series on family incomes to construct a measure of parental income $y^{p}$ during the child's adolescence, consistent with recent evidence emphasizing the importance of parental income during this period on later-life outcomes \citep{Eshaghnia2024, Carneiro2021}. In particular, we calculate a weighted average of family income in the parental household, with weights $w$ decreasing with the temporal distance to the child's 15th birthday ($w = \max [0; 1 - |\text{Age} - 15 |/15] $).\footnote{Since the PSID records parental incomes only starting in 1968, parental income in the 1940s cohort is not centered on the child's 15th birthday, and we observe parental incomes at later points in their lifecycle than in the other cohorts. In robustness analyses, we show that our results are largely unchanged when replacing log parental income with parental income quintiles (Table \ref{tab:cohorts-ff}). Assuming that parental income quintiles are relatively stable during this life period, this suggests that differences in the timing of parental income measurement do not drive our results.} Third, we construct a measure of children's income $y^{c}$ during their prime working years. We focus on an age interval rather than an annual snapshot to smooth out transitory fluctuations in children's income and to account for differences in lifecycle earnings profiles \citep{Mello2025}. Specifically, we calculate the average family income in the child's household over ages 30-40, assigning equal weight to each annual observation.\footnote{Ideally, we would like to calculate this average over an even longer age interval. However, extending the age interval would compromise the comparability of birth cohorts in terms of their age structures, as not all considered cohorts are equally represented at each age.} We retain only observations with at least three family income observations in the considered age range. Fourth, we drop observations with missing values or incomes of zero for either $y^{p}$ or $y^{c}$. Since we measure family income averaged over several years, the share of individuals affected by this restriction is small. Finally, we reweight the resulting analysis sample to address selective attrition in the PSID \citep{Ward2023}. In particular, for each decadal birth cohort in our analysis sample, we rescale the PSID cross-sectional weights to match the relative frequencies of sex $\times$ race-cells in the corresponding pooled Current Population Survey (CPS, 1969-2023). 

Table \ref{tab:distributions} displays descriptive statistics for our analysis sample separately for each decadal birth cohort from the 1940s (1940-1949) to the 1980s (1980-1989). The sample is gender-balanced in all birth cohorts. The share of non-Hispanic Whites declines substantially, reflecting changes in the ethnic composition of the US. Across all birth cohorts, we measure children's income, on average, close to their 35th birthday. Our variables of interest, parental income $y^{p}$ and child income $y^{c}$, generally increase over time. The increase in child income $y^{c}$, however, is moderate relative to GDP per capita growth rates during this period. This discrepancy can be attributed to our unit of analysis: As we consider family income rather than individual incomes, the increasing number of single-person households depresses income growth over time \citep{Lundberg2016}.\footnote{We account for the impact of changing household compositions in sensitivity analyses (Table \ref{tab:cohorts-incdefs}).} Furthermore, due to our focus on labor income, we exclude capital income, which has grown particularly strongly over the observation period \citep{Piketty2018}. Generally, the documented cross-cohort trend of family incomes in the child generation is consistent with evidence from other data sources, including the Current Population Survey (CPS) and the National Longitudinal Surveys (NLS) \citep{davis2026decline, JKN25}.
Furthermore, we observe pronounced increases in the standard deviation of both $y^{p}$ and $y^{c}$, reflecting the well-documented increase in income inequality in the US \citep{Piketty2018}.

\parhead{Intergenerational mobility} We use these data to estimate the intergenerational elasticity of income separately for each birth cohort as follows:
\begin{align}
\ln y^c_i=\alpha+\beta^{IGE} \ln y^p_i + X_i'\delta+\epsilon_i, \label{ige}
\end{align}
where $\ln y^c_i$ and $\ln y^p_i$ are the logarithms of children's and parents' income, respectively. The vector $X_i$ includes the first- and second-order polynomials of child age to account for lifecycle effects within birth cohorts. $\epsilon_i$ is the error term. The coefficient of interest is $\beta^{IGE}$, which captures the intergenerational elasticity of income (IGE). This coefficient is often interpreted as a measure of inequality of opportunity \citep{Mazumder2023}: If it is high, children's incomes are strongly associated with their parents' incomes, limiting the influence of their own efforts and choices on economic outcomes.

Figure \ref{fig:elasticity} shows the evolution of the IGE in the US. Solid markers represent our estimates, whereas hollow markers show estimates from other studies for the birth cohorts we cover. We observe a marked increase in the IGE from $0.20$ in the 1940s to $0.44$ in the 1980s. This increase is particularly pronounced from the 1950s ($0.26$) to the 1960s ($0.38$). The IGE remains relatively stable at a high level across subsequent birth cohorts.

\begin{figure}[htbp]
  \caption{Intergenerational mobility across birth cohorts (1940s--1980s)}
  \begin{center}
  \begin{subfigure}[t]{0.49\textwidth}
    \centering
    \includegraphics[width=\textwidth]{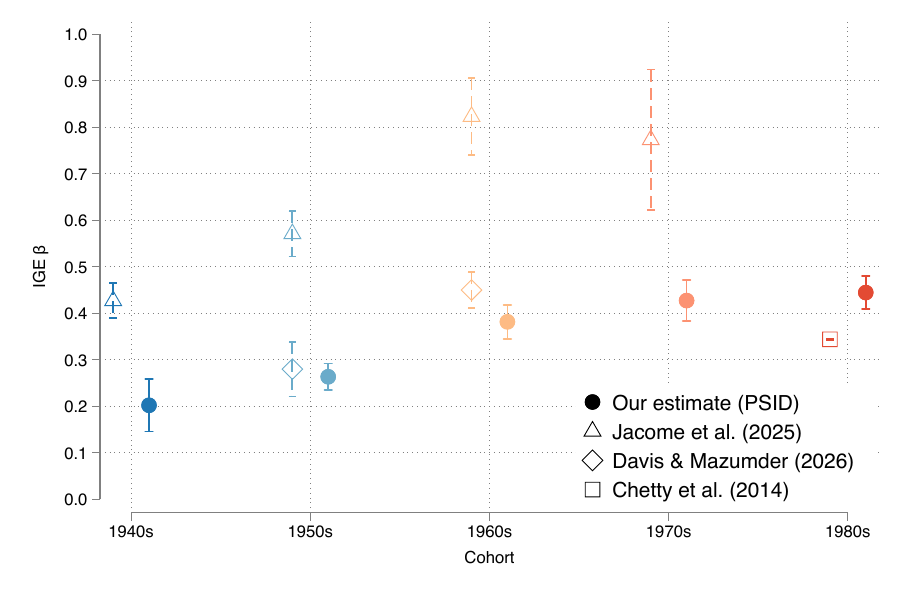}
    \caption{Intergenerational elasticity $\hat{\beta}^{IGE}$}
    \label{fig:elasticity}
  \end{subfigure}
  \hfill
  \begin{subfigure}[t]{0.49\textwidth}
    \centering
    \includegraphics[width=\textwidth]{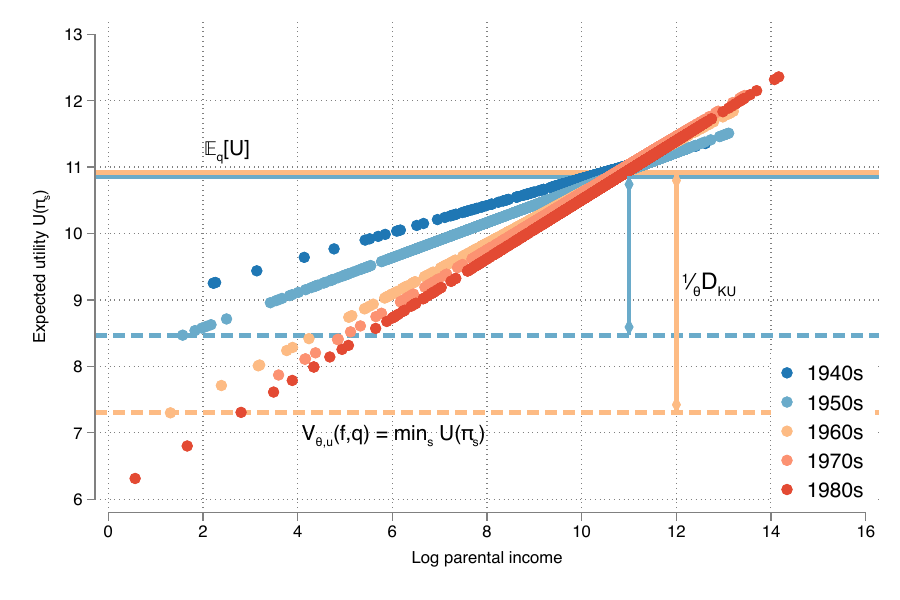}
    \caption{Expected type utilities $\hat{U}(\pi_s)$}
    \label{fig:exputility}
  \end{subfigure}
\end{center}
\begin{spacing}{0.9}
{\footnotesize
\textbf{Data:} PSID.\\
\textbf{Note:} This figure shows trends in intergenerational mobility in the US across the birth cohorts 1940s--1980s. Panel (a) shows point estimates and $95\%$ confidence intervals for the intergenerational elasticity $\beta^{IGE}$ based on Equation~\eqref{ige}. Solid markers indicate our estimates; hollow markers indicate comparable estimates from studies cited in the legend. Lower values of $\hat{\beta}^{IGE}$ indicate higher intergenerational mobility. Panel (b) shows the distribution of predicted type utilities $\hat{U}(\pi_s)$ based on Equation \eqref{exputil}. The additional horizontal and vertical lines illustrate the decomposition of social welfare $V_{\theta, u}(f,q)$ into efficiency $\mathbb{E}_q[U]$ and inequality of opportunity $\nicefrac{1}{\theta} D_{K_U}$ for the 1950s and 1960s cohorts, based on the mean-divergence representation from Theorem \ref{theo-main3}, evaluated at $\theta\to\infty$.}
\end{spacing}
\end{figure}

This pattern is consistent with previous literature. \citet{davis2026decline} use data from two cohorts of the NLS to document a pronounced increase from the 1950s to the 1960s, which they attribute to rising inequality in the 1980s, particularly affecting the 1960s birth cohort. This increase and a subsequent stabilization of the IGE are also visible in the estimates of \citet{JKN25}. Their estimates are based on imputing parental income from observable characteristics, leading to an underestimation of the variance of $y^p_i$ and a corresponding upward bias in the IGE estimates \citep{CortesOrihuela2025}. While the \emph{levels} of their IGE estimates cannot be compared with our estimates based on observed $y^p_i$, it is reassuring that the \emph{trends} across birth cohorts align. Lastly, we can compare our estimates to \citet{Chetty2014}, who base their estimates of the IGE for the birth cohort around 1980 on a large tax data sample. In this data, they observe children at comparatively young ages (29-32) and construct $y_i^p$ and $y_i^c$ as averages over relatively short time intervals of five and two years, respectively. This data structure likely depresses their IGE estimates relative to those based on more comprehensive measures of children's and parents' lifecycle incomes \citep{Mazumder2016}. 

Overall, we conclude that our estimates of IGE trends in the US are consistent with the existing literature, suggesting that intergenerational mobility has decreased for American children born in the post-war period. 

\parhead{Opportunity-sensitive social welfare} The IGE directly maps onto our theoretical framework. To see this, recall that an opportunity-sensitive planner evaluates social welfare by assessing the distribution of expected type utilities $U(\pi_s)$. These expected type utilities can be calculated using data inputs from IGE estimations: Types $s\in \states$ are defined by levels of log parental income $\ln y^p$, demographic weights $q(s)$ are given by the empirical distribution across $\ln y^p$, and children's utility is captured by log-transformed incomes, consistent with Theorem~\ref{theo-ln}. Thus, using the linearity assumption inherent in the IGE, we can estimate expected type utilities $U(\pi_s)$ as the conditional expectation of $\ln y^c$ given $\ln y_s^p$:\footnote{We demean age controls within each cohort before estimating Equation \eqref{ige}. Therefore, the level effects of age are absorbed by $\hat{\alpha}$, ensuring that the prediction shown in Equation \eqref{exputil} is evaluated at the cohort's mean age. Generally, age controls have a negligible effect as there is little variation around the average age at which income is measured (Table \ref{tab:distributions}).}
\begin{align}
\hat{U}(\pi_s)=\mathbb{E}[\ln y^c \lvert \ln y_s^p]=\hat{\alpha}+\hat{\beta}^{IGE}\times\ln y_s^p.\label{exputil}
\end{align}

Figure \ref{fig:exputility} translates this mapping into graphical form and provides an empirical analog to the representation of our theoretical framework in Figure \ref{fig: theosum}. Each data point represents the estimated expected utility $\hat{U}(\pi_s)$ for a given parental income level. Cohort-specific slopes correspond to the point estimates $\hat{\beta}^{IGE}$ in Figure \ref{fig:elasticity}, indicating a decrease in intergenerational mobility over time. Focusing on the cohorts of the 1950s and 1960s as an illustrative example, the solid horizontal lines indicate that their average expected type utilities are almost identical. Thus, social welfare differences between these birth cohorts are not driven by efficiency $\mathbb{E}_q[U]$. At the same time, the increase in $\hat{\beta}^{IGE}$ from the 1950s to the 1960s, together with a rise in the dispersion of parental income, induces a spread in $\hat{U}(\pi_s)$, leading to an increase in inequality of opportunity $\nicefrac{1}{\theta} D_{K_U}$. The strength of this change is governed by the extent of inequality-of-opportunity aversion $\theta$. In the figure, we highlight the limiting case of $\theta\to\infty$, i.e., the case where social welfare is determined solely by the outcomes of the most disadvantaged type. In this case, social welfare is substantially higher in the 1950s than in the 1960s.

Figure~\ref{fig:relindex} extends this analysis to all birth cohorts in our sample. In particular, we trace social welfare $V_{\theta, u}(f,q)$ for different levels of inequality-of-opportunity aversion $\theta$. The full length of the bars indicates the efficiency component $\mathbb{E}_q[U]$. The stacked bar segments capture the decrease in social welfare due to inequality of opportunity, $\nicefrac{1}{\theta} D_{K_U}$.
\begin{figure}[htbp]
  \caption{Social welfare across birth cohorts (1940s--1980s)}
    \begin{center}
  \begin{subfigure}[t]{0.49\textwidth}
\includegraphics[width=\textwidth]{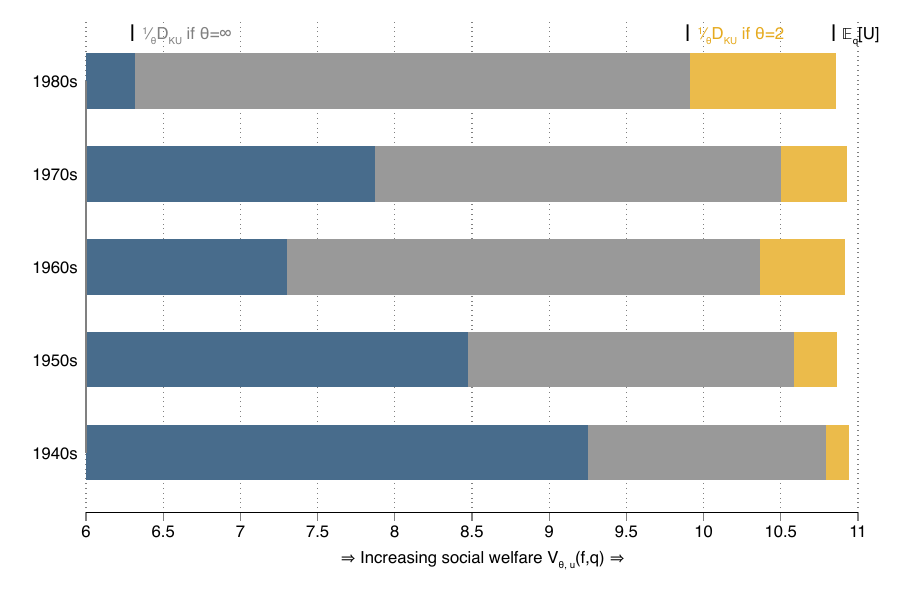}
  \caption{Decomposition}  \label{fig:relindex}
  \end{subfigure}
  \hfill
  \begin{subfigure}[t]{0.49\textwidth}
\includegraphics[width=\textwidth]{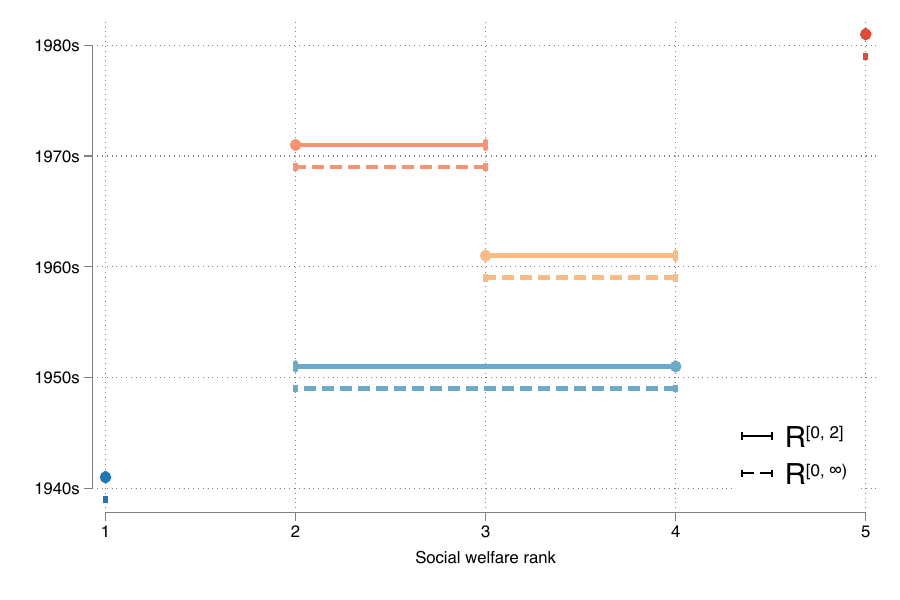}
 \caption{Dominance sets for ranking}  \label{fig:rol}
  \end{subfigure}
\end{center}
\begin{spacing}{0.9}
{\footnotesize
\textbf{Data:} PSID.\\
\textbf{Note:} This figure shows trends in social welfare across the birth cohorts 1940s--1980s. Panel (a) shows the decomposition of social welfare $V_{\theta, u}(f,q)$ into efficiency $\mathbb{E}_q[U]$ and inequality of opportunity $\nicefrac{1}{\theta} D_{K_U}$, based on the mean-divergence representation from Theorem \ref{theo-main3}, evaluated at $\theta=2$ and $\theta\to\infty$. Panel (b) shows dominance sets for the social welfare ranking based on $\theta \in [0, 2]$ (solid lines) and $\theta \in [0, \infty)$ (dashed lines). Rank 1 (5) corresponds to the cohort with the highest (lowest) social welfare.}
\end{spacing}
\end{figure}

The results suggest that the efficiency component $\mathbb{E}_q[U]$ does not vary substantially across birth cohorts in the US: The cross-cohort standard deviation of $\mathbb{E}_q[U]$ is merely $0.039$. The attenuated variation in $\mathbb{E}_q[U]$ is unsurprising, given the relative stability of family incomes over time (Table \ref{tab:distributions}). Instead, cross-cohort differences in social welfare are mostly driven by the inequality-of-opportunity component $\nicefrac{1}{\theta} D_{K_U}$. The corresponding cross-cohort standard deviation of $\nicefrac{1}{\theta} D_{K_U}$ is $0.306$ at $\theta=2$, i.e., roughly eight times larger than for $\mathbb{E}_q[U]$.\footnote{It is even larger at $\theta\to\infty$, when welfare is evaluated at the lowest observed parental income. Welfare rankings are unchanged when types are identified by parental income quintiles, suggesting that the dominance results presented below are not driven by sensitivity to the tail (Table \ref{tab:cohorts-ff}).}

Figure~\ref{fig:relindex} also illustrates that cross-cohort rankings are sensitive to the choice of $\theta$. For example, the 1950s are ranked lower than the 1960s for low values of inequality-of-opportunity aversion, whereas this ranking reverses as $\theta$ increases. To assess the extent of normative ambiguity in our cross-cohort comparisons, we next propose a method to compute dominance sets for the ranking of societies.

\parhead{Welfare ranking} In particular, we modify the rank inference procedure of \citet{MRSW24}, who convert confidence intervals on intergenerational mobility coefficients into confidence sets for mobility rankings. In their application, \citet{MRSW24} focus on the expected income rank of children who grew up with parents at the 25th percentile of the national income distribution. However, their procedure is equally applicable to other mobility statistics such as the IGE. It can be summarized as follows. Denote by $\hat{\psi_j}$ any intergenerational mobility estimate for society $j \in \mathcal{J}$, where a higher $\hat{\psi}$ indicates higher mobility. First, they fix a desired level of statistical significance $\alpha$. Second, they construct corresponding confidence intervals for all pairwise differences $\hat{\psi_j} - \hat{\psi_k}$. These confidence intervals are simultaneous, i.e., they jointly cover all pairwise differences $\hat{\psi_j} - \hat{\psi_k}$ across societies in $\mathcal{J}$ at level $1-\alpha$. Third, they count the number of societies $k\neq j$ for which $j$ can be unambiguously classified as having higher (lower) mobility at the chosen level of statistical significance. Denoting by $n_j^+$ ($n_j^-$) the number of societies $k\neq j$ for which the simultaneous confidence interval on $\hat{\psi_j} - \hat{\psi_k}$ is entirely above (below) zero, they obtain the confidence set

\begin{equation}
R_j^{1-\alpha}\;=\; \bigl\{\, n_j^- + 1, \;\ldots,\; |\mathcal{J}| - n_j^+ \,\bigr\}.
\label{eq:ranksets}
\end{equation}

The width of this confidence set depends on the degree of statistical uncertainty in intergenerational mobility estimates: If estimates are noisy, simultaneous confidence intervals are wide, few societies $k$ can be placed unambiguously above or below $j$, and $R_j^{1-\alpha}$ is large, potentially covering the entire range of ranks.

We adapt the logic of this procedure to normative ambiguity, i.e., we derive dominance sets for welfare rankings when the appropriate normative parameter $\theta$ is unknown. First, we fix an interval $[\theta_{\min}, \theta_{\max}]$ of admissible values of the aversion parameter. Second, we compute pairwise welfare differences $V^j_{\theta, u}(f,q) - V^k_{\theta, u}(f,q)$ across this interval. Third, denoting by $n_j^+$ ($n_j^-$) the number of societies $k\neq j$, for which $V^j_{\theta, u}(f,q) - V^k_{\theta, u}(f,q)$ is positive (negative) for all $\theta\in [\theta_{\min}, \theta_{\max}]$, we construct the dominance set $R_j^{[\theta_{\min}, \theta_{\max}]}$. The width of this dominance set depends on the degree to which pairwise welfare orderings vary across $\theta \in [\theta_{\min}, \theta_{\max}]$: If orderings are sensitive to $\theta$, few societies $k$ can be placed unambiguously above or below $j$, and $R_j^{[\theta_{\min}, \theta_{\max}]}$ grows large.

The resulting object is set-identified rather than probabilistic: It illustrates which ranks for $j$ are robust to normative disagreement across $\theta \in [\theta_{\min}, \theta_{\max}]$. The criterion becomes more stringent as the $\theta$-interval widens. In the limit, requiring the welfare ordering to hold for $\theta \in [0,\infty)$, it is equivalent to a ranking by stochastic dominance (Theorem \ref{theo-stoch_dom}).\footnote{Note $R_j^{[\theta_{\min}, \theta_{\max}]}$ reflects the full range of ranks society $j$ could hold in any complete ranking consistent with the unambiguous pairwise dominance preorder over $[\theta_{\min}, \theta_{\max}]$. This differs conceptually from a sharp attainable rank set, which would include only the realized ranks of society $j$ at each $\theta$. While our approach provides a theoretical outer bound on the attainable rank set, the empirical difference is small--see Tables \ref{tab:cohorts-sharp} and \ref{tab:cz-sharp}.} However, the procedure is flexible enough to accommodate less stringent criteria by restricting $[\theta_{\min}, \theta_{\max}]$ to a subrange that reflects a narrower class of normative positions.\footnote{To implement this in practice, we use a $\theta$ grid sufficiently fine to rule out pairwise rank reversals between neighboring grid nodes. Specifically, the rate of welfare changes with respect to $\theta$ can be bounded analytically. At a given rate, a pairwise welfare difference that has the same sign over a short enough $\theta$-interval cannot reach zero in between. We use this insight to check every pairwise difference over the whole grid, refining the nodes as needed, and close the limit $\theta \rightarrow \infty$ using the maximin value from Property 4. Therefore, the computed sets exactly reflect continuous-interval dominance.\label{fn:grid}}

Figure \ref{fig:rol} shows dominance sets $R^{[0, 2]}$ for US birth cohorts from the 1940s to the 1980s. There is no clear ranking among the birth cohorts 1950s--1970s, i.e., their social welfare rankings vary with assumptions about $\theta$. In contrast, the 1940s and 1980s cohorts can be unambiguously classified as having the highest and lowest social welfare, respectively. One may ask whether the unambiguous classification of these two cohorts is an artifact of a rather restrictive interval for the $\theta$ parameter. To answer this question, we also consider a more demanding criterion by calculating $R^{[0, \infty)}$. The results show that $R^{[0, 2]}$ and $R^{[0, \infty)}$ are fully overlapping across all cohorts, implying that the 1940s (1980s) cohort ranks best (worst) for any possible level of inequality-of-opportunity aversion. Importantly, we recall that while this dominance result is derived under the assumption of logarithmic income utility within types, the ranking extends to a much wider class of second-stage opportunity-aversion functions, including CRRA and CARA (Theorem \ref{theo-stoch_dom}).

\parhead{Additional analyses} We conduct a series of additional analyses to probe the sensitivity of our results. First, income data from surveys may be distorted by measurement error, particularly in the tails \citep[e.g.,][]{Brewer2017}. Table \ref{tab:cohorts-trim} shows that dominance sets are largely robust to trimming, suggesting that mismeasured tails are not driving our results. Second, in line with the literature on intergenerational mobility, our baseline specification uses a linear mapping between $\ln y^p$ and $\ln y^c$ to estimate expected type utilities (Equation (\ref{ige})). Table~\ref{tab:cohorts-ff} shows that dominance sets are largely similar whether we estimate $U(\pi_s)$ using second- or third-order polynomials in $\ln y^p$ or use non-parametric estimates based on parental income quintiles. Third, our baseline specification uses total family labor income as the outcome of interest. However, one may argue that this income concept does not adequately capture economic well-being as it neglects differences in household size and transfer income. Table \ref{tab:cohorts-incdefs} shows that dominance sets remain largely unchanged when we replace labor income with disposable income after taxes and transfers. However, dominance sets widen when we correct for differences in household size by using an equivalence scale. This result is driven by the fact that average outcomes improve in later cohorts once we account for the decrease in US household sizes over time \citep{Lundberg2016}. Lastly, \citet{MRSW24} show that any ranking of societies based on estimated objects is subject to statistical uncertainty. Therefore, we augment our analysis to account for sampling variation in $V_{\theta,u}(f,q)$.\footnote{We provide a detailed methodological overview of our approach to estimating dominance sets, confidence sets, and combined sets in Appendix~\ref{app: confsets}.} Table~\ref{tab:cohorts-cs} shows that confidence sets for statistical uncertainty are wider than dominance sets for normative ambiguity at $\theta=0$ and $\theta=2$. This result is driven by the relatively small sample sizes underlying each cohort-specific estimate in the PSID, and the relatively small differences in $\mathbb{E}_q[U]$ between cohorts. As $\theta \to \infty$, statistical confidence sets narrow because differences in $V_{\theta,u}(f,q)$ are now governed by the relatively large cross-cohort utility differences among the worst-off types. 
\subsection{Differences across regions (Commuting Zones)}
\label{sec-space}

\parhead{Data} We draw on the data produced by \citet{Chetty2014, Chetty2026}. These papers are based on the universe of federal income tax returns, linking parents and their children across generations. The sample is large and representative of the entire population, allowing the authors to calculate indicators of intergenerational mobility at various levels of regional disaggregation. Many summary statistics from these analyses are publicly available and have become important inputs for analyses of intergenerational mobility in the US, among others, in \citet{B23, MRSW24, Derenoncourt2022, Gallagher2019}.

The data of \citet{Chetty2014} cover the birth cohorts 1980--1985.\footnote{Their main analysis sample covers the cohorts 1980--1982. However, various CZ-level statistics required for our analysis, such as quintile transition probabilities, are reported for the 1980--1985 cohort in the public data. This difference matters little in practice given the stability of various rank-mobility statistics in neighboring birth cohorts \citep{Chetty2014b}.} Their data preparation protocol parallels our processing of the PSID in several important dimensions. First, \citet{Chetty2014} also focus on family income. Concretely, they compute family income as the sum of the individual incomes of married partners from labor and capital, plus the non-taxable portion of Social Security and disability benefits. Second, parental income $y^p$ is measured during children's adolescence by averaging across 1996--2000, i.e., when children are between 14 and 20 years old. Third, the child income $y^c$ is the average of annual incomes in 2011 and 2012, i.e., when children are between 26 and 32 years old. The data cover $95\%$ of the targeted population, and are therefore broadly representative of children in the relevant birth cohorts. 

We do not have access to the individual-level data used by \citet{Chetty2014}. Therefore, we combine various summary statistics from Opportunity Insights to reconstruct for each CZ the joint distribution of children's family incomes ($y^c$) and parental income quintiles defined at the national level ($r^p$). We proceed in two steps. First, for each CZ, we construct a frequency-weight matrix $w^{CZ}(r^c,r^p)$ over children's and parents' national income positions. Concretely, in the publicly available data, parents are grouped into their five national income quintiles $r^p=1,\dots,5$. For each parental income quintile, we know the CZ-specific share of children in that quintile $q^{CZ}(r^p)$, and the CZ-specific parent-child quintile transition $t^{CZ}(k(r^c)\mid r^p)$, i.e., the probability that a child of origin $r^p$ reaches the child income quintile $k$ containing income percentile $r^c$. Furthermore, we know the national within-quintile placement $h(r^c \mid k(r^c), r^p)$, i.e., the share of children of origin $r^p$ who reach percentile $r^c$ within $k$. Based on these inputs, we construct CZ-specific frequency-weight matrices as follows:
\begin{equation}
  w^{CZ}(r^c,r^p)\;=\;q^{CZ}(r^p) \times t^{CZ}\!\big(k(r^c)\mid r^p \big) \times h(r^c\mid k(r^c), r^p),
  \label{eq:fmap}
\end{equation} 
where we increase the resolution within children's income quintiles by assuming that the placement probability $h(r^c \mid k(r^c), r^p)$ is common across CZ. Second, we combine $w^{CZ}(r^c, r^p)$ with the national percentile-income schedule $g(y^c \mid r^c)$ provided by Opportunity Insights, allowing us to back out the CZ-level joint distribution of $y^c$ and $r^p$. We assess the accuracy of this procedure in two steps. First, we marginalize the resulting joint distributions over $r^p$ and compare various moments of $y^c$ with the corresponding moments computed by \citet{Chetty2014} on individual-level microdata (Figure \ref{fig:validation}). Second, we analogously compare the joint distributions of child ranks and parental ranks $r^p$ (Figure \ref{fig:validation_joint}). Reassuringly, both tests show that the distributions align very closely.\footnote{The accuracy of our procedure decreases somewhat in the upper tail around the 99th income percentile of children's incomes---a result of the fact that the national-level placement function $h(r^c \mid k(r^c), r^p)$ cannot capture CZ-specific idiosyncracies in the tails of the distributions.} In line with the PSID application, we restrict resulting data to observations with positive $y^c$, renormalize frequency weights within each CZ, and express monetary values in 2019 USD.  

Table~\ref{tab:cz-summary} provides summary statistics for the 50 largest CZs, which we will focus on in our analysis. Population size varies substantially across CZs, from around $85,000$ children in Austin to more than $1,018,000$ in Los Angeles. While the sample is gender-balanced across all regions, CZs vary in their ethnic composition: The share of non-Hispanic Whites is highest in Manchester and lowest in Miami. Average child income ranges from around $\$45,000$ in Jacksonville and Orlando to approximately $\$65,000$ in Minneapolis.  

\parhead{Intergenerational mobility} In this application, we cannot rely on standard measures like the IGE or the rank-rank slope (RRS). On the one hand, we cannot compute the IGE as we do not have access to $\ln y^p$. On the other hand, we do not compute the RRS used by \citet{Chetty2014, Chetty2026}, as our normative framework assesses child outcomes in monetary terms and not ranks. Taking account of these restrictions, we estimate intergenerational income mobility separately for each CZ as follows:
\begin{align}
\ln y^{c}_{i} 
= \alpha + \beta^{{IRS}} \, r_i^p + \epsilon_{i},
\label{eq:irs}
\end{align}
where $\ln y^{c}_{i}$ is the log of child income and $r_i^p$ is the parental income quintile, represented by its midpoint rank, $ r_i^p \ in \{10, 30, 50, 70, 90\}$. As we do not observe individual children, an observation $i$ is a rank-quintile $(r^c, r^p)$ pair from the reconstructed joint distribution, weighted by its frequency $w^{CZ}(r^c, r^p)$. The coefficient of interest, $\beta^{IRS}$, captures the average change in log child income associated with a one-rank increase in parental income. We call this measure income-rank slope (IRS) and interpret it as a summary measure of intergenerational mobility at the CZ level.

Figure \ref{fig:mobility_cz} shows that the IRS replicates much of the regional variation embedded in standard mobility measures. In particular, we compare CZ-specific rank-rank slopes (RRS) estimated by \citet{Chetty2014} in Panel (a) to our estimates of the IRS in Panel (b).
\begin{figure}[htbp]
  \caption{Intergenerational mobility across commuting zones}
  \begin{center}
  \label{fig:mobility_cz}
  \begin{subfigure}[t]{0.49\textwidth}
    \centering
    \includegraphics[width=\textwidth]{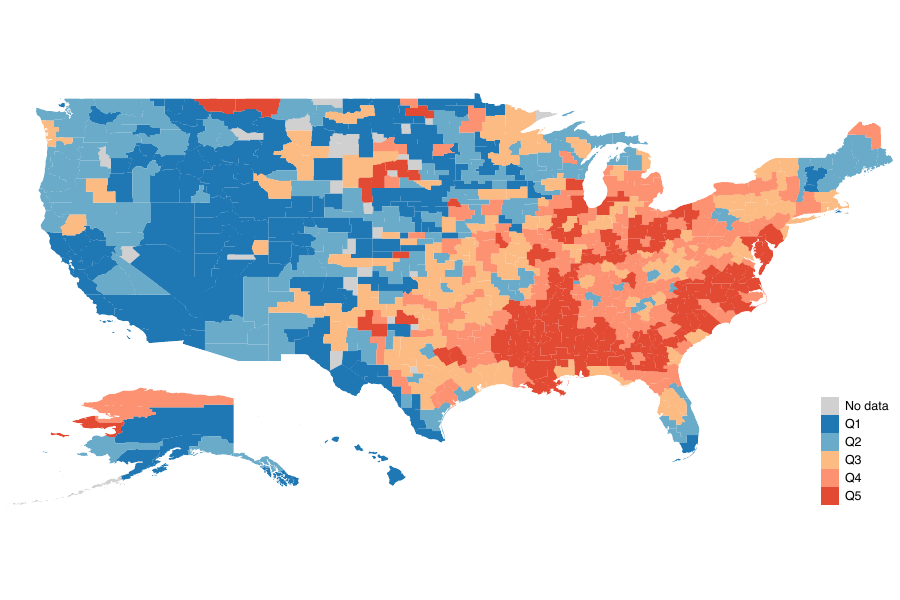}
    \caption{Rank-rank slope (RRS)}
    \label{fig:rrs}
  \end{subfigure}
  \hfill
  \begin{subfigure}[t]{0.49\textwidth}
    \centering
\includegraphics[width=\textwidth]{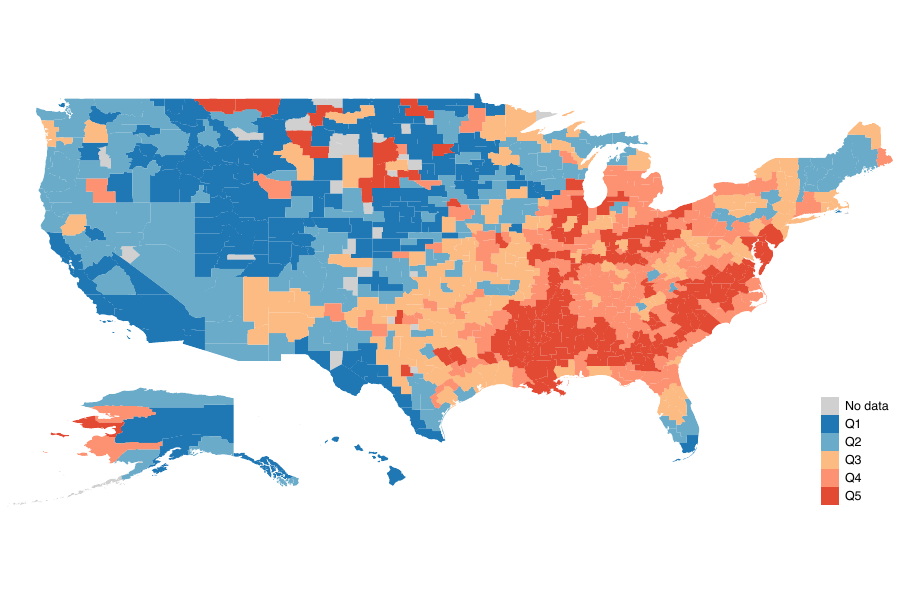}
  \caption{Income-rank slope (IRS)}  \label{fig:irs}
  \end{subfigure}
\end{center}
\begin{spacing}{0.9}
{\footnotesize
\textbf{Data:} Opportunity Insights.\\
\textbf{Note:} This figure shows intergenerational mobility in the US across commuting zones (CZs). Panel (a) shows the variation of rank-rank slopes (RRS) estimated by \citet{Chetty2014}. Panel (b) shows the variation of income-rank slopes (IRS) estimated based on Equation \eqref{eq:irs}. CZs are grouped by quintiles of the RRS and IRS, respectively. Q1 (Q5) corresponds to the group of CZs with the highest (lowest) intergenerational mobility.}
\end{spacing}
\end{figure}
Both estimators of intergenerational mobility yield similar spatial patterns: CZs in the South exhibit the lowest levels, while those in the West and Midwest exhibit the highest. The Spearman rank correlation coefficient between the RRS and the IRS is $0.95$, indicating a strong concordance between the two measures. The correlation increases to $0.99$ when considering only the 50 largest CZs, which are the focus of our analysis.

\parhead{Opportunity-sensitive social welfare} Figure~\ref{fig:relindex_cz} illustrates our decomposition of social welfare $V_{\theta, u}(f,q)$ into efficiency $\mathbb{E}_q[U]$ and inequality of opportunity $\nicefrac{1}{\theta} D_{K_U}$ based on the representation from Theorem \ref{theo-main3} for the 50 largest CZs in the US. As before, the full length of the bars indicates $\mathbb{E}_q[U]$. The stacked bar segments show decreases in social welfare driven by $\nicefrac{1}{\theta} D_{K_U}$ across different values of $\theta$. 
\begin{figure}[htbp]
  \caption{Social welfare across commuting zones}
  \begin{center}
  \begin{subfigure}[t]{\textwidth}
    \centering
\includegraphics[width=0.8\textwidth]{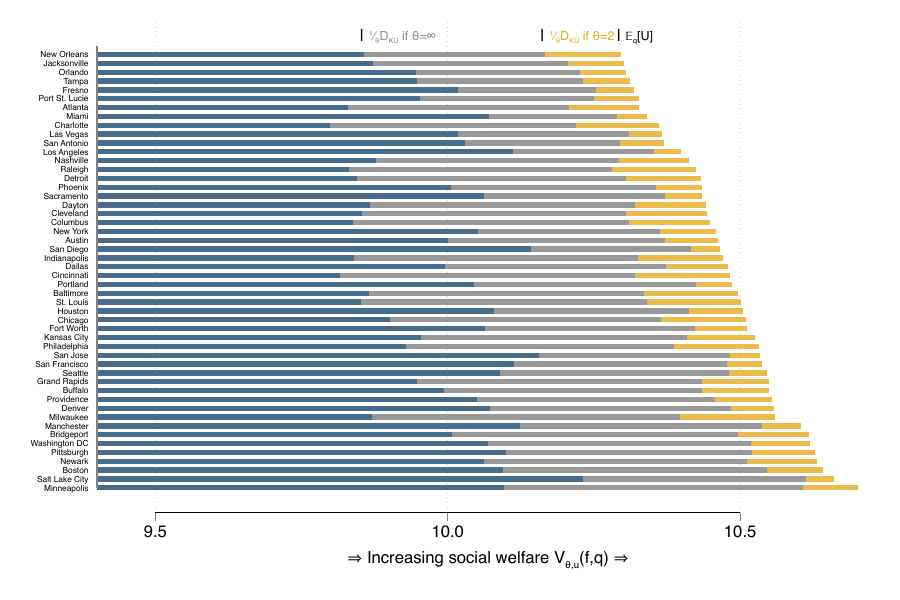}
  \caption{Decomposition}  \label{fig:relindex_cz}
  \end{subfigure}
  \begin{subfigure}[t]{\textwidth}
    \centering
\includegraphics[width=0.8\textwidth]{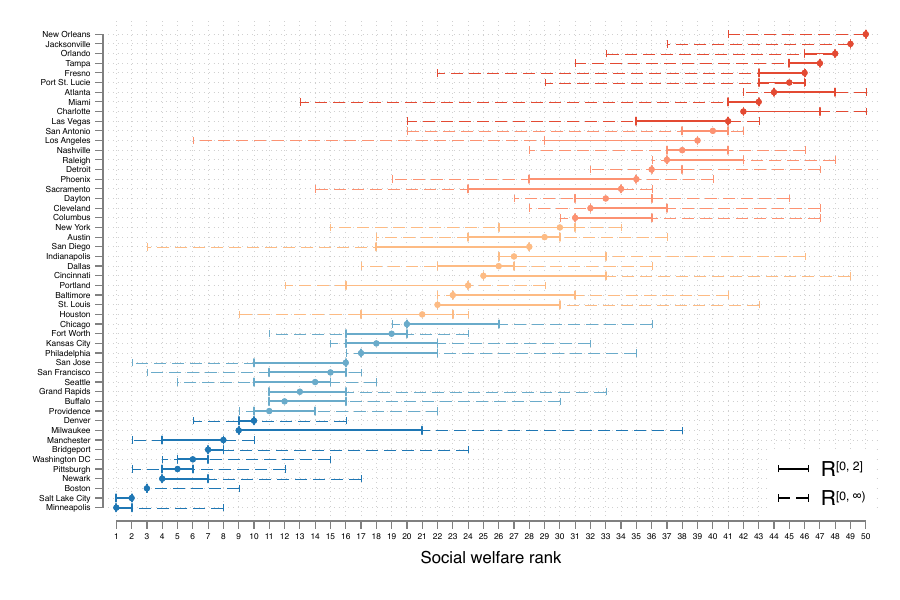}
 \caption{Dominance sets for ranking}  \label{fig:rol_cz}
  \end{subfigure}
    \end{center}
\begin{spacing}{0.9}
{\footnotesize
\textbf{Data:} Opportunity Insights.\\
\textbf{Note:} This figure shows variation in social welfare across the 50 largest commuting zones (CZs). Panel (a) shows the decomposition of social welfare $V_{\theta, u}(f,q)$ into efficiency $\mathbb{E}_q[U]$ and inequality of opportunity $\nicefrac{1}{\theta} D_{K_U}$, based on the mean-divergence representation from Theorem \ref{theo-main3}, evaluated at $\theta=2$ and $\theta\to\infty$. Panel (b) shows dominance sets for the social welfare ranking based on $\theta \in [0, 2]$ (solid lines) and $\theta \in [0, \infty)$ (dashed lines). Rank 1 (50) corresponds to the CZ with the highest (lowest) social welfare.}
\end{spacing}
\end{figure}

Compared with the cross-cohort comparison, regional differences in the efficiency component $\mathbb{E}_q[U]$ are much more pronounced: The cross-CZ standard deviation in $\mathbb{E}_q[U]$ is $0.104$. Furthermore, the variation of $\mathbb{E}_q[U]$ exceeds the variation of inequality of opportunity for low levels of inequality-of-opportunity aversion: When setting $\theta=2$, the cross-CZ standard deviation in $\nicefrac{1}{\theta} D_{K_U}$ is $0.034$, i.e., it is less than half of the corresponding variation in $\mathbb{E}_q[U]$. However, variation in inequality of opportunity grows with the aversion parameter $\theta$: When $\theta\to \infty$, the cross-CZ standard deviation in $\nicefrac{1}{\theta} D_{K_U}$ is $0.110$, i.e., it is on par with the corresponding variation in $\mathbb{E}_q[U]$.

Figure \ref{fig:relindex_cz} also illustrates strong heterogeneity in the importance of inequality of opportunity for social welfare comparisons. For example, Los Angeles ranks 39th in social welfare based on efficiency, $\mathbb{E}_q[U]$. However, compared with other CZs, social welfare declines less with increasing inequality-of-opportunity aversion, suggesting Los Angeles offers better opportunities for children from lower parental-income ranks than many other CZs. Conversely, Bridgeport ranks 7th when considering $\mathbb{E}_q[U]$. However, social welfare decreases more sharply as aversion to inequality of opportunity increases. At the highest level of $\theta$, Bridgeport is outperformed by Los Angeles. Hence, the comparative evaluation of Los Angeles and Bridgeport crucially depends on our normative stance: How much do we care about the outcomes of children in the lowest tail of the parental income distribution, as opposed to other children who are better off but still worse off than the average?\footnote{This intuition also relates to recent studies that evaluate societies using absolute mobility statistics, i.e., the average outcomes of children from a fixed position in the parental income distribution \citep[among others,][]{Boustan2025, Britto2026}. Assuming expected child outcomes rise monotonically with parental rank, each level of $\theta$ can be translated into the average parental income rank considered by the social planner. For example, when evaluating societies based on the average outcomes of children at the lowest parental-income percentile, researchers implicitly assume $\theta\to\infty$. We illustrate this link formally in Appendix \ref{app: normweights}.} In light of this normative ambiguity, we again calculate dominance sets to summarize the range of potential rankings for each CZ.

\parhead{Welfare ranking} Figure \ref{fig:rol_cz} shows dominance sets for the 50 largest CZs in the US. Solid markers indicate the ranking of each CZ by efficiency $\mathbb{E}_q[U]$, and solid (dashed) lines indicate the dominance sets $R^{[0, 2]}$ ($R^{[0, \infty)}$). 

In contrast to the cross-cohort comparison, the results show markedly higher normative ambiguity. Considering $\theta\in [0,2]$, only Boston, Jacksonville, and New Orleans can be ranked unambiguously at the 3rd, 49th, and 50th ranks, respectively. The average (median) dominance set comprises 5.6 (6) ranks across the 50 considered CZs. Thus, the amount of normative ambiguity is comparable to the amount of statistical uncertainty in mobility rankings documented by \citet{MRSW24}, which ranges between 3 and 7 ranks for the $95\%$ simultaneous confidence sets of the top and bottom five CZs, ranked in terms of absolute upward mobility at the 25th percentile of parental income. The extent of normative ambiguity increases considerably when considering the full range of $\theta$. In this case, no CZ can be ranked unambiguously, and the average (median) dominance set increases to 17.6 (18) ranks across the 50 considered CZs.

\parhead{Additional analyses}  We again conduct additional analyses to probe the sensitivity of these results. First, Figure \ref{fig:validation} shows that our reconstruction of CZ-specific income distributions is less accurate in the tails. To assess the implications for dominance sets, we trim the child income distributions at the bottom and top. Table~\ref{tab:cz-trim} shows that trimming the child income distribution narrows dominance sets. This change is driven by changes in the inequality-of-opportunity component $\nicefrac{1}{\theta} D_{K_U}$, where excluding the tails increases the homogeneity of income-rank slopes across CZs. Rankings by the efficiency component $\mathbb{E}_q[U]$ remain largely unaffected. Second, we again assess the importance of the linearity assumption when estimating expected type utilities $U(\pi_s)$ (Equation (\ref{eq:irs})). Table~\ref{tab:cz-ff} shows results are robust to various functional-form assumptions, similar to the application based on birth cohorts. Third, \citet{Chetty2014} observe children relatively early in their lifecycle at ages 26-32. To assess the sensitivity of our results to lifecycle bias, we repeat our analysis by replacing the national percentile-income schedule $g(y^c \mid r^c)$ of \citet{Chetty2014} with the corresponding schedule of \citet{Chetty2026}, in which children's incomes are observed at ages 31--37.\footnote{This procedure is valid if the CZ-level mapping between parental income ranks and children's income ranks is the same across both data sources. \citet{Chetty2014} provide initial support for this assumption by showing that the RRS stabilizes when children are in their early 30s, suggesting that any lifecycle bias is likely confined to changes in marginal income distributions. Figure~\ref{fig:vintages} further supports this assumption by showing that CZ-level estimates of the RRS and absolute upward mobility at the 25th percentile of parental incomes are broadly concordant in both data sources.} Table~\ref{tab:cz-val} shows that dominance sets are largely unchanged when considering children at later ages. Finally, we also incorporate statistical uncertainty into our welfare comparisons. When considering normative ambiguity alone, dominance sets on average span 5.6 ranks for $\theta \in [0,2]$ and 17.6 ranks for $\theta \in [0,\infty)$ (Table~\ref{tab:cz-cs}). These sets increase to 9.6 ranks and 22.7 ranks when additionally accounting for statistical uncertainty via simultaneous confidence sets. This moderate increase is consistent with the results of \citet{MRSW24}, who show that statistical uncertainty in ranks is small when focusing on mobility estimates in large samples, such as the 50 largest CZs in the US. As a consequence, normative ambiguity and statistical uncertainty contribute in comparable measure when inequality-of-opportunity aversion is restricted to $\theta\in[0, 2]$. However, when it is left unrestricted, normative ambiguity becomes the primary barrier to a clear-cut welfare ranking of these regions.  


\section{Conclusion}
\label{sec-conclusion}

This paper develops a framework for evaluating income distributions when inequality of opportunity is a source of social concern. We characterize a tractable class of social welfare functions in which a single parameter governs aversion to inequality of opportunity. The resulting criterion admits a mean-divergence representation that separates average welfare from a structural penalty for disparities in opportunities, as well as a stochastic dominance criterion for welfare comparisons that are robust to the planner’s degree of opportunity aversion. 

Our theoretical framework lends itself to straightforward empirical implementation and accommodates standard measures of intergenerational mobility, which are the most prevalent proxies for unequal opportunities in the current economics literature. Using data from the United States, we show that robust rankings of opportunity-sensitive social welfare are the exception rather than the rule: Welfare rankings vary substantially with the assumed degree of inequality-of-opportunity aversion. 

Importantly, by providing an explicit mapping from observable income distributions to opportunity-sensitive social welfare, our framework makes normative judgments transparent and allows researchers to characterize conclusions that remain valid despite disagreement over the appropriate trade-off between efficiency and equality of opportunity.

\clearpage

\bibliographystyle{ecca}
\bibliography{Documents/Biblio_202602}

@string{ AER   = {American Economic Review} }

@string{ EER   = {European Economic Review} }

@string{ EK    = {Econometrica} }

@string{ EL    = {Economics Letters} }

@string{ IER   = {International Economic Review} }

@string{ JEL   = {Journal of Economic Literature} }

@string{ JET   = {Journal of Economic Theory} }

@string{ JME   = {Journal of Mathematical Economics} }

@string{ JMP   = {Journal of Mathematical Psychology} }

@string{ JPoE  = {Journal of Political Economy} }

@string{ MS    = {Management Science} }

@string{ QJE   = {Quarterly Journal of Economics} }

@string{ SJE   = {Scandinavian Journal of Economics} }

@string{ RES   = {Review of Economic Studies} }

@string{ TE    = {Theoretical Economics} }

@book{A66,
  author    = {Acz\'el, J.},
  title     = {Lectures on Functional Equations and Their Applications},
  publisher = {Academic Press},
  address   = {New York},
  year      = {1966}
}

@book{Ar65,
  author    = {Arrow, K.J.},
  title     = {Aspects of the Theory of Risk-Bearing},
  publisher = {Yrj{\"o} Jahnssonin S{\"a}{\"a}ti{\"o}},
  address   = {Helsinki},
  year      = {1965}
}

@book{DE97,
  author    = {Dupuis, P. and Ellis, R.S.},
  title     = {A Weak Convergence Approach to the Theory of Large Deviations},
  series    = {Wiley Series in Probability and Statistics},
  publisher = {John Wiley \& Sons},
  address   = {New York},
  year      = {1997}
}

@book{Morgenstern1944,
  title     = {Theory of Games and Economic Behavior},
  author    = {von Neumann, John and Morgenstern, Oskar},
  year      = {1944},
  publisher = {Princeton University Press},
  address   = {Princeton}
}

@book{Savage1954,
  author    = {Savage, Leonard J.},
  title     = {The Foundations of Statistics},
  year      = {1954},
  publisher = {John Wiley \& Sons},
  address   = {New York}
}

@book{SS07,
  author    = {Shaked, M. and Shanthikumar, J.G.},
  title     = {Stochastic Orders},
  series    = {Springer Series in Statistics},
  publisher = {Springer},
  address   = {New York},
  year      = {2007}
}

@incollection{De60,
  author    = {Debreu, G.},
  title     = {Topological Methods in Cardinal Utility Theory},
  booktitle = {Mathematical Methods in the Social Sciences, 1959},
  editor    = {K. J Arrow and S. Karlin and P. Suppes},
  publisher = {Stanford University Press},
  address   = {Stanford},
  year      = {1960},
  pages     = {16-26}
}

@incollection{Mazumder2016,
  author    = {Mazumder, Bhashkar},
  title     = {Estimating the Intergenerational Elasticity and Rank Association
               in the United States: Overcoming the Current Limitations of
               Tax Data},
  booktitle = {Inequality: Causes and Consequences},
  series    = {Research in Labor Economics},
  volume    = {43},
  publisher = {Emerald Group Publishing},
  year      = {2016},
  pages     = {83--129},
  doi       = {10.1108/S0147-912120160000043012}
}

@inproceedings{NMBN11,
  author    = {Nock, R. and Magdalou, B. and Briys, E. and Nielsen, F.},
  title     = {On Tracking Portfolios with Certainty Equivalents on a Generalization of {M}arkowitz Model: The Fool, the Wise and the Adaptive},
  booktitle = {Proceedings of the 28th International Conference on Machine Learning (ICML)},
  year      = {2011},
  publisher = {Omnipress},
  address = {Bellevue, WA},
  pages     = {73-80}
}

@incollection{Wh89,
  author    = {Whitmore, G.A.},
  title     = {Stochastic Dominance for the Class of Completely Monotonic Utility Functions},
  booktitle = {Studies in the Economics of Uncertainty},
  editor    = {Fomby, Thomas B. and Seo, Tae Kun},
  publisher = {Springer},
  address   = {New York, NY},
  year      = {1989},
  pages     = {77-88}
}

@unpublished{Carroll25,
  author    = {Carroll, G.},
  title     = {Is Equal Opportunity Different from Welfarism?},
  year      = {2025},
  note      = {{W}orking Paper, Yale Department of Economics}
}

@article{Eshaghnia2024,
  author      = {Eshaghnia, Sadegh S.M. and Heckman, James J. and
                 Landers{\o}, Rasmus and Qureshi, Rafeh},
  title       = {Intergenerational Transmission of Family Influence},
  journal = {National Bureau of Economic Research Working Paper},
  volume      = {33023},
  year        = {2024},
  doi         = {10.3386/w33023}
}

@unpublished{GPS09,
  author    = {Grant, S. and Polak, B. and Strzalecki, T.},
  title     = {Second-Order Expected Utility},
  year      = {2009},
  note      = {{Working Paper, SSRN}}
}

@phdthesis{vdG93,
  author  = {van de Gaer, D.},
  title   = {Equality of Opportunity and Investment in Human Capital},
  school  = {KU Leuven},
  address = {Leuven},
  year    = {1993}
}

@article{Aczel2005,
  title={Extension of a generalized {Pexider} equation},
  author={Acz{\'e}l, J{\'a}nos},
  journal={Proceedings of the American Mathematical Society},
  volume={133},
  number={11},
  pages={3227-3234},
  year={2005}
}

@article{AFGK26,
  author  = {Andreoli, F. and Faure, M. and Gravel, N. and Kundu, T.},
  title   = {Evaluating Allocations of Opportunities},
  journal = IER,
  year    = {2026},
  volume  = {67},
  number ={1},
  pages   = {365-397}
}

@article{Adler25,
  author  = {Adler, M.D. and Bossert, W. and Cato, S. and Kamaga, K.},
  title   = {Ex Post Approaches to Prioritarianism and Sufficientarianism},
  journal = TE,
  year    = {2025},
  volume  = {20},
  number = {4},
  pages   = {1367-1410}
}

@article{AA63,
  author  = {Anscombe, Francis J. and Aumann, Robert J.},
  title   = {A Definition of Subjective Probability},
  journal = {The Annals of Mathematical Statistics},
  volume  = {34},
  number  = {1},
  pages   = {199-205},
  year    = {1963}
}

@article{B23,
  author  = {Biasi, B.},
  title   = {School finance equalization increases intergenerational mobility},
  journal = {Journal of Labor Economics},
  year    = {2023},
  volume  = {41},
  number = {1},
  pages   = {1-38}
}

@article{BD82,
  author  = {Blackorby, C. and Donaldson, D.},
  title   = {Ratio-Scale and Translation-Scale Full Interpersonal Comparability Without Domain Restrictions: Admissible Social-Evaluation Functions},
  journal = IER,
  year    = {1982},
  volume  = {23},
  number = {2},
  pages   = {249-268}
}

@article{BG87,
  author  = {Brockett, P.L. and Golden, L.L.},
  title   = {A Class of Utility Functions Containing All the Common Utility Functions},
  journal = MS,
  year    = {1987},
  volume  = {33},
  number ={8},
  pages   = {955-964}
}

@article{CVMMM11,
  author  = {Cerreia-Vioglio, Simone and Maccheroni, Fabio and Marinacci, Massimo and Montrucchio, Luigi},
  title   = {Uncertainty Averse Preferences},
  journal = JET,
  year    = {2011},
  volume  = {146},
  number  = {4},
  pages   = {1275-1330}
}

@article{Carneiro2021,
  author  = {Carneiro, Pedro and {López García}, Italo and Salvanes, Kjell G. and Tominey, Emma},
  title   = {Intergenerational Mobility and the Timing of Parental Income},
  journal = JPoE,
  year    = {2021},
  volume  = {129},
  number = {3},
  pages   = {757-788}
}

@article{Chetty2014,
  author  = {Chetty, Raj and Hendren, Nathaniel and Kline, Patrick and Saez, Emmanuel},
  title   = {{Where is the land of opportunity? The geography of intergenerational mobility in the United States}},
  journal = QJE,
  year    = {2014},
  number = {4},
  volume  = {129},
  pages   = {1553-1623}
}

@article{Chetty2026,
  author  = {Chetty, Raj and Friedman, John N and Hendren, Nathaniel and Jones, Maggie R and Porter, Sonya R},
  title   = {{The Opportunity Atlas: Mapping the Childhood Roots of Social Mobility}},
  journal = AER,
  year    = {2026},
  volume  = {116},
  number = {1},
  pages   = {1-51}
}

@article{Derenoncourt2022,
  author  = {Derenoncourt, Ellora},
  title   = {Can You Move to Opportunity? {Evidence} from the Great Migration},
  journal = AER,
  year    = {2022},
  volume  = {112},
  number = {2},
  pages   = {369-408}
}

@article{EG09,
  author  = {Ergin, H. and Gul, F.},
  title   = {A Theory of Subjective Compound Lotteries},
  journal = JET,
  year    = {2009},
  volume  = {144},
  number ={3},
  pages   = {899-929}
}

@article{FP13,
  author  = {Fleurbaey, M. and Peragine, V.},
  title   = {Ex Ante versus Ex Post Equality of Opportunity},
  journal = {Economica},
  year    = {2013},
  volume  = {80},
    number ={317},
  pages   = {118-130}
}

@article{Fleurbaey10,
  author  = {Fleurbaey, M.},
  title   = {Assessing Risky Social Situations},
  journal = JPoE,
  year    = {2010},
  volume  = {118},
  number = {4},
  pages   = {649-680}
}

@article{GKPS10,
  author  = {Grant, S. and Kajii, A. and Polak, B. and Safra, Z.},
  title   = {{Generalized Utilitarianism and Harsanyi Impartial Observer Theorem}},
  journal = EK,
  year    = {2010},
  volume  = {78},
    number ={6},
  pages   = {1939-1971}
}

@article{GKPS12,
  author  = {Grant, S. and Kajii, A. and Polak, B. and Safra, Z.},
  title   = {Equally-Distributed Equivalent Utility, Ex Post
Egalitarianism and Utilitarianism},
  journal = JET,
  year    = {2012},
  volume  = {147},
  number ={4},
  pages   = {1545-1571}
}

@article{GS89,
  author  = {Gilboa, I. and Schmeidler, D.},
  title   = {Maxmin Expected Utility with Non-Unique Prior},
  journal = JME,
  year    = {1989},
  volume  = {18},
    number ={2},
  pages   = {141-153}
}

@article{Gallagher2019,
  author  = {Gallagher, Ryan and Kaestner, Robert and Persky, Joseph},
  title   = {The geography of family differences and intergenerational mobility},
  journal = {Journal of Economic Geography},
  year    = {2019},
  volume  = {19},
  number ={3},
  pages   = {589-618}
}

@article{GM02,
  author  = {Ghirardato, P. and Marinacci, M.},
  title   = {Ambiguity Made Precise: A Comparable Foundation},
  journal = JET,
  year    = {2002},
  volume  = {102},
    number ={2},
  pages   = {251-289}
}

@article{HKP22,
  author  = {Hufe, P. and Kanbur, R. and Peichl, A.},
  title   = {Measuring Unfair Inequality: Reconciling Equality of Opportunity and Freedom from Poverty},
  journal = RES,
  year    = {2022},
  volume  = {89},
  number ={6},
  pages   = {3345-3380}
}

@article{HS01,
  author  = {Hansen, L.P. and Sargent, T.},
  title   = {Robust Control and Model Uncertainty},
  journal = AER,
  year    = {2001},
  volume  = {91},
  number ={2},
  pages   = {60-66}
}

@article{Hartley2022,
  author  = {Hartley, Robert Paul and Lamarche, Carlos and Ziliak, James P.},
  title   = {Welfare Reform and the Intergenerational Transmission of Dependence},
  journal = JPoE,
  year    = {2022},
  volume  = {130},
  number ={3},
  pages   = {523-565}
}

@article{Hufe2025,
  author  = {Hufe, Paul and Kobus, Martyna and Peichl, Andreas and Schüle, Paul},
  title   = {{Multidimensional Equality of Opportunity in the United States}},
  journal = EER,
  year    = {2025},
  volume = {180},
  pages = {105162}
}

@article{JKN25,
  author  = {Jácome, E. and Kuziemko, I. and Naidu, S.},
  title   = {Mobility for All: Representative Intergenerational Mobility Estimates over the Twentieth Century},
  journal = JPoE,
  year    = {2025},
  volume  = {133},
  number = {1},
  pages   = {306-354}
}

@article{KMM05,
  author  = {Klibanoff, P. and Marinacci, M. and Mukerji, S.},
  title   = {A Smooth Model of Decision Making Under Ambiguity},
  journal = EK,
  year    = {2005},
  volume  = {73},
  number ={6},
  pages   = {1849-1892}
}

@article{KP78,
  author  = {Kreps, D.M. and Porteus, E.L.},
  title   = {Temporal Resolution of Uncertainty and Dynamic Choice Theory},
  journal = EK,
  year    = {1978},
  volume  = {46},
  number ={1},
  pages   = {185-200}
}

@article{MMR06,
  author  = {Maccheroni, F. and Marinacci, M. and Rustichini, A.},
  title   = {Ambiguity Aversion, Robustness, and the Variational Representation of Preferences},
  journal = EK,
  year    = {2006},
  volume  = {74},
  number ={6},
  pages   = {1447-1498}
}

@article{MMR13,
  author  = {Maccheroni, F. and Marinacci, M. and Ruffino, D.},
  title   = {Alpha as Ambiguity: Robust Mean-Variance Portfolio Analysis},
  journal = EK,
  year    = {2013},
  volume  = {81},
  number = {3},
  pages   = {1075-113}
}

@article{MN11,
  author  = {Magdalou, B. and Nock, R.},
  title   = {Income Distributions and Decomposable Divergence Measures},
  journal = JET,
  year    = {2011},
  volume  = {146},
  number ={6},
  pages   = {2440-2454}
}

@article{MRSW24,
  author  = {Mogstad, M. and Romano, J. P. and Shaikh, A. M. and Wilhelm, D.},
  title   = {Inference for ranks with applications to mobility across neighbourhoods and academic achievement across countries},
  journal = RES,
  year    = {2024},
  volume  = {91},
  number = {1},
  pages   = {476-518}
}

@article{Mazumder2018,
  author  = {Mazumder, Bhashkar},
  title   = {Intergenerational Mobility in the {United States}: What We Have Learned from the {PSID}},
  journal = {The ANNALS of the American Academy of Political and Social Science},
  year    = {2018},
  volume  = {680},
  number = {1},
  pages   = {213-234}
}

@article{Mello2025,
  author  = {Mello, Ursula and Nybom, Martin and Stuhler, Jan},
  title   = {A Life-Cycle Estimator of Intergenerational Income Mobility},
  journal = {Review of Economics and Statistics},
  year    = {forthcoming}
}

@article{Moramarco2024,
  author      = {Fleurbaey, Marc and Moramarco, Domenico and Peragine, Vito},
  title       = {Measuring Inequality and Welfare When Some Inequalities
                 Matter More than Others},
  journal        = {SERIES Working Paper},
  volume      = {03/2024},
  year        = {2024}
}

@article{P64,
  author  = {Pratt, J.W.},
  title   = {Risk Aversion in the Small and in the Large},
  journal = EK,
  year    = {1964},
  volume  = {32},
  number = {1/2},
  pages   = {122-136}
}

@article{Piketty2018,
  author  = {Piketty, Thomas and Saez, Emmanuel and Zucman, Gabriel},
  title   = {{Distributional national accounts: Methods and estimates for the United States}},
  journal = QJE,
  year    = {2018},
  volume  = {133},
  number = {2},
  pages   = {553-609}
}

@article{RT16,
  author  = {Roemer, J.E. and Trannoy, A.},
  title   = {Equality of Opportunity: Theory and Measurement},
  journal = JEL,
  year    = {2016},
  volume  = {54},
  number ={4},
  pages   = {1288-1332}
}

@article{Ramos2016,
  author  = {Ramos, Xavier and {Van de gaer}, Dirk},
  title   = {Approaches to Inequality of Opportunity: Principles, Measures and Evidence},
  journal = {Journal of Economic Surveys},
  year    = {2016},
  volume  = {30},
  number = {5},
  pages   = {855-883}
}

@article{Schmeidler89,
  author  = {Schmeidler, D.},
  title   = {Subjective Probability and Expected Utility Without Additivity},
  journal = EK,
  year    = {1989},
  volume  = {57},
  number = {3},
  pages   = {571-587}
}

@article{Solon1992,
  author  = {Solon, Gary},
  title   = {{Intergenerational Income Mobility in the United States}},
  journal = AER,
  year    = {1992},
  volume  = {82},
  number = {3},
  pages   = {393-408}
}

@article{St11,
  author  = {Strzalecki, T.},
  title   = {Axiomatic Foundations of Multiplier Preferences},
  journal = EK,
  year    = {2011},
  volume  = {79},
  number ={1},
  pages   = {47-73}
}

@article{Vosters2017,
  author  = {Vosters, Kelly and Nybom, Martin},
  title   = {{Intergenerational Persistence in Latent Socioeconomic Status: Evidence from Sweden and the United States}},
  journal = {Journal of Labor Economics},
  year    = {2017},
  number ={3},
  volume  = {35},
  pages   = {869-901}
}

@article{Wa88,
  author  = {Wakker, P.},
  title   = {The Algebraic versus the Topological Approach to Additive Representations},
  journal = JMP,
  year    = {1988},
  volume  = {32},
  number = {4},
  pages   = {421-435}
}

@article{Ward2023,
  author  = {Ward, Zachary},
  title   = {Intergenerational Mobility in American History: Accounting for Race and Measurement Error},
  journal = AER,
  year    = {2023},
  volume  = {113},
  number = {2},
  pages   = {3213-3248}
}

@article{davis2026decline,
  author  = {Davis, Jonathan MV and Mazumder, Bhashkar},
  title   = {The Decline in Intergenerational Mobility after 1980},
  journal = {Review of Economics and Statistics},
  year    = {2026},
  volume = {108},
  number = {1},
  pages = {1--15}
}

@article{Nybom2024,
  author    = {Nybom, Martin and Stuhler, Jan},
  title     = {Interpreting Trends in Intergenerational Mobility},
  journal   = {Journal of Political Economy},
  volume    = {132},
  number    = {8},
    pages     = {2531-2570},
  year      = {2024},
  publisher = {University of Chicago Press},
  doi       = {10.1086/729582},
  url       = {https://doi.org/10.1086/729582}
}

@article{CortesOrihuela2025,
  author    = {Cortes Orihuela, Javier and D{\'i}az, Juan D. and Guti{\'e}rrez Cubillos, Pablo and Troncoso, Pablo A.},
  title     = {The two-sample two-stage least squares method to estimate the intergenerational earnings elasticity},
  journal   = {Journal of Economic Inequality},
  volume    = {23},
  pages     = {549-570},
  year      = {2025},
  publisher = {Springer},
  doi       = {10.1007/s10888-024-09643-8},
  url       = {https://doi.org/10.1007/s10888-024-09643-8}
}

@article{Mazumder2023,
Author = {Deutscher, Nathan and Mazumder, Bhashkar},
Title = {Measuring Intergenerational Income Mobility: A Synthesis of Approaches},
Journal = {Journal of Economic Literature},
Volume = {61},
Number = {3},
Year = {2023},
Pages = {988--1036},
DOI = {10.1257/jel.20211413},
URL = {https://www.aeaweb.org/articles?id=10.1257/jel.20211413}
}

@article{Boustan2025,
 title = "Intergenerational Mobility of Immigrants in 15 Destination Countries",
 author = "Boustan, Leah and Jensen, Mathias Fjællegaard and Abramitzky, Ran and Jácome, Elisa and Manning, Alan and Pérez, Santiago and Watley, Analysia and Adermon, Adrian and Arellano-Bover, Jaime and Åslund, Olof and Connolly, Marie and Deutscher, Nathan and Gielen, Anne C and Giesing, Yvonne and Govind, Yajna and Halla, Martin and Hangartner, Dominik and Jiang, Yuyan and Karmel, Cecilia and Landaud, Fanny and Macmillan, Lindsey and Martinez, Isabel and Polo, Alberto and Poutvaara, Panu and Rapoport, Hillel and Roman, Sara and Salvanes, Kjell G and San, Shmuel and Siegenthaler, Michael and Sirugue, Louis and Soria Espin, Javier and Stuhler, Jan and Violante, Giovanni L and Webbink, Dinand and Weber, Andrea and Zhang, Jonathan and Zheng, Angela and Zohar, Tom",
  journal = {National Bureau of Economic Research Working Paper},
  volume      = {33558},
  year        = {2025},
 doi = {10.3386/w33558},
}

@Article{Lundberg2016,
    author = "Lundberg, Shelly and Pollak, Robert A. and Stearns, Jenna",
    title = "Family Inequality: Diverging Patterns in Marriage, Cohabitation, and Childbearing",
    journal = "Journal of Economic Perspectives",
    year = "2016",
    volume = "30",
    number = "2",
    pages = "79-102",
}

@Article{Britto2026,
    author = "Britto, Diogo and Fonseca, Alexandre and Pinotti, Paolo and Sampaio, Breno and Warwar, Lucas",
    title = "Intergenerational Mobility in the Land of Inequality",
    journal = "Review of Economics and Statistics",
    year = "forthcoming",
}

@article{Adermon2026,
author = {Adermon, Adrian and Brandén, Gunnar and Nybom, Martin},
title = {The relationship between intergenerational mobility and equality of opportunity},
journal = SJE,
year = {forthcoming},
doi = {https://doi.org/10.1111/sjoe.70033},
eprint = {https://onlinelibrary.wiley.com/doi/pdf/10.1111/sjoe.70033},
}

@article{Jantti2020,
title = {Intergenerational mobility, intergenerational effects, sibling correlations, and equality of opportunity: A comparison of four approaches},
journal = {Research in Social Stratification and Mobility},
volume = {70},
pages = {100455},
year = {2020},
issn = {0276-5624},
doi = {https://doi.org/10.1016/j.rssm.2019.100455},
url = {https://www.sciencedirect.com/science/article/pii/S0276562419301544},
author = {Anders Björklund and Markus Jäntti},
}

@Article{Alesina2018,
    author = "Alesina, Alberto and Stantcheva, Stefanie and Teso, Edoardo",
    title = "Intergenerational Mobility and Preferences for Redistribution",
    journal = "American Economic Review",
    year = "2018",
    volume = "108",
    number = "2",
    pages = "521-554",
}

@Article{Brewer2017,
    author = "Brewer, Mike and Etheridge, Ben and O'Dea, Cormac",
    title = "{Why are Households that Report the Lowest Incomes So Well-off?}",
    journal = "Economic Journal",
    year = "2017",
    volume = "127",
    number = "605",
    pages = "F24-F49",
    doi = "10.1111/ecoj.12334",
}

@Article{Chetty2014b,
    author = "Chetty, Raj and Hendren, Nathaniel and Kline, Patrick and Saez, Emmanuel and Turner, Nicholas",
    title = "{Is the United States Still a Land of Opportunity? Recent Trends in Intergenerational Mobility}",
    journal = "{American Economic Review}",
    year = "2014",
    volume = "104",
    number = "5",
    pages = "141--47",
}

\newpage
\appendix
\pagenumbering{arabic}
\renewcommand\thefigure{A.\arabic{figure}}
\renewcommand\thetable{A.\arabic{table}}
\setcounter{section}{0}
\setcounter{figure}{0}
\setcounter{table}{0}

\clearpage
\section{Proofs}
\label{app: proofs}

\parhead{Proof of Theorem~\ref{theo-general}} Assume that the axioms stated in the theorem are satisfied. The set $\Delta(\incomes)$, endowed with the weak topology, forms a convex mixture space, which ensures it is a connected topological space. Consequently, the space of opportunity profiles $\SocietiesFixed = \Delta(\incomes)^{|\states|}$, endowed with the product topology, is also a connected space. 

Then, we say that a type $s \in \mathcal{S}$ is essential for the relation $\succeq$ if there exists $f \in \SocietiesFixed$ and two conditional distributions $\pi_s, \pi'_s \in \Delta(\incomes)$ such that $(f_{-s}, \pi_s) \succ (f_{-s}, \pi'_s)$. By assumption, $\succeq$ is non-degenerate, meaning that $f \succ f'$ for some $f, f' \in \SocietiesFixed$. This implies there exist at least two income distributions  $\pi, \pi' \in \Delta(\incomes)$ such that $f_{\pi} \succ f_{\pi'}$. Indeed, if $f_{\pi} \sim f_{\pi'}$ held for all $\pi, \pi' \in \Delta(\incomes)$, the weak part of A\ref{ax-Pareto} would force $f \sim f'$ for all $f, f' \in \SocietiesFixed$, contradicting non-degeneracy. Now consider these two distributions $\pi, \pi' \in \Delta(\incomes)$, a state $s \in \states$ and two profiles $f, f' \in \SocietiesFixed$ such that, respectively, $f'(t) = \pi'_t = \pi'$ for all $t \in \states$, and $f(t) = \pi_t = \pi'$ for all $t \neq s$ and $f(s) = \pi_s = \pi$. Because $f_{\pi} \succ f_{\pi'}$ and $\pi_t = \pi'_t$ for all $t \neq s$, A\ref{ax-Pareto} implies that $f =(f_{-s}, \pi_s) \succ (f_{-s}, \pi'_s) = f'$. Since this holds for any $s \in \states$, all types are essential. Because it is assumed that $|\states| \geq 3$, at least three types are essential. The relation $\Wpref$ is a continuous weak order. By adding A\ref{ax-ICCD}, the prerequisites of Theorem~4.1 in \cite{Wa88} are satisfied, hence there exist a list of continuous functions $H_s : \Delta(\incomes) \rightarrow \real$ for all $s \in \states$, defined up to an increasing affine transformation,\footnote{The uniqueness result for additive representations on connected spaces applies here \cite[see][Theorem~3]{De60}.} such that $\Wpref$ can be represented by the function $W(f) = \sum_{s \in \states} H_s(\pi_s)$.

Now, we prove that $H_s(\pi_s) = q(s)H(\pi_s)$, under A\ref{ax-anonymity}. Consider any $s, s' \in \states$ and any $f \in \SocietiesFixed$. Let $f'$ be the profile obtained by permuting the conditional lotteries of types $s$ and $s'$, while leaving all other types unchanged. From A\ref{ax-anonymity}, we have $f \sim f'$. Applying the additive representation, we deduce that $H_{s}(\pi_s) + H_{s'}(\pi_{s'}) = H_{s}(\pi_{s'}) + H_{s'}(\pi_s)$. Rearranging yields $H_{s}(\pi_s) - H_{s'}(\pi_s) = H_{s}(\pi_{s'}) - H_{s'}(\pi_{s'})$. Since this must hold for any arbitrary conditional distributions, the difference $H_s - H_{s'}$ must be a constant on $\Delta(\incomes)$. Because the uniqueness theorem for additive representations allows each function $H_s$ to be shifted by an independent constant without altering the overall preference ranking, we can calibrate these type-specific constants to absorb this difference. This normalization establishes $H_{s} = H_{s'} = \tilde{H}$ for all $s, s' \in \states$, so that $W(f) = \sum_{s \in \states} \tilde{H}(\pi_s)$. By assumption, $q(s) = 1/ |\states|$ for all $s \in \states$. Since $W$ is unique up to a positive affine transformation, we can multiply the entire evaluation by the positive constant $1/|\states|$ without altering the represented preferences. By defining $H(\pi) = |\states| \cdot \tilde{H}(\pi)$, we deduce that $H_s(\pi_s) = q(s)H(\pi_s)$.

Then, we establish that $H(\pi_s) = \phi \left(U(\pi_s)\right)$, with $U(\pi_s)$ the expected utility of $\pi_s$, and $\phi$ an increasing function. For two distributions $\pi, \pi' \in \Delta(\incomes)$, we define $\pi \Wpref^\star \pi' \Leftrightarrow f_{\pi} \Wpref f_{\pi'}$. By construction, the subspace of equal-opportunity profiles is isomorphic to the space of distributions $\Delta(\incomes)$, with respect to the probability mixture operation. Since the relation $\Wpref$, when restricted to equal-opportunity profiles, satisfies the axioms of weak order, continuity, and independence (A\ref{ax-ICEO}), it follows that the induced relation $\Wpref^\star$ on $\Delta(\incomes)$ satisfies the von Neumann-Morgenstern (vNM) axioms. By the vNM expected utility theorem, there exists an affine function $U : \Delta(\incomes) \rightarrow \real$ that represents~$\Wpref^\star$. Therefore, there exists a (continuous) utility function over incomes $u : \incomes \rightarrow \real$ such that the evaluation of a distribution $\pi \in \Delta(\incomes)$ is given by $U(\pi) = \sum_{y \in \incomes} \pi(y) u(y)$. Thus, $U(\pi) \geq U(\pi') \Leftrightarrow \pi \Wpref^\star \pi' \Leftrightarrow f_{\pi} \Wpref f_{\pi'}$. Moreover, we know from the previous representation that, for any $f \in \SocietiesFixed$, we have $W(f) = \sum_{s \in \states} q(s)H(\pi_s)$. We deduce that $W(f_\pi) = H(\pi)$. It follows that $U(\pi)$ and $H(\pi)$ represent $\Wpref^\star$, hence there exists an increasing function $\phi$ such that $H(\pi) = \phi \left(U(\pi)\right)$. From A\ref{ax-Pareto}, we know that $u$ is non-constant and $\phi$ is strictly increasing. For the sake of consistency with the rest of the paper, we prefer the ordinally equivalent representation in~(\ref{EDEU-general-representation}), rather than $W(f) = \sum_{s \in \states} q(s)\phi \left(U(\pi_s)\right)$. 

The next step is to show that $\phi$ is weakly concave. Consider profiles $f, f' \in \SocietiesFixed$ as defined in A\ref{ax-aversion}. By assuming A\ref{ax-aversion} we have, after simplification (by multiplying by $|\states|$ and dividing by $2$):
\[
f \Wpref f'\ \Longleftrightarrow\ \left( \nicefrac{1}{2} \right) \phi \left(U(\pi_s)\right) + \left( \nicefrac{1}{2} \right) \phi \left(U(\pi_{s'})\right) \geq \left( \nicefrac{1}{2} \right) \phi \left(U(\pi'_s)\right) + \left( \nicefrac{1}{2} \right) \phi \left(U(\pi'_{s'})\right)\,.
\]
By definition, $\pi_s = \pi_{s'} = \left( \nicefrac{1}{2} \right) (\pi'_s + \pi'_{s'})$. Because $U$ is affine, $U(\pi_s) = U(\pi_{s'}) = \left( \nicefrac{1}{2} \right) U(\pi'_s) + \left( \nicefrac{1}{2} \right) U(\pi'_{s'})$. After simplification, one obtains:
\[
f \Wpref f'\ \Longleftrightarrow\ \phi \left(\left( \nicefrac{1}{2} \right) (U(\pi'_s) + U(\pi'_{s'})) \right) \geq \left( \nicefrac{1}{2} \right) \phi \left(U(\pi'_s)\right) + \left( \nicefrac{1}{2} \right) \phi \left(U(\pi'_{s'})\right)\,.
\]
Since $\phi$ is continuous, this midpoint concavity implies global weak concavity of $\phi$.

Finally, we deal with the uniqueness of the functions $u$ and $\phi$. First, $u$ is obtained in a vNM framework, which implies that $u$ and $u'$ represent the same preferences if there exist $a > 0$ and $b \in \real$ such that $u' = a u + b$. We have previously established that $H$, hence also $\phi$, is defined up to an increasing affine transformation. Hence $\phi$ and $\phi'$ represent the same preferences if there exist $A > 0$ and $B \in \real$ such that $\phi'(u') = A \phi(u) + B$. The proof of the consistency of $W_{\phi, u, q}$ with the considered axioms is left to the reader. \qed

\parhead{Proof of Theorem~\ref{theo-main}} Assume that the axioms stated in the theorem are satisfied. Because $\Wpref$ is a continuous weak order, A\ref{ax-ICEO2} is stronger than A\ref{ax-ICEO}, thus A\ref{ax-ICEO} also hold. We know from Theorem~\ref{theo-general} that $f \succeq f'$ implies $\phi^{-1}\left(\sum_{s \in \states}q(s)\phi(U(\pi_{s}))\right) \ge \phi^{-1}\left(\sum_{s \in \states}q(s)\phi(U(\pi'_{s}))\right)$, with $\phi$ a strictly increasing function. Let $\mathcal{U} = U(\Delta(\incomes)) \subseteq \mathbb{R}$ be the range of the expected utility function, which is an interval. Let $u=(u_s)_{s\in\states}$ and $u'=(u'_s)_{s\in\states}$ with, respectively, $u_s=U(\pi_s)$ and $u'_s=U(\pi'_s)$, for all $s\in\states$. We first define a relation $\succeq^*$ over $\mathcal{U}^{|\states|}$ induced by $\succeq$ such that, for all $f, f' \in \mathcal{F}_{\states}$, we have $u \succeq^* u' \iff f \succeq f'$. By construction, $\succeq^*$ is represented by the continuous function $H_{\phi,u,q}(u)=\phi^{-1}(\sum_{s\in\states}q(s)\phi(u_s))$ on $\mathcal{U}^{|\states|}$. Then, define $c = U(\pi)$ and $c' = U(\pi')$. Because $U$ is affine, A\ref{ax-ICEO2} implies that, for all $f,f' \in \SocietiesFixed$, all $\pi, \pi' \in \Delta(\incomes)$ and all $\alpha \in (0,1)$:
\[
\begin{aligned}
& \phi^{-1} \left( \sum_{s \in \states} q(s)\phi \left(\alpha u_s + (1-\alpha) c\right) \right) \geq \phi^{-1} \left( \sum_{s \in \states} q(s)\phi \left(\alpha u'_s + (1-\alpha) c\right) \right)\\
\Longleftrightarrow\quad & \phi^{-1} \left( \sum_{s \in \states} q(s)\phi \left(\alpha u_s + (1-\alpha) c'\right) \right) \geq \phi^{-1} \left( \sum_{s \in \states} q(s)\phi \left(\alpha u'_s + (1-\alpha) c'\right) \right)\,.
\end{aligned}
\]
By letting $1_{|\states|}$ be a vector of one repeated $|\states|$ times, it is equivalent to write:
\[
\alpha u + (1-\alpha) c \cdot 1_{|\states|}\ \Wpref\!^\star\ \alpha u' + (1-\alpha) c \cdot 1_{|\states|}\ \Longleftrightarrow\ \alpha u + (1-\alpha) c' \cdot 1_{|\states|}\ \Wpref\!^\star\ \alpha u' + (1-\alpha) c' \cdot 1_{|\states|}\,.
\]
Let $v = \alpha u + (1-\alpha) c \cdot 1_{|\states|}$, $v' = \alpha u' + (1-\alpha) c \cdot 1_{|\states|}$, and $k = (1-\alpha)(c'-c)$. After simplification: 
\[
v\ \Wpref\!^\star\ v'\ \Longleftrightarrow\ v + k \cdot 1_{|\states|}\ \Wpref\!^\star\ v' + k \cdot 1_{|\states|}\,.
\]
By continuity, this translation invariance extends globally to all $u,u' \in \mathcal{U}^{|\states|}$ and any admissible translation $k$. From Theorem 3, page 257 in \cite{BD82}, we know that $\Wpref\!^\star$ can be represented by the function $H_{\phi_\theta,u,q}(u) = \phi^{-1}_\theta \left(  \sum_{s \in \states} q(s) \phi_\theta (u_s) \right)$, with $\phi_\theta$ as defined in~(\ref{expo-transfo}). Note that, in \cite{BD82}, $\Wpref\!^\star$ compares elements of 
$\real^{|\states|}$, and not elements of $\mathcal{U}^{|\states|} \subset \real^{|\states|}$ as we do. Drawing on the results established in \cite{A66}, Section 2.1.4, this representation remains valid on our restricted domain, given that $\mathcal{U}$ is an interval. Then, the restriction on the domain of $\theta$ results from the concavity of $\phi_\theta$ (and A\ref{ax-aversion}).  We know that the function $u$ can be replaced by $u' = a u + b$, with $a > 0$ and $b \in \real$.
If $H_{\phi_\theta,u,q}$ and $H_{\phi_{\theta'}, u', q}$ represent the same preferences, we have $H_{\phi_{\theta'}, u', q} = \phi^{-1}_{\theta'} \left(  \sum_{s \in \states} q(s) \left( - e^{- \theta'( a u_s + b)} \right) \right) = \phi^{-1}_{\theta'} \left( e^{- \theta' b} \sum_{s \in \states} q(s) \left( - e^{- \theta'a u_s} \right) \right)$. Therefore, we must also have $\theta' = \theta / a$. To sum up, $\Wpref$ can be represented by the function $W_{\phi_\theta, u, q}(f) = \phi^{-1}_\theta \left(  \sum_{s \in \states} q(s) \phi_\theta \left( U(\pi_s) \right) \right)$. All the axioms considered here are satisfied by this representation. \qed

\parhead{Proof of Theorem~\ref{theo-extended}} Assume that the axioms stated in the theorem are satisfied. Because the set of potential individuals $\labels$ is countably infinite and populations $\pop$ are finite, the set of all possible partitions $\types$ is countable. The extended domain $\domain$ is thus a countable union of separable spaces. Relation $\succeq^*$ is assumed to be a continuous weak order, hence representable by a continuous function $H:\domain\rightarrow\real$, so that $(f,q)\succeq^*(f',q')\iff H(f,q)\ge H(f',q')$. In Theorem~\ref{theo-general}, we established that the local preference $\Wpref_\states$ on a fixed partition $\states$ of types with $|\states| \ge 3$, is represented by $W_{|\states|}(f)=\phi_{|\states|}^{-1}\left(\sum_{s\in\states} \overline{q}(s) \phi_{|\states|}(U_{|\states|}(\pi_{s}))\right)$, with $\overline{q}(s) = 1 / |\states|$ for all $s \in \states$. Note that this representation depends only on the partition size $|\states|$ and not on the specific partition identity. Indeed, within any given partition, A\ref{ax-anonymity} ensures all types are treated symmetrically. Furthermore, across different partitions, A\ref{ax-split} allows any two uniform societies of the same size and with identical conditional income distributions to be transformed into one another via type-splitting and merging, guaranteeing their ordinal equivalence in $\domain$.

To establish that the cardinal representation $(u_N, \phi_N)$ is independent of partition size and population scale, consider two uniformly distributed societies $(f,q)$ and $(f',q')$ defined on partitions of sizes $N$ and $N'$, respectively. Because the distributions are uniform, each type in the first society contains an identical number $c \in \Ninteger$ of individuals, and each type in the second society contains $c' \in \Ninteger$ individuals. Because the set of potential individuals $\labels$ is countably infinite, we can use A\ref{ax-pop} to replicate the first society by a factor of $c' N'$, and the second society by a factor of $c N$. Both replicated societies now share the exact same total population size, $c c' N N'$. Then, we apply A\ref{ax-split} to split each of the $N$ types of the first replicated society into $N'$ identical subtypes, and each of the $N'$ types of the second replicated society into $N$ identical subtypes. Both transformed societies now consist of exactly $N N'$ uniform types, with each type containing exactly $c c'$ individuals. Because these two resulting societies share the exact same demographic structure, a type-preserving bijection can be drawn between their partitions. By A\ref{ax-pop} (applied as a $1$-fold replication), evaluating a profile on one partition is strictly equivalent to evaluating it on the other. These profiles can be treated as operating on a common domain $\mathcal F_{\states^\star}$ with $|\states^\star| = N N'$. By the uniqueness result of Theorem~\ref{theo-general} on this common domain, $(u_N, \phi_N)$ and $(u_{N'}, \phi_{N'})$ must both be positive affine transformations of this refined representation $(u_{N N'}, \phi_{N N'})$, yielding a unique pair $(u, \phi)$. Thus,  a local preference $\Wpref_\states$ on a fixed partition $\states$ is represented by $W(f)=\phi^{-1}\left(\sum_{s\in\states} \overline{q}(s) \phi(U(\pi_{s}))\right)$. A\ref{ax-cons} implies, for all $f, f' \in \SocietiesFixed$, that $(f,\overline{q}) \Wpref^\star (f',\overline{q})\  \Longleftrightarrow\ f \Wpref f'$. We can therefore calibrate $H$ globally such that, for any uniformly distributed society $(f,\overline{q}) \in \domain$ of arbitrary size, we have:
\begin{equation} \label{eq-extended}
  H(f,\overline{q}) = \phi^{-1} \left( \sum_{s \in \states} \overline{q}(s) \phi \left( U(\pi_s) \right) \right)\,.  
\end{equation}

Consider now an arbitrary society $(f,q) \in \domain$ defined on a partition $\states$, where the distribution $q \in \Delta(\states)$ consists of strictly positive rational probabilities. Let $K \in \Ninteger$ be a common denominator for all $q(s)$ such that for each $s \in \states$, we can write $q(s) = k_s / K$ with $k_s \in \Ninteger$. To evaluate this society, we first apply A\ref{ax-pop} to replicate the entire population by a factor $K$, ensuring that each type's absolute size becomes a multiple of $K$. Then, by iterative applications of A\ref{ax-split}, we can partition each type $s \in \states$ into $k_s$ distinct subtypes. Each new subtype is assigned a demographic weight of $1/K$ and receives the exact same conditional income distribution $\pi_s$ as the original type $s$. This procedure maps the original society $(f,q)$ to a new society $(f', \overline{q}') \in \domain$ defined on a finer partition $\states'$ of size $|\states'| = K$. The distribution $\overline{q}'$ is uniform over $\states'$, so that $\overline{q}'(s) = 1/K$ for all $s \in \states'$. A\ref{ax-split} guarantees that this procedure leaves the extended evaluation unchanged. Hence $(f,q) \sim^\star (f',\overline{q}')$, which implies $H(f,q) = H(f',\overline{q}')$. Since $\overline{q}'$ is a uniform society, we evaluate it using the representation in~(\ref{eq-extended}):
\begin{equation} \label{eq-extended2}
  H(f',\overline{q}') = \phi^{-1} \left( \sum_{s \in \states'} \frac{1}{K} \phi \left( U(\pi_s) \right) \right)\,.  
\end{equation}
Grouping the terms of this sum back according to the original types $s \in \states$ (since each original type $s$ was split into $k_s$ identical sub-types), the sum simplifies to:
\begin{equation} \label{eq-extended3}
  H(f,q) = \phi^{-1} \left( \sum_{s \in \states} \frac{k_s}{K} \phi \left( U(\pi_s) \right) \right) = \phi^{-1} \left( \sum_{s \in \states} q(s) \phi \left( U(\pi_s) \right) \right)\,.  
\end{equation}

Finally, the set of strictly positive rational distributions is dense in the probability simplex $\Delta(\states)$. Since the extended preference relation $\Wpref^\star$ is a continuous weak order, the representation $H(f,q)$ must be continuous with respect to $q$. Therefore, the functional form derived in~(\ref{eq-extended3}) extends to all real-valued $q \in \Delta(\states)$, including boundary points where terms with $q(s)=0$ trivially vanish. In conclusion, we immediately see that $W_{\phi,u}(f,q)$ satisfies A\ref{ax-cons}--A\ref{ax-split}. \qed

\parhead{Proof of Theorem~\ref{theo-main2}} When $\theta = 0$, the expression of $V_{\theta, u}(f,q)$ in~(\ref{HS-representation}) is immediately deduced from~(\ref{expo-transfo}). Now assume that $\theta > 0$. Let $H_{\theta, u}(f,q,p) = \sum_{s \in \states} p(s) U(\pi_s) + \frac{1}{\theta} D_{KL}(p || q)$, with $p \in \Delta^q(\states)$. We first prove that $p^\star_{\theta, u}(f,q;\cdot) = \argmin_{p \in \Delta^q(\states)} H_{\theta, u}(f,q,p)$ can be written as in~(\ref{opt-p}). We then establish that $H_{\theta, u}(f,q,p^\star_{\theta, u}(f,q;\cdot)) = V_{\theta, u}(f,q)$.

Because $p \in \Delta^q(\states)$, we have $p^\star_{\theta, u}(f,q;s) = 0$ as soon as $q(s)=0$. Now, denote by $\overline{\states} \subseteq \states$ the set of $s$ such that $q(s)>0$. Let us assume that the minimum is an interior solution on~$\overline{\states}$, meaning $p(s) > 0$ for all $s \in \overline{\states}$. The Lagrangian of this minimization problem, with $\nu = {(\nu_s)}_{s \in \overline{\states}}$, can be written as:
\[
L(p, \lambda, \nu) = \sum_{s \in \overline{\states}}  p(s) \left[ U(\pi_s) + \frac{1}{\theta} \ln{\frac{p(s)}{q(s)}} \right] + \lambda \left[ \sum_{s \in \overline{\states}}p(s) - 1 \right] - \sum_{s \in \overline{\states}}\nu_s p(s)\,.
\]
For all $s \in \overline{\states}$, the first Karush-Kuhn-Tucker (KKT) optimality condition needs $\nu_s p(s) = 0$, hence $\nu_s = 0$. Then, the stationarity condition $\frac{\partial L(p, \lambda, \nu)}{\partial p(s)} = 0$ leads to:
\begin{equation} \label{KKT-1}
   p^\star_{\theta, u}(f,q;s) = q(s) \exp \left(\theta \left( - U(\pi_s) - \lambda \right) - 1\right)\,. 
\end{equation}
Because $q(s)>0$, the right hand-side is positive, hence $p^\star_{\theta, u}(f,q;s) >0$ for all $s \in \overline{\states}$. Moreover, $\sum_{t \in \overline{\states}} p^\star_{\theta, u}(f,q;t) = 1$, which is equivalent, by using~(\ref{KKT-1}), to:
\begin{equation} \label{KKT-2}
   \exp \left(- \theta \lambda - 1\right) = \frac{1}{\sum_{t \in \overline{\states}} q(t) \exp \left(- \theta U(\pi_t) \right)}\,. 
\end{equation}
One also knows from~(\ref{KKT-1}) that $p^\star_{\theta, u}(f,q;s) = q(s) \exp \left(- \theta U(\pi_s)\right) \exp \left(- \theta\lambda - 1\right)$. Hence, $p^\star_{\theta, u}(f,q;s)$ as defined in~(\ref{opt-p}) is obtained by substituting~(\ref{KKT-2}) into~(\ref{KKT-1}). The other KKT conditions are satisfied. We verify that $p^\star_{\theta, u}(f,q;s) > 0$, which is consistent with our initial interiority assumption. Finally, because the objective function is strictly convex on the probability simplex, the global minimum is unique.

Let's now prove that $H_{\theta, u}(f,q,p^\star_{\theta, u}(f,q;\cdot)) = V_{\theta, u}(f,q)$. By using $p^\star_{\theta, u}(f,q;s)$, we have: 
\begin{align*}
H_{\theta, u}(f,q,p^\star_{\theta, u}(f,q;\cdot)) &= \sum_{s \in \states}  p^\star_{\theta, u}(f,q;s) \left[ U(\pi_s) + \frac{1}{\theta} \ln \left( \frac{\exp \left(- \theta U(\pi_s) \right)}{\sum_{t \in \states} q(t) \exp \left(- \theta U(\pi_t) \right)} \right) \right] \\
&= - \frac{1}{\theta} \ln \left( - \sum_{t \in \states} q(t) \phi_\theta \left( U(\pi_t) \right) \right) \times \sum_{s \in \states}  p^\star_{\theta, u}(f,q;s)\,.
\end{align*}
The result is deduced from $\sum_{s \in \states}  p^\star_{\theta, u}(f,q;s) = 1$, and $\phi^{-1}_\theta(t) = - (\nicefrac{1}{\theta}) \ln (-t)$ when $\theta > 0$. \qed

\parhead{Proof of Theorem~\ref{theo-main3}} When $\theta = 0$, the expression of $V_{\theta, u}(f,q) = \expect_q [U]$ is immediately deduced from~(\ref{expo-transfo}). For $\theta > 0$, we first need a definition. The `cumulant-generating function' of a real-valued random variable $X$ can be written, for any $\tau \in \real$, by $K_X(\tau) = \ln \left( \expect [e^{\tau X}] \right)$. By definition, such a function is convex. If we denote by $K_X^{(n)}$ the n-th derivative of $K_X$, a Taylor expansion of order $m$ around zero yields, as $\tau$ tends to zero, $K_X(\tau) = \sum_{n=1}^{m} K_X^{(n)}(0) \frac{\tau^n}{n !} + \mathcal{O}(\tau^{m+1})$.  
The cumulants are sufficient statistics to characterize the full distribution of the random variable $X$, with the first two cumulants $K_X^{(1)}(0) = \expect [X]$ and $K_X^{(2)}(0) = \var(X)$. 

If we let $U = {(U(\pi_s))}_{s \in \states}$ be the list of expected utilities of all types in $f \in \SocietiesFixed$, distributed according to $q = {(q(s))}_{s \in \states}$, we have $V_{\theta, u}(f,q) = \left( - \nicefrac{1}{\theta} \right) \ln \left( \expect_q \left[ e^{- \theta U} \right] \right)$ (see Theorem~\ref{theo-main}). Therefore, we can write $V_{\theta, u}(f,q) = (-\nicefrac{1}{\theta}) K_{U}(-\theta)$. By limiting the expansion of $K_{U}$ to order 2, one obtains the local approximation, as $\theta$ tends to zero:
\[V_{\theta, u}(f,q) = \expect\!_q [U] - \left(\nicefrac{\theta}{2}\right) \var_q(U) + \mathcal{O}(\theta^{2})\,.
\]
Following the standard Arrow-Pratt logic, the mean-variance approximation is accurate either when the planner's aversion $\theta$ is close to zero, or when the disparities in type-specific expected utilities are small (so that higher-order cumulants are negligible relative to the variance). Then, from Definition~\ref{def-Bregman}, $D_{K_U}(-\theta||0) = K_U(-\theta) - K_U(0) +\theta K_U^{(1)}(0)$. Moreover, by definition, $K_U(0) = 0$. Because $V_{\theta, u}(f,q) = (-\nicefrac{1}{\theta}) K_{U}(-\theta)$, we also have $V_{\theta, u}(f,q)  = \expect\!_q[U] - \left(\nicefrac{1}{\theta}\right) D_{K_U}(-\theta||0)$. \qed

\parhead{Proof of Theorem~\ref{theo-stoch_dom}} We know that $V_{\theta, u}(f,q) = \phi^{-1}_\theta \left( \expect\!_q \left[ \phi_\theta (U) \right] \right)$, with $U = {(U(\pi_s))}_{s \in \states}$. Because $\phi_\theta^{-1}$ is strictly increasing, $V_{\theta, u}(f,q) \geq V_{\theta, u}(f',q')$ if and only if $\expect_q [\phi_\theta(U)] \geq  \expect_{q'} [\phi_\theta(U')]$. Moreover, because $u$ can be assumed strictly positive, $U(\pi_s) > 0$ for all $\pi_s \in \Delta(\incomes)$. Thus, $U$ and $U'$ are strictly positive random variables. It follows that (i) $\Leftrightarrow$ (ii) can be deduced from Theorem 5.A.4 in \cite{SS07} and Theorem 1 in \cite{BG87} \cite[see also][]{Wh89}. (i) $\Rightarrow$ (ii) is immediate as $\phi_\theta$ is completely alternating for all $\theta \geq 0$. 

Let's prove (ii) $\Rightarrow$ (i). If $\varphi$ is completely alternating, then $\varphi^{(1)}$ is completely monotone and, from Hausdorff-Bernstein-Widder theorem with $t > 0$, we know that $\varphi^{(1)}(t) = \int_{[0,\infty)} \exp{(-rt)} m(\textnormal{d}r)$ for some positive measure $m$ on $[0,\infty)$. Isolating the case $r=0$, this expression can be written as $\varphi^{(1)}(t) = c + \int_{(0,\infty)} \exp{(-rt)} m(\textnormal{d}r)$ for some constant $c \geq 0$ and measure $m$ on $(0,\infty)$. Assume that $\varphi(0^+)=0$.\footnote{This normalization entails no loss of generality. First, if $\varphi(0^+)$ is finite and non-zero, the reasoning applies to the shifted function $\varphi^\star(t) = \varphi(t) - \varphi(0^+)$, which preserves the dominance between distributions \cite[see][]{BG87}. Second, if $\varphi(0^+) = -\infty$, because $u$ is continuous and assumed to be strictly positive on the compact domain $\incomes$, all expected utilities $U$ and $U'$ are bounded below by $t_0 = \min_{y \in \mathcal{Y}} u(y) > 0$. Integrating $\varphi^{(1)}$ over $[t_0, z]$ instead of $(0, z]$ ensures that all terms and integrals remain finite and yields the same conclusion.} By integration over $(0,z]$ and reversing the order of integration, we deduce that 
$\varphi(z) = c z + \int_{(0,\infty)} {\left[ \left( \nicefrac{-1}{r} \right) \exp{(-rt)} \right]}_0^z m(\textnormal{d}r) = c z + \int_{(0,\infty)} \left( \nicefrac{1}{r} \right) \left[ 1 - \exp{(-rz)} \right] m(\textnormal{d}r)$. Statement (ii) is equivalent to $\expect\!_q [- \exp{(- \theta U)}] \geq  \expect\!_{q'} [- \exp{(-\theta U')}]$ for all $\theta \in (0,\infty)$ and $\expect\!_q [U] \geq \expect\!_{q'} [U']$ if $\theta = 0$. Assume first that $\theta = 0$. If (ii) is true, it follows that:
\begin{equation} \label{eq-bernstein0}
    \expect\!_q [cU] \geq \expect\!_{q'} [cU']\,,
\end{equation}
for any $c \geq 0$, which establishes statement (i) for any purely linear function $\varphi(z) = cz$. Then, for all $\theta > 0$, statement (ii) implies that:
\begin{equation} \label{eq-bernstein}
    \expect\!_q \left[ \left( \nicefrac{1}{\theta} \right) \left[ 1 - \exp{(-\theta U)} \right] \right] \geq  \expect\!_{q'} \left[ \left( \nicefrac{1}{\theta} \right) \left[ 1 - \exp{(-\theta U')} \right] \right]\,.
\end{equation}
For any completely alternating function $\varphi$, one deduces that:
\begin{flalign} \label{eq-bernstein2}
\expect\!_q [\varphi(U)] &\overset{\text{def}}{=} \expect\!_q \left[ cU + \int_{(0,\infty)} \left( \nicefrac{1}{r} \right)
  \left[ 1 - \exp{(-r U)} \right] \, m(\textnormal{d}r) \right] \nonumber\\
&= c \expect\!_q [U] + \int_{(0,\infty)} \expect\!_q \left[ \left( \nicefrac{1}{r} \right)
  \left[ 1 - \exp{(-r U)} \right] \right] m(\textnormal{d}r) \nonumber\\
&\overset{\text{(\ref{eq-bernstein0} and \ref{eq-bernstein})}}{\geq} c \expect\!_{q'} [U'] + \int_{(0,\infty)} \expect\!_{q'} \left[ \left( \nicefrac{1}{r} \right)
  \left[ 1 - \exp{(-r U')} \right] \right] m(\textnormal{d}r) \nonumber\\
&= \expect\!_{q'} \left[ cU' + \int_{(0,\infty)} \left( \nicefrac{1}{r} \right)
  \left[ 1 - \exp{(-r U')} \right] \, m(\textnormal{d}r) \right]
\overset{\text{def}}{=} \expect\!_{q'} [\varphi(U')]\,,
\end{flalign}
Therefore, (ii) implies that statement (i) is true. \qed


\parhead{Proof of Property~\ref{prop-V_positive}} Let $y \in \incomes$. By assumption $u(y) \geq 0$, thus $U(\pi_s) \geq 0$ for all $s \in \states$. If $\theta = 0$, then $V_{0, u}(f,q) = \sum_{s \in \states} q(s) U(\pi_s) \geq 0$. Now assume that $\theta > 0$. We have $\exp(-\theta U(\pi_s)) \leq \exp(0) = 1$. It follows that $\sum_{s \in \states} q(s) \exp(-\theta U(\pi_s)) \in (0,1]$, and that $\ln \left( \sum_{s \in \states} q(s) \exp(-\theta U(\pi_s)) \right) \in (-\infty, 0]$. Finally, $V_{\theta, u}(f,q)$ is equal to this last term, times $- (\nicefrac{1}{\theta})$, thus $V_{\theta, u}(f,q) \geq 0$. \qed

\parhead{Proof of Property~\ref{prop-derivative_EU}} If $\theta = 0$, we have $V_{0, u}(f,q) = \sum_{s \in \states} q(s) U(\pi_s)$, which yields $\frac{\partial V_{0, u}}{\partial U(\pi_s)}(f,q) = q(s)$. Since $p^\star_{0, u}(f,q;s) = q(s)$, the property trivially holds. Assume now that $\theta > 0$. Let $\tilde{V}_{\theta, u}(f,q) = \sum_{t \in \states} q(t) \exp(-\theta U(\pi_t))$, so that $V_{\theta, u}(f,q) = - \left( \nicefrac{1}{\theta}\right) \ln{\tilde{V}_{\theta, u}(f,q)}$. One observes that $\nicefrac{\partial \tilde{V}_{\theta, u}}{\partial U(\pi_s)} = - \theta q(s) \exp(-\theta U(\pi_s))$. Thus:
\[
\frac{\partial V_{\theta, u}}{\partial U(\pi_s)}(f,q) = - \frac{1}{\theta} \frac{- \theta q(s) \exp(-\theta U(\pi_s))}{\tilde{V}_{\theta, u}(f,q)} = p^\star_{\theta, u}(f,q;s)\,.
\]
\qed

\parhead{Proof of Property~\ref{prop-derivative_u}} We have $U(\pi_s) = \sum_{y \in \incomes} \pi_s(y) u(y)$. If $\theta = 0$, then $V_{0, u}(f,q) = \sum_{y \in \incomes} \pi(y) u(y)$, which yields $\frac{\partial V_{0, u}}{\partial u(y)}(f,q) = \pi(y)$. Since $p^\star_{0, u}(f,q;s) = q(s)$, $\pi^\star_{0, u}(f,q;y) = \sum_{s \in \states} q(s) \pi_s(y) = \pi(y)$, meaning that the property holds. Assume now that $\theta > 0$. Let $\tilde{V}_{\theta, u}(f,q) = \sum_{s \in \states} q(s) \exp(-\theta U(\pi_s))$, such that $V_{\theta, u}(f,q) = - \left( \nicefrac{1}{\theta}\right) \ln{\tilde{V}_{\theta, u}(f,q)}$. We have $\frac{\partial \tilde{V}_{\theta, u}}{\partial u(y)} = - \theta \sum_{s \in \states} \pi_s(y) q(s) \exp \left(- \theta U(\pi_s) \right)$, recalling that $U(\pi_s) = \sum_{y \in \incomes} \pi_s(y) u(y)$. It follows that:
\[
\frac{\partial V_{\theta, u}}{\partial u(y)}(f,q) = - \frac{1}{\theta} \frac{\nicefrac{\partial \tilde{V}_{\theta, u}}{\partial u(y)}}{\tilde{V}_{\theta, u}(f,q)} = \sum_{s \in \states} \left[ \frac{q(s) \exp \left(- \theta U(\pi_s) \right)}{\tilde{V}_{\theta, u}(f,q)}\right] \pi_s(y) = \sum_{s \in \states} p^\star_{\theta, u}(f,q;s) \pi_s(y)\,.
\]
\qed

\parhead{Proof of Property~\ref{prop-extreme}} 
The case $V_{0, u}(f,q)$ is stated in equation~(\ref{HS-representation}). With Theorem~\ref{theo-main2}, it suffices to calculate $\lim_{\theta \rightarrow \infty} p^\star_{\theta, u}(f,q;s)$. Let $\overline{\states} = \{s \in \states \mid q(s) > 0\}$ be the support of $q$. Let $U_{\textnormal{min}} = \min_{s \in \overline{\states}} U(\pi_s)$ and $\states_{\textnormal{min}} = \{s \in \overline{\states}\,|\, U(\pi_s) = U_{\textnormal{min}} \}$. For any $s \notin \overline{\states}$, we trivially have $p^\star_{\theta,u}(f,q;s) = 0$ for all $\theta \ge 0$. For $s \in \overline{\states}$, by multiplying and dividing $p^\star_{\theta,u}(f,q;s)$ by $\exp(\theta U_{\min})$, one has:
\[
p^\star_{\theta, u}(f,q;s) = \frac{q(s) \exp\left(\theta (U_{\textnormal{min}}-U(\pi_s))\right)}{\sum_{t \in \states} q(t) \exp\left( \theta (U_{\textnormal{min}}-U(\pi_t))\right)} = \frac{q(s) \exp\left(\theta (U_{\textnormal{min}}-U(\pi_s))\right)}{\sum_{t \in \states_{\textnormal{min}}} q(t) + \sum_{t \notin \states_{\textnormal{min}}} q(t) \exp\left( \theta (U_{\textnormal{min}}-U(\pi_t))\right)}\,.
\]
Moreover, for all $t \in \overline{\states} \setminus \mathcal{S}_{\min}$, we have $\lim_{\theta \rightarrow \infty} \exp\left( \theta (U_{\textnormal{min}}-U(\pi_t))\right) = \exp(-\infty) = 0$. It follows that:    
\begin{equation} \label{eq-limit}
    \lim_{\theta \rightarrow \infty} p^\star_{\theta, u}(f,q;s) =  
    \begin{cases} 
        \frac{q(s)}{\sum_{t \in \states_{\textnormal{min}}} q(t)}\,, & \text{for } s \in \states_{\textnormal{min}}\,, \\
        0\,, & \text{for } s \notin \states_{\textnormal{min}}\,.
    \end{cases}
\end{equation}
By Theorem~\ref{theo-main2}, we have $V_{\theta, u}(f,q) = \expect\!_{p^\star_{\theta, u}} [U] + D_{KL}(p^\star_{\theta, u} || q) / \theta$. Observe from~(\ref{eq-limit}) that $D_{KL}(p^\star_{\theta, u} || q)$ converges to  a finite constant when $\theta \rightarrow \infty$, so that $\lim_{\theta \rightarrow \infty} D_{KL}(p^\star_{\theta, u} || q) / \theta = 0$. One deduces that
$\lim_{\theta \rightarrow \infty} V_{\theta, u}(f,q) = \lim_{\theta \rightarrow \infty} \sum_{s \in \states} p^\star_{\theta, u}(f,q;s) U(\pi_s) = U_{\textnormal{min}}$. 
\qed

\parhead{Proof of Theorem~\ref{theo-comparative}} We first prove (i) $\Leftrightarrow$ (ii). If $\theta = 0$, then $V_{0, u}(f,q) = \sum_{s \in \states} q(s) U(\pi_s)$ and the equivalence is a result of Jensen's inequality. Now assume that $\theta, \theta' > 0$. We know from the proof of Theorem~\ref{theo-main3} that $V_{\theta, u}(f,q) = \left(-\nicefrac{1}{\theta}\right) K_{U}(-\theta)$, hence:
\begin{equation} \label{derivative-V}
    \frac{\partial V_{\theta, u}(f,q)}{\partial \theta} = \frac{K_{U}(-\theta) + \theta K_{U}^{(1)}(-\theta)}{\theta^2}\,.
\end{equation}
By definition, and because we restrict our evaluation to societies where there exist $s, s' \in \states$ such that $U_s \neq U_{s'}$, the CGF $K_{U}(-\theta)$ is strictly convex. Thus $K_{U}(x) - K_{U}(y) - (x-y)K_{U}^{(1)}(y) > 0$ for all $x, y \in \real$, with $x \neq y$. Let $x = 0$ and $y = - \theta$. It follows that $K_{U}(0) > K_{U}(-\theta) + \theta K_{U}^{(1)}(-\theta)$. Moreover, $K_{U}(0) = \ln \left( \sum_{s \in \states}q(s) e^{-0 \times U(\pi_s)} \right) = 0$, thus the numerator of~(\ref{derivative-V}) is strictly negative, and $\nicefrac{\partial V_{\theta, u}(f,q)}{\partial \theta} < 0$.


Next, we prove (i) $\Rightarrow$ (iii). Assume $\theta < \theta'$. Consider an opportunity profile $f = (\pi_{s},\pi_{s'}) \in \SocietiesFixed$, with $U(\pi_s) < U(\pi_{s'})$ and $q = (\nicefrac{1}{2},\nicefrac{1}{2})$. Let $p^\star_{\theta,u}(f,q;s) = p^\star_{\theta,u}(s)$ for all $s \in \states$. From Theorem~\ref{theo-main2} we have, for $t=s,s'$:
\[
 p^\star_{\theta,u}(t) = \frac{\exp\left(-\theta U(\pi_t)\right)}{\exp\left(-\theta U(\pi_{s})\right) + \exp\left(-\theta U(\pi_{s'})\right)}\,.
\]
If $\theta = 0$, then $p^\star_{\theta,u}(s) = p^\star_{\theta,u}(s') = \nicefrac{1}{2}$. If $\theta > 0$, then $U(\pi_s) < U(\pi_{s'})$ implies $p^\star_{\theta,u}(s) > p^\star_{\theta,u}(s')$. After simplification, we also have $p^\star_{\theta,u}(s) = {\left[ 1 + \exp\left(-\theta (U(\pi_{s'})-U(\pi_{s})\right) \right]}^{-1}$.
One deduces that:
\begin{equation} \label{eq-deriv}
\frac{\partial p^\star_{\theta,u}(s)}{\partial \theta} = \left(U(\pi_{s'}) - U(\pi_{s}) \right)  \times \exp\left(-\theta (U(\pi_{s'})-U(\pi_{s})\right) \times p^\star_{\theta,u}(s)^2 > 0\,.  
\end{equation}
Thus, $\theta < \theta'$ implies that $p^\star_{\theta,u}(s) < p^\star_{\theta',u}(s)$. Because $p^\star_{\theta,u}(s') = 1 - p^\star_{\theta,u}(s)$, we have $ \frac{\partial p^\star_{\theta,u}(s')}{\partial \theta} < 0$ and $p^\star_{\theta,u}(s') > p^\star_{\theta',u}(s')$. 

We now prove (iii) $\Rightarrow$ (i). Assume statement (iii) holds. Suppose, by contradiction, $\theta \geq \theta'$. If $\theta = \theta'$, then $p^\star_{\theta,u}(s) = p^\star_{\theta',u}(s)$, which contradicts the strict inequality in (iii). If $\theta > \theta'$, the strict monotonicity established by the positive derivative in~(\ref{eq-deriv}) implies $p^\star_{\theta,u}(s) > p^\star_{\theta',u}(s)$, which again contradicts (iii). Therefore, we must have $\theta < \theta'$. \qed 
\clearpage
\section{Formal connections with decision theory under ambiguity}
\label{app: risk}

Our axiomatic characterization bridges the literature on inequality of opportunity with the literature on decision theory under ambiguity. Formally, our setup is equivalent to the \cite{AA63} framework: The `states of the world' correspond to the social types $s \in \states$, while the lotteries assigned to each state are type-specific income distributions $\pi_s \in \Delta(\incomes)$. An `act', which assigns a lottery to each state of the world, corresponds to an opportunity profile $f \in \SocietiesFixed$. Because our framework is entirely finite and discrete in states (to directly accommodate empirical data), we have $\SocietiesFixed = \Delta(\incomes)^{|\states|}$. 

\parhead{Discussion related to Theorem~\ref{theo-general}} Our baseline representation in Theorem~\ref{theo-general} is directly connected to the `second-order utility' (SOEU) and `smooth ambiguity' models characterized by \cite{KMM05}, \cite{EG09}, and \cite{GPS09}.\footnote{The dynamic choice model of \cite{KP78} also shares this recursive structure. However, it relies on temporal consistency to evaluate the timing of uncertainty resolution, whereas our static representation builds on an independence axiom to compare societies.}  However, our framework departs from these models in two fundamental ways: the nature and derivation of the outer expectation, and the characterization of (ambiguity / inequality-of-opportunity) aversion.

First, in the ambiguity literature, the second-order measure operates as a representation of a decision-maker's \textit{subjective uncertainty} regarding the `true' probability model \cite[][]{KMM05, EG09}. To derive this subjective prior, these models rely on heavy axiomatic machinery, such as complex behavioral assumptions over hypothetical second-order acts \cite[Assumptions 2 and 3 in][]{KMM05}, rich Savage-type product state spaces designed to capture unobservable second-order risks \cite[Axioms 4 and 5b in][]{EG09}, or the necessity to either impose a highly demanding ambiguity aversion axiom or couple standard uncertainty aversion with a demanding translation invariance condition \cite[A.5 for Theorem 2, or A.6--A.7 for Theorem 3 in ][]{GPS09}.\footnote{It is worth noting that while \cite{GPS09}'s translation invariance (A.7) is not required to establish our baseline SOEU representation (Theorem~\ref{theo-general}), we explicitly impose a strong form of translation invariance along the equal-opportunity line via A\ref{ax-ICEO2}, but only in Theorem~\ref{theo-main} to restrict the planner's preferences to the single-parameter exponential family.} By contrast, our probability distribution $q \in \Delta(\states)$ is discrete, uniform, and represents perfectly known, \textit{objective} demographic weights across types. That allows us to bypass the need to axiomatize unobservable subjective beliefs over a state space. By substituting the heavy machinery of subjective priors with the anonymity requirement (A\ref{ax-anonymity}) and a weak independence axiom applied to equal opportunity profiles (A\ref{ax-ICEO}), we recover the exact SOEU structure using a much leaner framework.

Second, our inequality-of-opportunity aversion axiom (A\ref{ax-aversion}) differs from the usual axioms of uncertainty aversion. Standard approaches in decision theory rely on global mixing conditions across all states simultaneously \cite[][]{GS89, Schmeidler89} or comparisons to ambiguity-neutral benchmarks \cite[][]{GM02}. Instead, our axiom operates \textit{locally} between two types, providing a direct formal equivalent to the Pigou-Dalton transfer principle from the standard literature on inequality measurement. 

\parhead{Discussion related to Theorem~\ref{theo-main}} In the literature on decision under ambiguity, \cite{St11} provides an axiomatic foundation for the single-parameter exponential form, which we obtain in Theorem~\ref{theo-main}. However, his axiomatic route is different. Because he defines this family of preferences as a subclass of the `variational preferences' introduced by \cite{MMR06}, he is required to first assume their full suite of axioms to establish this baseline. To then restrict this class down to the exponential form, he additionally imposes the `P2 Sure-Thing Principle' of \citet{Savage1954}, over the entire Anscombe-Aumann domain of acts. By contrast, our derivation takes the SOEU general class as a starting point (Theorem~\ref{theo-general}). Rather than imposing Savage's global P2 axiom over a general space of acts, we rely on the local row-separability condition A\ref{ax-ICCD}, already used to derive our Theorem~\ref{theo-general}. Indeed, on our finite product domain, repeated applications of A\ref{ax-ICCD} naturally yield the sure-thing implications necessary for the representation. To restrict the general SOEU representation to the exponential family, we do not need a novel axiom, but only a strengthening of the existing one: The independence condition A\ref{ax-ICEO} is replaced by A\ref{ax-ICEO2}. Note that A\ref{ax-ICEO} is comparable, but adapted to our framework, to the `risk independence' axiom in \cite{CVMMM11}, whereas A\ref{ax-ICEO2} corresponds to the `weak certainty independence' axiom in \cite{MMR06} for variational preferences.

\parhead{Discussion related to Theorem~\ref{theo-main2}} The dual representation of the exponential family in Theorem~\ref{theo-main2}, using the Kullback-Leibler (KL) divergence, has already been identified by \cite{St11}. In fact, he demonstrates that this specific exponential form constitutes the unique class of functions lying at the exact intersection between variational preferences \citep{MMR06} and SOEU preferences. This unique intersection corresponds to the `multiplier preferences' introduced by \cite{HS01}. However, the mathematical derivation differs substantially. The duality established by \cite{St11} builds on results from the `theory of large deviations' discussed in \cite{DE97}. By contrast, because our objective framework already yields the exponential form as a baseline (Theorem~\ref{theo-main}), we obtain the dual KL representation much more transparently through a simple optimization argument. Furthermore, rather than penalizing model misspecification as in the decision theory literature, the KL divergence measures the endogenous displacement of normative weights (induced by the profile of type-specific expected utilities and the social planner's aversion to inequality of opportunity) from demographic weights. 

\parhead{Discussion related to Theorem~\ref{theo-main3}} Our Theorem~\ref{theo-main3} establishes an additive separation within the representation function: It disentangles the efficiency component from the inequality component, with the latter explicitly integrating the planner's aversion parameter. In standard risk theory, such an additive separation is achieved via the mean-variance local approximation of the Expected Utility (EU) model (typically assuming CARA or CRRA preferences), as pioneered by \cite{Ar65} and \cite{P64}. Broadening this classic framework, \cite{NMBN11} identify a closed-form `mean-divergence' representation for the EU model that is comparable to what we obtain in Theorem~\ref{theo-main3}. They demonstrate that for exponential families of distributions, the risk premium translates exactly into a Bregman divergence.

In the specific context of SOEU preferences, \cite{MMR13} identified a generalization of the Arrow-Pratt approximation. They derive a mean-variance evaluation (based on a second-order Taylor expansion) where the ambiguity premium is additively separated and determined by the variance of expected returns across subjective models, multiplied by the agent's ambiguity aversion index. As the authors note, their local approximation becomes exact only under restrictive conditions, namely CARA utility combined with a normal distribution of risks. Our mean-variance approximation in Theorem~\ref{theo-main3} relies on a different variance penalty term. Instead of capturing an ambiguity premium through the variance of unobservable subjective beliefs, our variance term measures a purely objective reality: the variance of expected outcomes across observable demographic types.

Furthermore, our mathematical derivation takes an original route by explicitly relying on the cumulant generating function. By restricting this expansion to the second order, we recover the mean-variance approximation, mapping the planner's inequality aversion to the second cumulant. The true analytical advantage of our framework lies beyond this quadratic truncation. The Bregman divergence approach allows us to isolate the first-order cumulant (the expectation of type-specific expected utilities) from all distributional features beyond the mean, by operating through an algebraic identity rather than relying on an infinite series expansion. It encapsulates these remaining dimensions into an exact, closed-form expression. Consequently, rather than being forced to rely on a local approximation, our model yields a globally tractable structure in which this exact divergence serves as a comprehensive penalty for inequality of opportunity.
\clearpage
\section{Representation with a logarithmic utility function}
\label{app: log}

This section follows on from Sections~\ref{sec-preferences} and~\ref{sec-representation-uniform}, where we have assumed a fixed distribution $q \in \Delta(\states)$ over types, uniform across all societies, such that $q(s) = 1 / |\states|$ for all $s \in \states$. Under Axioms A\ref{ax-Pareto}-A\ref{ax-ICEO2}, Theorem~\ref{theo-main} in Section~\ref{sec-representation-uniform} characterizes the following social welfare function, for any opportunity profile $f \in \SocietiesFixed$:
\begin{equation}
    V_{\theta, u, q}(f) = 
    \begin{cases} 
            \sum_{s \in \states} q(s) U(\pi_s)\,, & \text{for } \theta = 0\,, \\
            \left( - \nicefrac{1}{\theta} \right) \ln\left(  \sum_{s \in \states} q(s) \left(e^{- \theta U(\pi_s)} \right) \right)\,, & \text{for } \theta > 0\,.
    \end{cases}
\end{equation}
We establish that two additional (yet standard) axioms imply that, within this framework, strictly positive inequality-of-opportunity aversion necessarily leads to a logarithmic utility function in income.

The following axiom states that, in the absence of inequality of opportunity, social welfare increases with additional income, regardless of the individual who receives it. Specifically, among equal-opportunity profiles, one profile is preferred to another if it can be obtained from the latter through a `favorable probability shift'.
\begin{axiom}[\textsc{Monotonicity}] \label{ax-mon}
For all $\pi, \pi' \in \Delta(\incomes)$, if there exist $\epsilon > 0$ and two incomes $x < x'$ such that $\pi(y) = \pi'(y)$ for all $y \neq x, x'$ and, furthermore, $\pi(x) = \pi'(x) - \epsilon$ and $\pi(x') = \pi'(x') + \epsilon$, then $f_{\pi} \succ f_{\pi'}$.
\end{axiom}
Consider the following example:
\[
f_{\pi}\ =\ 
\begin{pNiceMatrix}[first-row, first-col]
               & \textnormal{\$1} & \textnormal{\$2} \\
\textnormal{R} & 0.2 & 0.8 \\
\textnormal{P} & 0.2 & 0.8
\end{pNiceMatrix}
\quad \succ \quad
\begin{pNiceMatrix}[first-row, first-col]
               & \textnormal{\$1} & \textnormal{\$2} \\
\textnormal{R} & 0.4 & 0.6 \\
\textnormal{P} & 0.4 & 0.6
\end{pNiceMatrix}
\ =\ f_{\pi'}\ .
\]

The next axiom requires that multiplying all incomes by a positive constant (e.g., converting between currencies) should not affect the ranking of two opportunity profiles. For any profile  $f \in \SocietiesFixed$ and any $\lambda > 0$, we define a new profile $\lambda f$ by $(\lambda f)(s) = \pi^\lambda_s$ where, for any $s \in \states$ and any $y \in \incomes$, $\pi^\lambda_s(y) = \pi_s(y/\lambda)$. That is, profile  $\lambda f$ corresponds to $f$ where the income of all individuals is multiplied by $\lambda$. 
\begin{axiom}[\textsc{Scale invariance}] \label{ax-scale}
For all $f,f' \in \SocietiesFixed$ and all $\lambda > 0$, if $\lambda f, \lambda  f' \in \SocietiesFixed$, then 
$f \Wpref f'\ \Rightarrow\ \lambda f \Wpref \lambda f'$.
\end{axiom}
Consider the following example, where we set $\lambda = 3$:
\[
\begin{pNiceMatrix}[first-row, first-col]
               & \textnormal{\$1} & \textnormal{\$2} \\
\textnormal{R} & 0.2 & 0.8\\
\textnormal{P} & 0.2 & 0.8\\
\end{pNiceMatrix}
\ \Wpref\
\begin{pNiceMatrix}[first-row, first-col]
               & \textnormal{\$1} & \textnormal{\$2} \\
\textnormal{R} & 0.4 & 0.6\\
\textnormal{P} & 0.4 & 0.6\\
\end{pNiceMatrix}
\quad
\Longrightarrow\quad 
\begin{pNiceMatrix}[first-row, first-col]
               & \textnormal{\$3} & \textnormal{\$6} \\
\textnormal{R} & 0.2 & 0.8\\
\textnormal{P} & 0.2 & 0.8\\
\end{pNiceMatrix}
\ \Wpref\
\begin{pNiceMatrix}[first-row, first-col]
               & \textnormal{\$3} & \textnormal{\$6} \\
\textnormal{R} & 0.4 & 0.6\\
\textnormal{P} & 0.4 & 0.6\\
\end{pNiceMatrix}\,.
\]

The following result demonstrates the implications of these two new axioms, in addition to those of Theorem~\ref{theo-main}.  The proof is provided at the end of this Appendix. 

\begin{theorem} \label{theo-ln} Let $q(s) = 1/ |\states|$ for all $s \in \states$, and $|\states| \geq 3$. The social planner's preferences relation $\Wpref$ satisfies A\ref{ax-Pareto}-A\ref{ax-ICEO2} (as in Theorem~\ref{theo-main}) and A\ref{ax-mon}-A\ref{ax-scale} iff it can be represented by the function $V_{\theta, u_\sigma, q}$, with $V_{\theta, u_\sigma, q}$ as defined in Theorem~\ref{theo-main}, and with the following utility function $u_\sigma$. If $\theta = 0$ then, for all $y \in \incomes$: 
    \begin{equation} \label{util-CRRA}
        \quad u_\sigma(y) = 
        \begin{cases} 
            -y^\sigma\,, & \text{for } \sigma < 0\,, \\
            \ln{y}\,, & \text{for } \sigma = 0\,, \\
            y^\sigma \,, & \text{for } \sigma > 0\,.
        \end{cases}
    \end{equation}
If $\theta > 0$, then $u_\sigma(y) = \ln{y}$, for all $y \in \incomes$ and all $\sigma \in \real$. To simplify the notation, we let $V_{\theta, q} = V_{\theta, \ln, q}$.
\end{theorem}

The resulting restriction is significant when inequality-of-opportunity aversion is strictly positive: The social welfare function now has only one parameter, namely $\theta$, which captures the planner's aversion to inequality of opportunity. Whilst this restriction is relatively strong from a theoretical perspective, it is not particularly restrictive for empirical studies, where the use of a logarithmic utility function is fairly standard. A real advantage is that the normative choices are limited to selecting the value of $\theta$.

Note that this representation remains valid on the extended domain (as presented in Section~\ref{sec-representation-extended}), where we compare societies
$(f,q) \in \domain$, with potentially different populations, type structures, and type distributions. In this case, we denote the social welfare function $V_{\theta}(f,q)$ instead of $V_{\theta, q}(f)$.

\parhead{Proof of Theorem~\ref{theo-ln}} Assume that the axioms stated in the theorem are satisfied. From Theorem~\ref{theo-main} we know that, under A\ref{ax-Pareto}-A\ref{ax-ICEO2}, preference $\Wpref$ can be represented by the function $V_{\theta, u, q}$. Because~$u$ is defined up to a positive affine transformation, we assume, without loss of generality, that $u \geq 0$. Consider $\pi, \pi' \in \Delta(\incomes)$ as in~A\ref{ax-mon}, so that $f_\pi \succ f_{\pi'}$. This is equivalent to $\sum_{y \in \incomes} \pi(y) u(y) > \sum_{y \in \incomes} \pi'(y) u(y)$ or, by definition, to $\epsilon (u(x')-u(x)) > 0$ with $x < x'$ and $\epsilon > 0$. Therefore, A\ref{ax-mon} implies that $u$ is strictly increasing. 

Now, assume A\ref{ax-scale}, let $\lambda > 0$ and consider two profiles $f,f' \in \SocietiesFixed$ such that $f \Wpref\, f'$. By definition of $\lambda f$, we have $ V_{\theta, u, q}(\lambda f) = \phi^{-1}_\theta \left(  \sum_{s \in \states} q(s) \phi_\theta \left( \sum_{y \in \incomes} \pi^\lambda_s(\lambda y) u(\lambda y) \right) \right)$. Recalling that $\pi^\lambda_s(\lambda y) = \pi_s(y)$ and denoting $u^\lambda(y) = u(\lambda y)$, we have $V_{\theta, u, q}(\lambda f) = V_{\theta, u^\lambda, q}(f)$. Using A\ref{ax-scale}, $f \Wpref\, f' \Leftrightarrow \lambda f \Wpref\, \lambda f'$. For any $\lambda>0$ such that the overlap $\incomes \cap \lambda^{-1} \incomes$ forms an interval, this implies that $V_{\theta, u, q}$ and $V_{\theta, u^\lambda, q}$ represent the same preference, restricted to profiles supported on this interval. Applying the uniqueness result of Theorem~\ref{theo-main} to this interval, $u^\lambda$ must be a positive affine transformation of $u$ specifically on this overlap. Thus, there exist $a(\lambda) > 0$ and $b(\lambda) \in \real$ such that $u^\lambda (y) = u(\lambda y) = a(\lambda) u(y) + b(\lambda)$ for all $y, \lambda y \in \incomes$. As established by \cite{A66} on unbounded domains (see Equation~(42), Page 159) and extended in \cite{Aczel2005} for measurable functions on an interval (Corollary p. 3233), the solution to this equation is any increasing affine transformation of $u_\sigma(y)$ as introduced in~(\ref{util-CRRA}). Note that we cannot consider $u_\sigma(y) = - \ln{y}$ as solution, because $u$ has to be increasing (A\ref{ax-mon}). The second restriction is  $\theta' = \theta / a(\lambda)$. Because $V_{\theta', u', q} = V_{\theta, u^\lambda, q}$, we have $\theta = \theta / a(\lambda)$. If $\theta = 0$, this condition is satisfied and~(\ref{util-CRRA}) is the global solution. If $\theta > 0$, then $a(\lambda) = 1$, and we must have $u^\lambda (y) = u(\lambda y) = u(y) + b(\lambda)$, whose solution reduces to $u(y) = \ln{y}$. 

Conversely, it is not difficult to establish that $u_\sigma$ is sufficient to satisfy A\ref{ax-mon}-A\ref{ax-scale} by observing that, when $u_\sigma(y) = \ln{y}$ and $\theta \geq 0$, $U(\pi^\lambda_s) = U(\pi_s) + \ln{\lambda}$ and $V_{\theta, \ln, q}(\lambda f) = V_{\theta, \ln, q}(f) + \ln{\lambda}$. \qed
\clearpage
\section{Additional Exhibits}
\label{app: tables}
\begin{table}[H]
    \begin{center}
    \caption{Summary statistics for birth cohorts (1940s--1980s)}
    \renewcommand{\arraystretch}{0.3}{
\def\sym#1{\ifmmode^{#1}\else\(^{#1}\)\fi}
\begin{tabular}{l*{1}{cccc}}
\toprule
                &     Mean&       SD&      Min&      Max\\
\midrule
\textbf{1940s ($ N=484 $)}           &         &         &         &         \\
Male            &     0.50&     0.50&     0.00&     1.00\\
Non-Hispanic White&     0.81&     0.39&     0.00&     1.00\\
Age             &    34.86&     0.64&    31.00&    38.00\\
Child income  $ y^{c} $ ('000 USD)&    72.62&    43.10&     0.11&   273.92\\
Parent income $ y^{p} $ ('000 USD)&    59.50&    43.74&     0.01&   300.06\\
\midrule
\textbf{1950s ($ N=2,413 $)}           &         &         &         &         \\
Male            &     0.49&     0.50&     0.00&     1.00\\
Non-Hispanic White&     0.76&     0.43&     0.00&     1.00\\
Age             &    34.70&     1.04&    31.00&    39.00\\
Child income  $ y^{c} $ ('000 USD)&    72.37&    54.12&     0.04&   516.78\\
Parent income $ y^{p} $ ('000 USD)&    65.47&    50.35&     0.00&   490.47\\
\midrule
\textbf{1960s ($ N=1,842 $)}           &         &         &         &         \\
Male            &     0.49&     0.50&     0.00&     1.00\\
Non-Hispanic White&     0.68&     0.47&     0.00&     1.00\\
Age             &    34.07&     1.00&    31.00&    38.00\\
Child income  $ y^{c} $ ('000 USD)&    76.50&    63.74&     0.04&   777.59\\
Parent income $ y^{p} $ ('000 USD)&    68.17&    51.85&     0.00&   540.95\\
\midrule
\textbf{1970s ($ N=1,713 $)}           &         &         &         &         \\
Male            &     0.49&     0.50&     0.00&     1.00\\
Non-Hispanic White&     0.60&     0.49&     0.00&     1.00\\
Age             &    34.89&     0.56&    32.00&    38.00\\
Child income  $ y^{c} $ ('000 USD)&    80.83&    67.52&     0.01&   904.62\\
Parent income $ y^{p} $ ('000 USD)&    67.23&    55.78&     0.04&   684.81\\
\midrule
\textbf{1980s ($ N=2,229 $)}           &         &         &         &         \\
Male            &     0.50&     0.50&     0.00&     1.00\\
Non-Hispanic White&     0.58&     0.49&     0.00&     1.00\\
Age             &    34.04&     1.06&    32.00&    38.00\\
Child income  $ y^{c} $ ('000 USD)&    77.49&    73.52&     0.07& 1,141.78\\
Parent income $ y^{p} $ ('000 USD)&    73.94&    82.96&     0.00& 1,413.26\\
\bottomrule
\end{tabular}
}

    \label{tab:distributions}
    \end{center}
    \begin{spacing}{0.9}
{\footnotesize
\textbf{Data:} PSID.\\
\textbf{Note:} This table reports descriptive statistics for our analysis sample separately for each decadal birth cohort from the 1940s (1940--1949) to the 1980s (1980--1989). Children's income $y^{c}$ is measured as the average of total family labor income over ages 30--40, and parental income $y^{p}$ is measured as a weighted average of total family labor income during the child's adolescence. All income values are expressed in 2019 USD. Observations are weighted to match the relative frequencies of sex $\times$ race cells in the corresponding pooled Current Population Survey (CPS, 1969--2023).}
\end{spacing}
\end{table}

\begin{table}[H]
\footnotesize
\begin{center}
\caption{Summary statistics for the 50 largest CZs}
    \label{tab:cz-summary}
\renewcommand{\arraystretch}{0.5}
\begin{tabular}{l*{1}{ccccc}}
\toprule
 & State & N & Share & Share & Mean \\
 & & children & male\textsuperscript{a} & non-Hispanic & child income \\
 & & ('000) & & White\textsuperscript{a} & ('000 USD) \\
\midrule
Atlanta & GA &   206.19 & 0.51 & 0.53 &    47.24 \\
Austin & TX &    84.52 & 0.51 & 0.55 &    53.01 \\
Baltimore & MD &   181.24 & 0.51 & 0.60 &    56.49 \\
Boston & MA &   309.64 & 0.51 & 0.75 &    62.87 \\
Bridgeport & CT &   225.77 & 0.51 & 0.71 &    61.90 \\
Buffalo & NY &   181.45 & 0.51 & 0.80 &    56.47 \\
Charlotte & NC &    97.40 & 0.51 & 0.63 &    48.73 \\
Chicago & IL &   576.41 & 0.51 & 0.54 &    56.53 \\
Cincinnati & OH &   157.68 & 0.51 & 0.81 &    54.64 \\
Cleveland & OH &   202.81 & 0.51 & 0.74 &    52.03 \\
Columbus & OH &   120.63 & 0.51 & 0.77 &    52.40 \\
Dallas & TX &   241.51 & 0.51 & 0.52 &    54.06 \\
Dayton & OH &    94.63 & 0.51 & 0.81 &    51.59 \\
Denver & CO &   169.22 & 0.51 & 0.68 &    57.44 \\
Detroit & MI &   392.12 & 0.51 & 0.70 &    52.47 \\
Fort Worth & TX &   135.27 & 0.51 & 0.61 &    55.55 \\
Fresno & CA &   109.07 & 0.51 & 0.36 &    46.14 \\
Grand Rapids & MI &   109.43 & 0.51 & 0.81 &    56.08 \\
Houston & TX &   345.84 & 0.51 & 0.45 &    55.70 \\
Indianapolis & IN &   115.70 & 0.51 & 0.76 &    53.87 \\
Jacksonville & FL &    91.83 & 0.51 & 0.64 &    45.01 \\
Kansas City & MO &   139.10 & 0.51 & 0.74 &    55.96 \\
Las Vegas & NV &    89.04 & 0.51 & 0.56 &    47.74 \\
Los Angeles & CA &  1018.87 & 0.51 & 0.34 &    49.88 \\
Manchester & NH &    91.28 & 0.51 & 0.90 &    59.28 \\
Miami & FL &   236.35 & 0.51 & 0.30 &    46.02 \\
Milwaukee & WI &   131.29 & 0.51 & 0.72 &    57.97 \\
Minneapolis & MN &   219.66 & 0.51 & 0.82 &    64.95 \\
Nashville & TN &    88.52 & 0.51 & 0.74 &    50.76 \\
New Orleans & LA &   118.80 & 0.51 & 0.49 &    46.42 \\
New York & NY &   660.84 & 0.51 & 0.38 &    54.60 \\
Newark & NJ &   361.29 & 0.51 & 0.53 &    63.42 \\
Orlando & FL &   114.69 & 0.51 & 0.55 &    44.79 \\
Philadelphia & PA &   407.58 & 0.51 & 0.65 &    57.87 \\
Phoenix & AZ &   208.32 & 0.51 & 0.62 &    51.19 \\
Pittsburgh & PA &   181.13 & 0.51 & 0.86 &    61.15 \\
Port St. Lucie & FL &    87.12 & 0.51 & 0.58 &    46.54 \\
Portland & OR &   131.47 & 0.51 & 0.77 &    54.23 \\
Providence & RI &   104.93 & 0.51 & 0.79 &    57.19 \\
Raleigh & NC &    86.70 & 0.51 & 0.63 &    51.96 \\
Sacramento & CA &   182.42 & 0.51 & 0.56 &    52.08 \\
Salt Lake City & UT &   137.34 & 0.51 & 0.84 &    60.56 \\
San Antonio & TX &   138.51 & 0.51 & 0.37 &    47.40 \\
San Diego & CA &   175.30 & 0.51 & 0.45 &    53.01 \\
San Francisco & CA &   270.01 & 0.51 & 0.43 &    58.22 \\
San Jose & CA &   147.56 & 0.51 & 0.39 &    57.99 \\
Seattle & WA &   275.14 & 0.51 & 0.72 &    57.25 \\
St. Louis & MO &   192.62 & 0.51 & 0.72 &    55.39 \\
Tampa & FL &   135.60 & 0.51 & 0.66 &    45.29 \\
Washington DC & DC &   303.16 & 0.51 & 0.50 &    62.59 \\
\bottomrule
\end{tabular}

\end{center}
\begin{spacing}{0.9}
{\footnotesize
\textbf{Data:} Opportunity Insights.\\
\textbf{Note:} This table reports descriptive statistics for the 50 largest commuting zones (CZs) in the US, based on the birth cohorts 1980--1985 \citet{Chetty2014}. $^\text{a}$ The share of males and non-Hispanic Whites is reported in \citep{Chetty2026} for the birth cohort 1978--1983. All income values are expressed in 2019 USD.}
\end{spacing}
\end{table}

\begin{figure}[htbp]
  \caption{Validation of marginal distribution against \citet{Chetty2014}}
  \begin{center}
  \label{fig:validation}
  \begin{subfigure}[t]{0.49\textwidth}
    \centering
    \includegraphics[width=\textwidth]{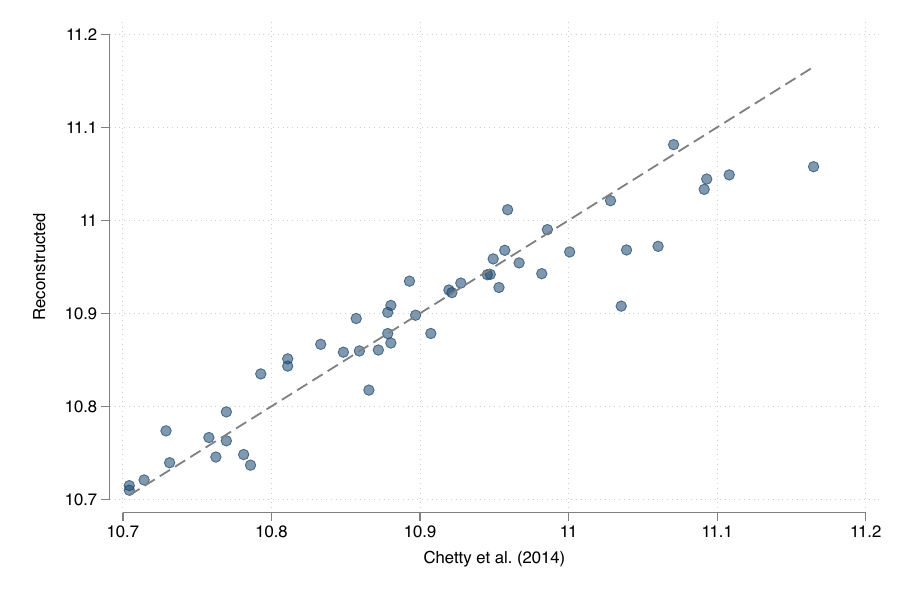}
    \caption{Mean}
  \end{subfigure}
  \hfill
  \begin{subfigure}[t]{0.49\textwidth}
    \centering
\includegraphics[width=\textwidth]{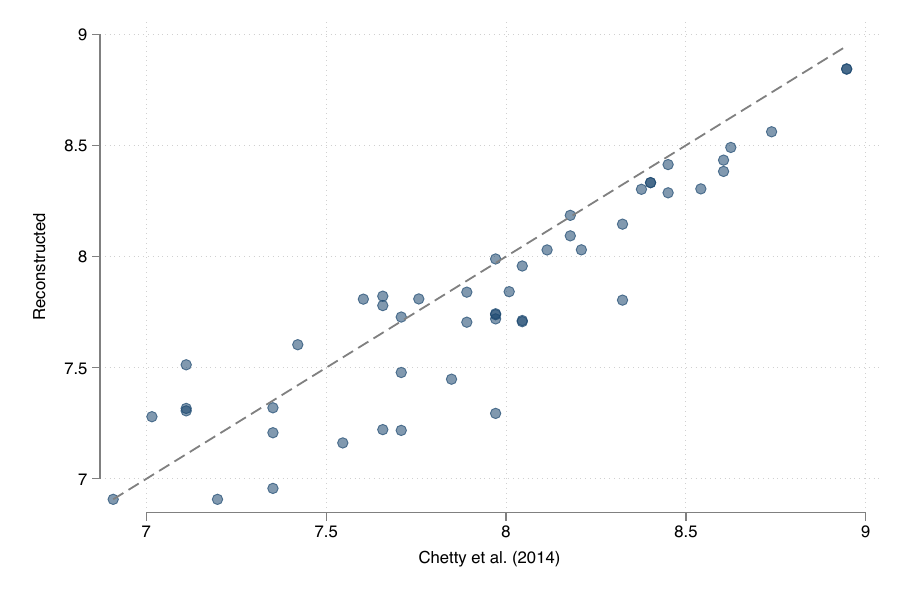}
  \caption{P10}
  \end{subfigure}
  \begin{subfigure}[t]{0.49\textwidth}
    \centering
    \includegraphics[width=\textwidth]{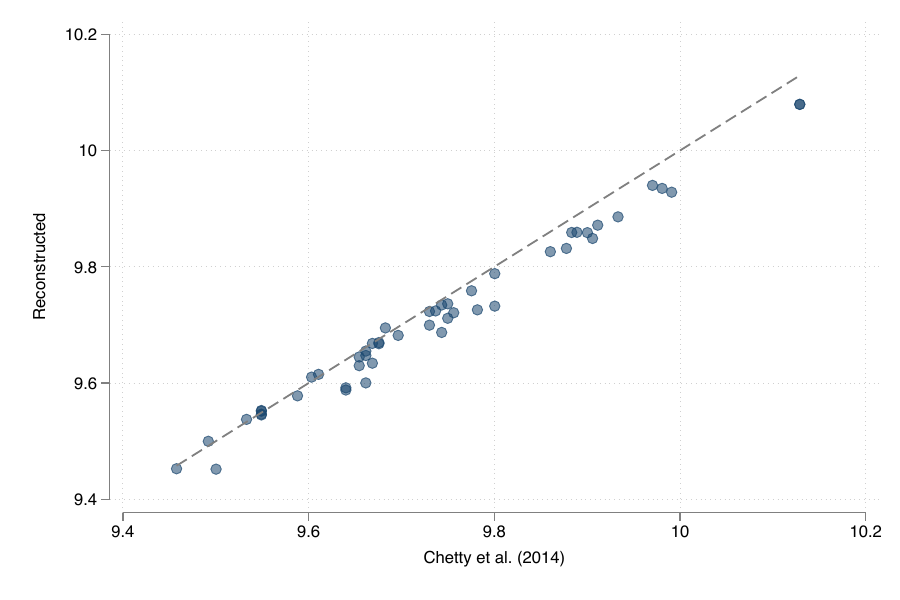}
    \caption{P25}
  \end{subfigure}
  \hfill
  \begin{subfigure}[t]{0.49\textwidth}
    \centering
\includegraphics[width=\textwidth]{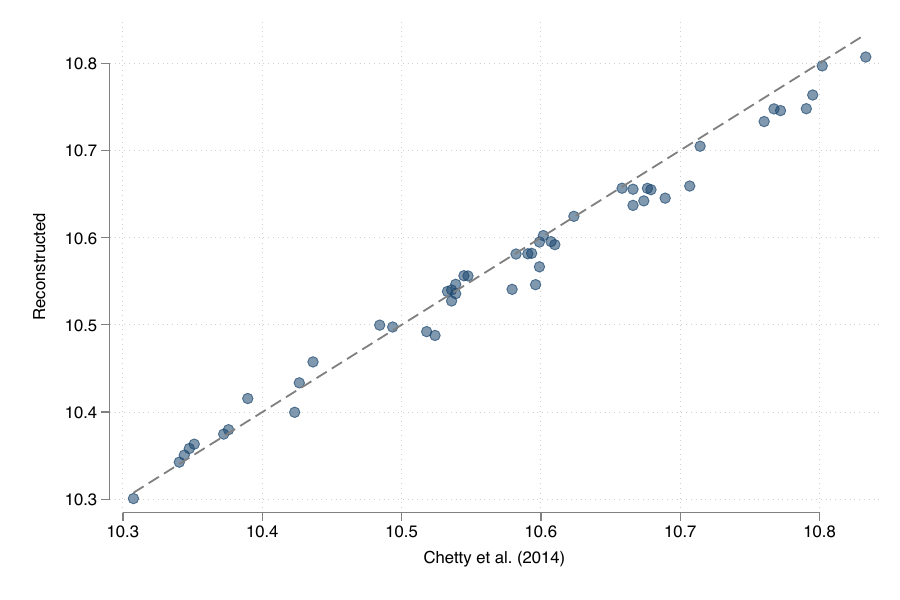}
  \caption{P50}
  \end{subfigure}
  \begin{subfigure}[t]{0.49\textwidth}
    \centering
    \includegraphics[width=\textwidth]{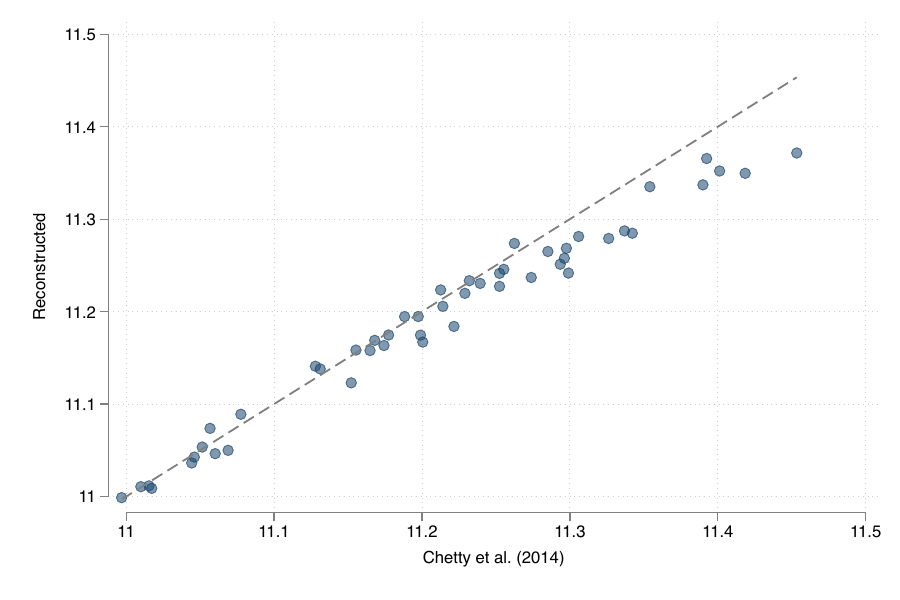}
    \caption{P75}
  \end{subfigure}
  \hfill
  \begin{subfigure}[t]{0.49\textwidth}
    \centering
\includegraphics[width=\textwidth]{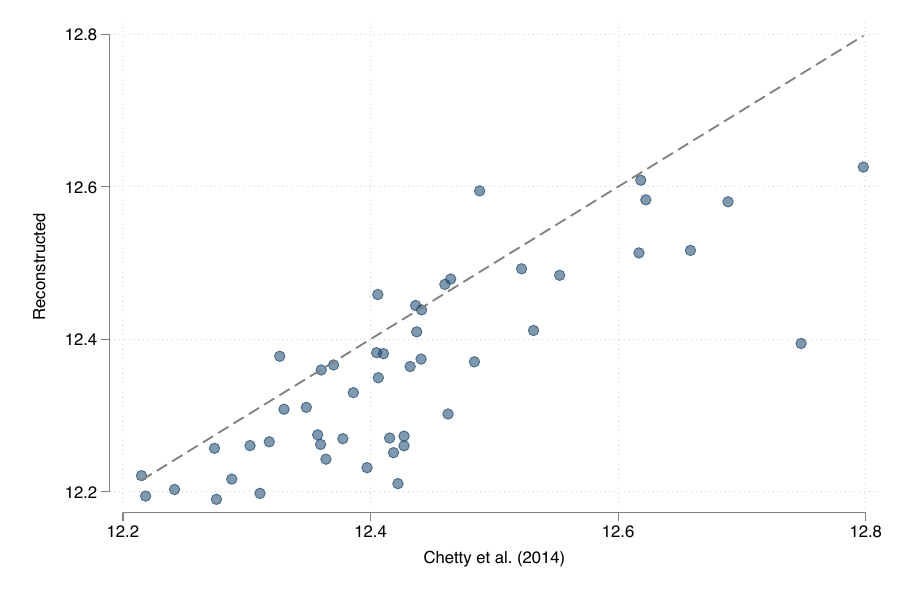}
  \caption{P99}
  \end{subfigure}
\end{center}
\begin{spacing}{0.9}
{\footnotesize
\textbf{Data:} Opportunity Insights.\\
\textbf{Note:} This figure compares moments of the income distributions of the 50 largest CZs. Panels (a)--(f) show mean income and the 10th, 25th, 50th, 75th, and 99th percentiles, respectively. In each panel, we plot the log income from our reconstructed distribution (y-axis) against the corresponding log income calculated by \citet{Chetty2014} from individual-level microdata. In the case of perfect alignment, all data points would cluster on the dashed 45-degree line. Moments reported in \citet{Chetty2014} are based on the 1980--1982 birth cohort, and reconstructed moments are based on the 1980--1985 cohort.}
\end{spacing}
\end{figure}

\begin{figure}[htbp]
  \caption{Validation of joint distribution against \citet{Chetty2014}}
  \begin{center}
  \label{fig:validation_joint}
  \begin{subfigure}[t]{0.49\textwidth}
    \centering
    \includegraphics[width=\textwidth]{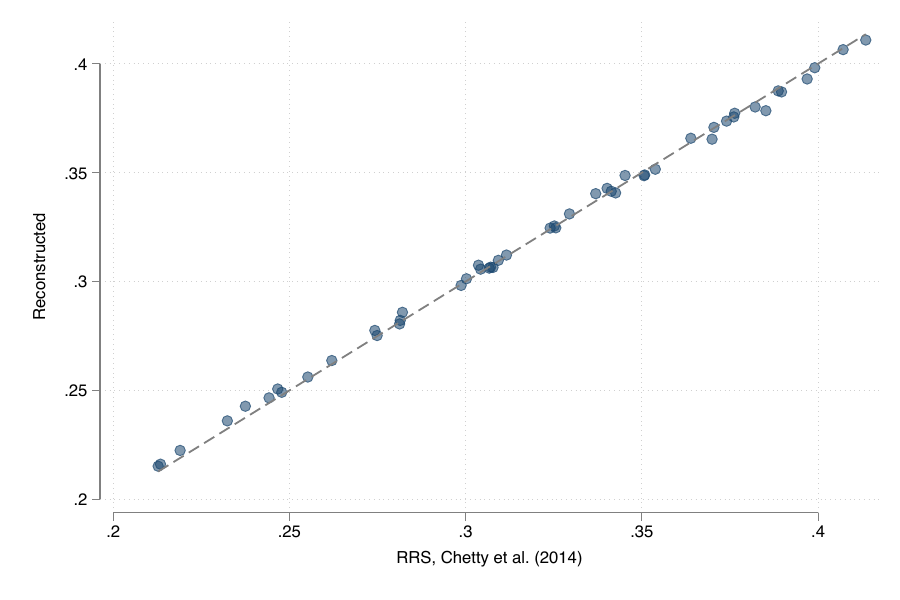}
    \caption{RRS}
  \end{subfigure}
  \hfill
  \begin{subfigure}[t]{0.49\textwidth}
    \centering
\includegraphics[width=\textwidth]{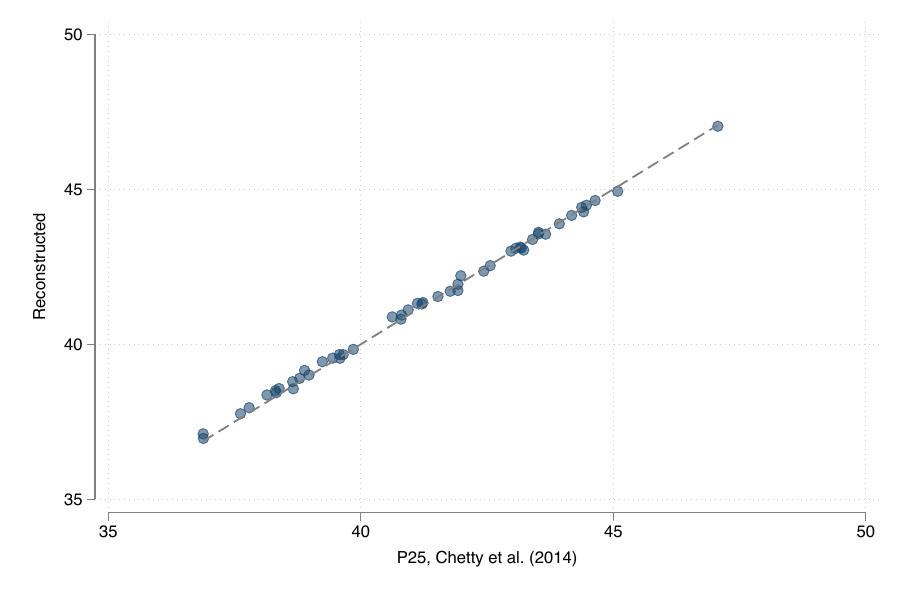}
  \caption{P25}
  \end{subfigure}
\end{center}
\begin{spacing}{0.9}
{\footnotesize
\textbf{Data:} Opportunity Insights.\\
\textbf{Note:} This figure compares moments of the joint income distributions of the 50 largest CZs. Panel (a) shows the rank-rank slope (RRS), and Panel (b) shows the mean child rank at the 25th percentile of the parental income distribution (P25). In the case of perfect alignment, all data points would cluster on the dashed 45-degree line. Estimates of \citet{Chetty2014} are based on the 1980--1985 birth cohort.}
\end{spacing}
\end{figure}
\clearpage
\section{Construction of statistical confidence sets}
\label{app: confsets}
\noindent This section describes how we construct the statistical confidence sets and the combined sets reported in Tables~\ref{tab:cohorts-cs} and~\ref{tab:cz-cs}. The procedure follows \citet{MRSW24} closely. While \citet{MRSW24} rank CZs by an estimate of intergenerational mobility, we rank societies by social welfare $V^j_{\theta,u}$. This has two main implications. First, standard errors are not readily available, and we therefore simulate the sampling distribution of $V^j_{\theta,u}$. Second, since $V^j_{\theta,u}$ varies over $\theta$, the set of pairwise comparisons underlying the inference procedure ranges over triples $(j,k,\theta)$ rather than over pairs of societies $(j,k)$. Evaluating dominance relations across pairs \textit{and} over a range of $\theta$ values turns statistical confidence sets into combined sets.

\parhead{Sampling distribution}
We first simulate the sampling distribution of $V^j_{\theta,u}$ by generating $B=2,000$ draws of the estimated inputs to the welfare function. In both applications, we assume independence across the societies of our comparison, which is justified given that children are assigned to only one cohort and CZ, respectively.\footnote{Technically, draws are only independent conditional on common inputs such as the CPS-calibrated weights in the cohort application, or the national within-quintile placement function $h(r^c\mid k(r^c), r^p)$ and the national percentile-income schedule $g(y^c \mid r^c)$ in the CZ application. Given the sample sizes underlying these national statistics, we treat the dependence induced by these inputs as negligible.} The resampled object differs across the two applications. In the trends application (Section \ref{sec-trends}), we estimate Equation~\eqref{ige} separately for each cohort $j$ and draw the intercept and slope jointly from a normal distribution centered on the point estimates, using the estimated variance-covariance matrix of the estimators. Note that these draws capture sampling variation in the estimated gradient but not in the type distribution itself. Each draw is then mapped into type utilities through Equation~\eqref{exputil}. In the spatial application (Section~\ref{sec-space}), we do not observe individual children. Instead, the simulation targets the CZ-specific joint distribution of parental and child income quintiles $p^{CZ}(r^p,k) = q^{CZ}(r^p) \times t^{CZ}(k \mid r^p)$, i.e., the first two factors of Equation~\eqref{eq:fmap}. For each CZ, we draw the $5\times 5=25$ cell counts from a multinomial distribution with sample size $N^{CZ}$ and cell probabilities $p^{CZ}$ provided in \citet{Chetty2014}, which is equivalent to resampling individual children with replacement. We then re-estimate the IRS (Equation~\eqref{eq:irs}) on the resulting frequency weights, holding the national within-quintile placement $h(r^c \mid k(r^c), r^p)$ and the percentile-income schedule $g(y^c \mid r^c)$ fixed across draws. We obtain one welfare estimate $V^{j,b}_{\theta,u}$ per society $j$ and $\theta$ for a common draw $b$.


\parhead{Standard errors of the pairwise welfare differences}
We then convert these draws into standard errors for the pairwise welfare differences $\Delta^{jk}_{\theta} = V^j_{\theta,u} - V^k_{\theta,u}$. For every pair $(j,k)$ and every $\theta$ in the grid $\Theta \cap [\theta_{\min},\theta_{\max}]$, we estimate $\widehat{se}_{jk}(\theta)$ as the standard deviation of $\hat\Delta^{jk,b}_{\theta}$ across the $b = 1,\dots,B$ draws. 

\parhead{Critical values and confidence regions}
Given the standard errors, we compute critical values following Section~3.1 of \citet{MRSW24}. For a set of comparisons $S \subseteq \{(j,k,\theta) \in \mathcal{J} \times \mathcal{J} \times \Theta :
j \neq k \}$, the critical value $c(S)$ is the $1-\alpha$ quantile of the studentized maximum absolute deviation of $\hat\Delta^{jk,b}_{\theta}$ from $\hat\Delta^{jk}_{\theta}$ over all comparisons in $S$, which we estimate from the $B$ draws. The main difference to Equation~(8) in \citet{MRSW24} is that the maximum is taken over $(j,k,\theta)$ rather than over $(j,k)$. The confidence region $C(1-\alpha,S)$ for the welfare differences in $S$ is then given by the intervals $\hat\Delta^{jk}_{\theta} \pm c(S)\,\widehat{se}_{jk}(\theta)$ (Equation~(11) in \citet{MRSW24}). We set $\alpha = 0.05$ throughout, so that the confidence region covers all welfare differences in $S$ simultaneously with $95\%$ probability. 


\parhead{Defining dominance relations}
We next translate the confidence region into dominance relations. Society $j$ dominates $k$ if the confidence band for $\Delta^{jk}_{\theta}$ lies entirely above zero at every $\theta$ contained in $S$, i.e., if the smallest lower endpoint across those $\theta$ values is strictly positive. In the statistical confidence sets, $S$ contains a fixed $\theta_0$, and only one confidence band is evaluated for each pairwise difference. In the combined sets, dominance also requires that the lower endpoint be positive at every $\theta \in \Theta \cap [\theta_{\min},\theta_{\max}]$. We then count the numbers $n_j^+$ and $n_j^-$ of societies that $j$ dominates and that dominate $j$, and form the set of ranks that $j$ can occupy according to Equation~\eqref{eq:ranksets} and Equation (13) in \citet{MRSW24}. Tables~\ref{tab:cohorts-cs} and~\ref{tab:cz-cs} vary which comparisons enter $S$ and how large $c(S)$ is as a result.


\parhead{The four types of sets}
Setting $c(S)=0$ at a single $\theta_0$ returns the point-estimate ranking, and setting $c(S)=0$ while letting $\theta$ range over $[\theta_{\min},\theta_{\max}]$ reproduces the normative dominance sets $R_j^{[\theta_{\min},\theta_{\max}]}$. Allowing $c(S)>0$ at a single $\theta_0$ gives the statistical confidence set $R_j^{1-\alpha}(\theta_0)$, which we report at $\theta_0 \in \{0,2,\infty\}$. Allowing $c(S)>0$ and letting $S$ range over the whole grid $\Theta \cap [\theta_{\min},\theta_{\max}]$ yields the combined set $R_j^{[\theta_{\min},\theta_{\max}],\,1-\alpha}$.\footnote{We maintain the fixed $\theta$ grid used to compute dominance sets when computing combined sets. This grid ensures that pairwise welfare differences $\hat\Delta^{jk}_{\theta}$ do not change sign between grid nodes (Footnote \ref{fn:grid}). Technically, combined sets require that the lower endpoint of the confidence band $\hat\Delta^{jk}_{\theta} - c(S)\,\widehat{se}_{jk}(\theta)$ remain positive between nodes. We use the maintained grid as an approximation to this stronger requirement.}

\parhead{Marginal versus simultaneous coverage}
Finally, we distinguish marginal from simultaneous sets, following Sections~3.2 and~3.3 of \citet{MRSW24}. Letting $S$ contain only the comparisons involving society $j$ yields the \emph{marginal} sets, which cover the rank of a single society with probability $1-\alpha$; letting it contain all pairs yields the \emph{simultaneous} sets, which jointly cover the entire ranking profile. Since $c(S)$ is weakly increasing in the scope of $S$, the reported sets nest by construction: Marginal within simultaneous, $R^{[0,2]}_j$ within $R^{[0,\infty)}_j$, and normative and statistical within combined. For simplicity, we do not implement the stepwise multiple-testing procedure in Sections~3.2.2 and~3.3.2 of \citet{MRSW24}, which would yield slightly narrower confidence sets.

\clearpage
\pagenumbering{Roman}
\thispagestyle{empty}
\begin{center}
    {\Large\textsc{Supplementary Material}}\\[1.2em]
    {\LARGE\bfseries Opportunity-Sensitive Social Welfare}\\[2.5em]
    {\large Tim Wienand \textcircled{r} Brice Magdalou \textcircled{r} Richard Nock \textcircled{r} Paul Hufe}\\[3em]
    \centering\large\textbf{Online publication only}
    \vspace*{\fill}
\end{center}
\clearpage

\pagenumbering{arabic}
\setcounter{section}{0}
\setcounter{figure}{0}
\setcounter{table}{0}
\setcounter{equation}{0}
\renewcommand\thesection{S.\arabic{section}}
\renewcommand\thefigure{S.\arabic{figure}}
\renewcommand\thetable{S.\arabic{table}}
\renewcommand\theequation{S.\arabic{equation}}
\chead{\textsc{Opportunity-Sensitive Social Welfare -- Supplementary Material}}

\section{Supplementary Tables and Figures}
\begin{table}[H]
 \footnotesize
 \begin{center}
 \caption{Dominance sets for birth cohorts under alternative trimming of extreme values}
     \label{tab:cohorts-trim}
 \renewcommand{\arraystretch}{1}
  \setlength{\tabcolsep}{2pt}
\begin{tabular}{l*{3}{ccc}}
\toprule
  & \multicolumn{3}{c}{Baseline} & \multicolumn{3}{c}{Trim \(y^c[p1,p99]\)} & \multicolumn{3}{c}{Trim \(y^c[p5,p95]\)} \\
 \cmidrule(lr){2-4} \cmidrule(lr){5-7} \cmidrule(lr){8-10}
Cohort & Rank & \(R^{[0,2]}\) & \(R^{[0,\infty)}\) & Rank & \(R^{[0,2]}\) & \(R^{[0,\infty)}\) & Rank & \(R^{[0,2]}\) & \(R^{[0,\infty)}\) \\
\midrule
1940s & 1 & [1,1] & [1,1] & 1 & [1,1] & [1,1] & 2 & [1,2] & [1,2] \\
1950s & 4 & [2,4] & [2,4] & 4 & [2,4] & [2,4] & 4 & [2,4] & [2,4] \\
1960s & 3 & [3,4] & [3,4] & 3 & [3,4] & [3,4] & 3 & [3,4] & [3,4] \\
1970s & 2 & [2,3] & [2,3] & 2 & [2,3] & [2,3] & 1 & [1,3] & [1,3] \\
1980s & 5 & [5,5] & [5,5] & 5 & [5,5] & [5,5] & 5 & [5,5] & [5,5] \\
\midrule
Average width & & 1.8 & 1.8 & & 1.8 & 1.8 & & 2.2 & 2.2 \\
\bottomrule
\end{tabular}
 
 \end{center}
 \begin{spacing}{1}
 {\footnotesize
 \textbf{Data:} PSID.\\
 \textbf{Note:} This table reports dominance sets for the birth cohorts 1940s--1980s conditional on trimming the lower and upper tails of the child income distribution. The baseline uses the full distribution of positive child income $y^c$. In the other columns, $y^c$ is trimmed at the 1st and 99th percentiles, or the 5th and 95th percentiles, of the cohort-specific distributions of $y^c$, respectively. Sampling weights are renormalized in each cohort. \textit{Rank} reports the ranking by $V_{0,u}$; rank 1 (5) corresponds to the highest (lowest) $V_{0,u}$.}
 \end{spacing}
 \end{table}

\begin{table}[H]
\footnotesize
\begin{center}
\caption{Dominance sets for birth cohorts under alternative functional forms}
     \label{tab:cohorts-ff}
\renewcommand{\arraystretch}{1}
  \setlength{\tabcolsep}{2pt}
\begin{tabular}{lc*{5}{cc}}
\toprule
  & & \multicolumn{2}{c}{Linear} & \multicolumn{2}{c}{Quadratic} & \multicolumn{2}{c}{Cubic} & \multicolumn{2}{c}{\shortstack{Quintiles,\\nonparametric}} & \multicolumn{2}{c}{\shortstack{Quintiles,\\linear}} \\
 \cmidrule(lr){3-4} \cmidrule(lr){5-6} \cmidrule(lr){7-8} \cmidrule(lr){9-10} \cmidrule(lr){11-12}
Cohort & Rank & \(R^{[0,2]}\) & \(R^{[0,\infty)}\) & \(R^{[0,2]}\) & \(R^{[0,\infty)}\) & \(R^{[0,2]}\) & \(R^{[0,\infty)}\) & \(R^{[0,2]}\) & \(R^{[0,\infty)}\) & \(R^{[0,2]}\) & \(R^{[0,\infty)}\) \\
\midrule
1940s & 1 & [1,1] & [1,1] & [1,1] & [1,1] & [1,1] & [1,1] & [1,1] & [1,1] & [1,1] & [1,1] \\
1950s & 4 & [2,4] & [2,4] & [4,4] & [2,4] & [3,4] & [2,4] & [4,4] & [4,4] & [3,4] & [2,4] \\
1960s & 3 & [3,4] & [3,4] & [3,3] & [3,4] & [3,4] & [3,4] & [2,3] & [2,3] & [2,3] & [2,4] \\
1970s & 2 & [2,3] & [2,3] & [2,2] & [2,3] & [2,2] & [2,3] & [2,3] & [2,3] & [2,4] & [2,4] \\
1980s & 5 & [5,5] & [5,5] & [5,5] & [5,5] & [5,5] & [5,5] & [5,5] & [5,5] & [5,5] & [5,5] \\
\midrule
Average width & & 1.8 & 1.8 & 1.0 & 1.8 & 1.4 & 1.8 & 1.4 & 1.4 & 1.8 & 2.2 \\
\bottomrule
\end{tabular}
 
\end{center}
\begin{spacing}{1}
{\footnotesize
 \textbf{Data:} PSID.\\
 \textbf{Note:} This table reports dominance sets for the birth cohorts 1940s--1980s under different functional forms of Equation \eqref{ige}. The linear specification replicates the baseline results of Figure~\ref{fig:rol}. The quadratic and cubic specifications contain second- and third-order polynomials in $\ln y^p$, respectively. The remaining two specifications redefine types as parental income quintiles and estimate $U(\pi_s)$ as the weighted mean of $\ln y^c$ within types (nonparametric) and as a linear projection analogous to the baseline specification. \textit{Rank} reports the ranking by $V_{0,u}$; rank 1 (5) corresponds to the highest (lowest) $V_{0,u}$.}
 \end{spacing}
 \end{table}

 \begin{table}[H]
 \footnotesize
 \begin{center}
 \caption{Dominance sets for birth cohorts under alternative income definitions}
     \label{tab:cohorts-incdefs}
 \renewcommand{\arraystretch}{1}
  \setlength{\tabcolsep}{2pt}
\begin{tabular}{l*{4}{ccc}}
\toprule
  & \multicolumn{3}{c}{\shortstack{Family labor\\income}} & \multicolumn{3}{c}{\shortstack{Family labor\\income, equivalized}} & \multicolumn{3}{c}{\shortstack{Family disposable\\income}} & \multicolumn{3}{c}{\shortstack{Family disposable\\income, equivalized}} \\
 \cmidrule(lr){2-4} \cmidrule(lr){5-7} \cmidrule(lr){8-10} \cmidrule(lr){11-13}
Cohort & Rank & \(R^{[0,2]}\) & \(R^{[0,\infty)}\) & Rank & \(R^{[0,2]}\) & \(R^{[0,\infty)}\) & Rank & \(R^{[0,2]}\) & \(R^{[0,\infty)}\) & Rank & \(R^{[0,2]}\) & \(R^{[0,\infty)}\) \\
\midrule
1940s & 1 & [1,1] & [1,1] & 4 & [1,4] & [1,4] & 2 & [1,2] & [1,2] & 4 & [3,4] & [1,4] \\
1950s & 4 & [2,4] & [2,4] & 5 & [2,5] & [2,5] & 4 & [4,4] & [2,4] & 5 & [4,5] & [4,5] \\
1960s & 3 & [3,4] & [3,4] & 2 & [1,4] & [1,4] & 3 & [3,3] & [2,4] & 2 & [2,2] & [1,3] \\
1970s & 2 & [2,3] & [2,3] & 1 & [1,4] & [1,4] & 1 & [1,2] & [1,4] & 1 & [1,1] & [1,3] \\
1980s & 5 & [5,5] & [5,5] & 3 & [3,5] & [3,5] & 5 & [5,5] & [5,5] & 3 & [3,5] & [3,5] \\
\midrule
Average width & & 1.8 & 1.8 & & 3.8 & 3.8 & & 1.4 & 2.6 & & 1.8 & 3.0 \\
\bottomrule
\end{tabular}
 
 \end{center}
 \begin{spacing}{1}
 {\footnotesize
 \textbf{Data:} PSID.\\
 \textbf{Note:} This table reports dominance sets for the birth cohorts 1940s--1980s for two definitions of family income: the baseline measure of pre-tax-and-transfer labor income and post-tax-and-transfer disposable income. For each income definition, we calculate equivalized income as the ratio of total family income to the square root of the number of household members. \textit{Rank} reports the ranking by $V_{0,u}$; rank 1 (5) corresponds to the highest (lowest) $V_{0,u}$.}
 \end{spacing}
 \end{table}

 \begin{table}[H]
 \footnotesize
 \begin{center}
 \caption{Dominance sets for birth cohorts accounting for statistical uncertainty}
     \label{tab:cohorts-cs}
 \renewcommand{\arraystretch}{1}
  \setlength{\tabcolsep}{2pt}
\begin{tabular}{lc cc *{3}{c} *{3}{c} cc cc}
\toprule
 & & \multicolumn{2}{c}{Normative} & \multicolumn{6}{c}{Statistical (95\%)} & \multicolumn{4}{c}{Combined (95\%)} \\
\cmidrule(lr){3-4} \cmidrule(lr){5-10} \cmidrule(lr){11-14}
 & & \multicolumn{2}{c}{} & \multicolumn{3}{c}{Marginal} & \multicolumn{3}{c}{Simultaneous} & \multicolumn{2}{c}{Marginal} & \multicolumn{2}{c}{Simultaneous} \\
\cmidrule(lr){5-7} \cmidrule(lr){8-10} \cmidrule(lr){11-12} \cmidrule(lr){13-14}
Cohort & Rank & \(R^{[0,2]}\) & \(R^{[0,\infty)}\) & \(\theta=0\) & \(\theta=2\) & \(\theta\to\infty\) & \(\theta=0\) & \(\theta=2\) & \(\theta\to\infty\) & \(R^{[0,2]}\) & \(R^{[0,\infty)}\) & \(R^{[0,2]}\) & \(R^{[0,\infty)}\) \\
\midrule
1940s & 1 & [1,1] & [1,1] & [1,5] & [1,2] & [1,1] & [1,5] & [1,2] & [1,1] & [1,5] & [1,5] & [1,5] & [1,5] \\
1950s & 4 & [2,4] & [2,4] & [1,5] & [1,4] & [2,2] & [1,5] & [1,4] & [2,2] & [1,5] & [1,5] & [1,5] & [1,5] \\
1960s & 3 & [3,4] & [3,4] & [1,5] & [2,4] & [3,4] & [1,5] & [2,4] & [3,4] & [1,5] & [1,5] & [1,5] & [1,5] \\
1970s & 2 & [2,3] & [2,3] & [1,4] & [2,4] & [3,4] & [1,5] & [2,4] & [3,4] & [1,5] & [1,5] & [1,5] & [1,5] \\
1980s & 5 & [5,5] & [5,5] & [1,5] & [5,5] & [5,5] & [1,5] & [5,5] & [5,5] & [1,5] & [1,5] & [1,5] & [1,5] \\
\midrule
Average width & & 1.8 & 1.8 & 4.8 & 2.6 & 1.4 & 5.0 & 2.6 & 1.4 & 5.0 & 5.0 & 5.0 & 5.0 \\
\bottomrule
\end{tabular}
 
 \end{center}
 \begin{spacing}{1}
 {\footnotesize
 \textbf{Data:} PSID.\\
 \textbf{Note:} This table reports normative, statistical, and combined sets for the birth cohorts 1940s--1980s. Details of the construction are discussed in Appendix~\ref{app: confsets}. The normative dominance sets $R^{[0,2]}$ and $R^{[0,\infty)}$ are the baseline results reported in Figure~\ref{fig:rol}. Statistical confidence sets account for sampling variation in the parameter estimates in Equation \eqref{ige} used to compute expected type utilities $U(\pi_s)$. They are reported for the two limits of inequality-of-opportunity aversion $\theta$ and an intermediate level ($\theta=2$). Following \citet{MRSW24}, we report marginal confidence sets for the rankings of individual cohorts and simultaneous confidence sets for the full ranking profile. Combined sets apply the normative dominance evaluation over both $\theta \in [0,2]$ and $\theta \in [0,\infty)$ to the marginal and simultaneous confidence sets. \textit{Rank} reports the ranking by $V_{0,u}$; rank 1 (5) corresponds to the highest (lowest) $V_{0,u}$.}
 \end{spacing}
 \end{table}

\begin{table}[H]
 \footnotesize
 \begin{center}
 \caption{Dominance sets and sharp sets for birth cohorts}
  \label{tab:cohorts-sharp}
 \renewcommand{\arraystretch}{1}
  \setlength{\tabcolsep}{2pt}
\begin{tabular}{lc*{2}{cc}}
\toprule
 & & \multicolumn{2}{c}{Dominance sets \(R_j\)} & \multicolumn{2}{c}{Sharp sets \(S_j\)} \\
\cmidrule(lr){3-4} \cmidrule(lr){5-6}
Cohort & Rank & \(R^{[0,2]}\) & \(R^{[0,\infty)}\) & \(S^{[0,2]}\) & \(S^{[0,\infty)}\) \\
\midrule
1940s & 1 & [1,1] & [1,1] & [1,1] & [1,1] \\
1950s & 4 & [2,4] & [2,4] & [2,4] & [2,4] \\
1960s & 3 & [3,4] & [3,4] & [3,4] & [3,4] \\
1970s & 2 & [2,3] & [2,3] & [2,3] & [2,3] \\
1980s & 5 & [5,5] & [5,5] & [5,5] & [5,5] \\
\midrule
Average width & & 1.8 & 1.8 & 1.8 & 1.8 \\
\bottomrule
\end{tabular}

 \end{center}
  \begin{spacing}{1}
 {\footnotesize
 \textbf{Data:} PSID.\\
 \textbf{Note:} This table compares the dominance sets $R^{[0,2]}$ and $R^{[0,\infty)}$ reported in Figure~\ref{fig:rol} with the sharp attainable rank sets $S^{[0,2]}$ and $S^{[0,\infty)}$ for the birth cohorts 1940s--1980s. Dominance sets contain every rank a cohort could hold in any complete ranking consistent with the unambiguous pairwise dominance preorder over the respective range of $\theta$. Sharp sets span only the ranks a cohort actually attains across that range. \textit{Rank} reports the ranking by $V_{0,u}$; rank 1 (5) corresponds to the highest (lowest) $V_{0,u}$.}
 \end{spacing}
 \end{table}

\begin{table}[H]
\footnotesize
\begin{center}
\caption{Dominance sets for the 50 largest CZs under alternative trimming of extreme values}
     \label{tab:cz-trim}
\renewcommand{\arraystretch}{0.5}
\begin{tabular}{l*{3}{ccc}}
\toprule
  & \multicolumn{3}{c}{Baseline} & \multicolumn{3}{c}{Trim \(y^c[p1,p99]\)} & \multicolumn{3}{c}{Trim \(y^c[p5,p95]\)} \\
 \cmidrule(lr){2-4} \cmidrule(lr){5-7} \cmidrule(lr){8-10}
Commuting zone & Rank & \(R^{[0,2]}\) & \(R^{[0,\infty)}\) & Rank & \(R^{[0,2]}\) & \(R^{[0,\infty)}\) & Rank & \(R^{[0,2]}\) & \(R^{[0,\infty)}\) \\
\midrule
Minneapolis & 1 & [1,2] & [1,8] & 1 & [1,1] & [1,5] & 1 & [1,1] & [1,2] \\
Salt Lake City & 2 & [1,2] & [1,2] & 2 & [2,2] & [1,2] & 2 & [2,2] & [1,2] \\
Boston & 3 & [3,3] & [3,9] & 3 & [3,3] & [3,8] & 3 & [3,3] & [3,5] \\
Newark & 4 & [4,7] & [4,17] & 4 & [4,7] & [4,14] & 4 & [4,7] & [4,12] \\
Pittsburgh & 5 & [4,6] & [2,12] & 5 & [4,6] & [3,10] & 5 & [4,5] & [4,7] \\
Washington DC & 6 & [5,7] & [4,15] & 6 & [5,7] & [5,13] & 6 & [5,6] & [5,11] \\
Bridgeport & 7 & [7,8] & [7,24] & 7 & [7,8] & [7,21] & 7 & [7,8] & [7,18] \\
Manchester & 8 & [4,8] & [2,10] & 8 & [4,8] & [2,9] & 8 & [6,8] & [3,8] \\
Milwaukee & 9 & [9,21] & [9,38] & 9 & [9,19] & [9,36] & 9 & [9,16] & [9,32] \\
Denver & 10 & [9,10] & [6,16] & 10 & [9,10] & [6,15] & 10 & [9,10] & [7,14] \\
Providence & 11 & [10,14] & [9,22] & 11 & [10,14] & [9,19] & 11 & [10,14] & [9,15] \\
Buffalo & 12 & [11,16] & [11,30] & 12 & [11,16] & [11,26] & 13 & [12,16] & [12,23] \\
Grand Rapids & 13 & [11,16] & [11,33] & 13 & [11,16] & [11,31] & 14 & [12,16] & [12,26] \\
Seattle & 14 & [10,15] & [5,18] & 14 & [10,14] & [5,18] & 12 & [10,12] & [6,15] \\
San Francisco & 15 & [11,16] & [3,17] & 15 & [11,16] & [4,17] & 15 & [11,15] & [6,16] \\
San Jose & 16 & [10,16] & [2,16] & 16 & [11,16] & [2,16] & 16 & [12,16] & [4,16] \\
Philadelphia & 17 & [17,22] & [16,35] & 17 & [17,22] & [16,33] & 17 & [17,19] & [16,32] \\
Kansas City & 18 & [16,22] & [15,32] & 18 & [16,20] & [15,32] & 18 & [17,18] & [16,29] \\
Fort Worth & 19 & [16,20] & [11,24] & 19 & [16,19] & [12,23] & 19 & [18,19] & [13,23] \\
Chicago & 20 & [20,26] & [19,36] & 20 & [20,24] & [19,36] & 20 & [20,23] & [20,33] \\
Houston & 21 & [17,23] & [9,24] & 21 & [19,22] & [11,24] & 21 & [20,22] & [13,24] \\
St. Louis & 22 & [22,30] & [22,43] & 22 & [22,29] & [22,41] & 23 & [22,26] & [22,40] \\
Baltimore & 23 & [23,31] & [22,41] & 23 & [23,30] & [22,40] & 22 & [22,27] & [22,39] \\
Portland & 24 & [16,24] & [12,29] & 24 & [17,24] & [12,28] & 24 & [20,24] & [13,26] \\
Cincinnati & 25 & [25,33] & [25,49] & 25 & [25,33] & [25,48] & 25 & [25,30] & [25,44] \\
Dallas & 26 & [22,27] & [17,36] & 26 & [22,27] & [18,36] & 26 & [23,27] & [19,33] \\
Indianapolis & 27 & [26,33] & [26,46] & 27 & [26,33] & [26,44] & 27 & [27,32] & [26,44] \\
San Diego & 28 & [18,28] & [3,28] & 28 & [21,28] & [4,28] & 28 & [22,28] & [8,28] \\
Austin & 29 & [24,30] & [18,37] & 29 & [25,29] & [19,37] & 30 & [26,30] & [19,34] \\
New York & 30 & [26,31] & [15,34] & 30 & [26,31] & [17,34] & 29 & [27,30] & [18,33] \\
Columbus & 31 & [31,36] & [30,47] & 31 & [31,35] & [30,45] & 31 & [31,34] & [29,44] \\
Cleveland & 32 & [32,37] & [28,47] & 32 & [32,36] & [29,45] & 33 & [33,36] & [30,44] \\
Dayton & 33 & [31,36] & [27,45] & 33 & [31,36] & [27,44] & 32 & [31,34] & [27,42] \\
Sacramento & 34 & [24,34] & [14,36] & 34 & [26,34] & [16,35] & 34 & [29,34] & [18,35] \\
Phoenix & 35 & [28,35] & [19,40] & 35 & [28,35] & [20,39] & 35 & [31,35] & [20,36] \\
Detroit & 36 & [36,38] & [32,47] & 36 & [36,37] & [35,46] & 36 & [36,37] & [36,45] \\
Raleigh & 37 & [37,42] & [36,48] & 37 & [37,41] & [36,48] & 37 & [37,39] & [37,46] \\
Nashville & 38 & [37,41] & [28,46] & 38 & [37,40] & [29,46] & 38 & [37,39] & [31,46] \\
Los Angeles & 39 & [29,39] & [6,39] & 39 & [31,39] & [9,39] & 39 & [35,39] & [14,39] \\
San Antonio & 40 & [38,41] & [20,42] & 40 & [39,41] & [21,42] & 41 & [41,41] & [28,42] \\
Las Vegas & 41 & [35,41] & [20,43] & 41 & [38,41] & [21,42] & 40 & [40,40] & [27,41] \\
Charlotte & 42 & [42,47] & [42,50] & 42 & [42,46] & [42,50] & 42 & [42,44] & [42,50] \\
Miami & 43 & [41,43] & [13,43] & 43 & [42,43] & [18,43] & 43 & [42,43] & [23,43] \\
Atlanta & 44 & [44,48] & [42,50] & 45 & [45,48] & [44,50] & 46 & [46,48] & [45,50] \\
Port St. Lucie & 45 & [43,46] & [29,46] & 44 & [43,45] & [30,45] & 44 & [44,45] & [33,45] \\
Fresno & 46 & [43,46] & [22,46] & 46 & [43,46] & [26,46] & 45 & [43,45] & [30,45] \\
Tampa & 47 & [45,47] & [31,47] & 47 & [45,47] & [33,47] & 47 & [46,47] & [35,47] \\
Orlando & 48 & [46,48] & [33,48] & 48 & [47,48] & [34,48] & 48 & [47,48] & [36,48] \\
Jacksonville & 49 & [49,49] & [37,49] & 49 & [49,49] & [40,49] & 49 & [49,49] & [47,49] \\
New Orleans & 50 & [50,50] & [41,50] & 50 & [50,50] & [46,50] & 50 & [50,50] & [48,50] \\
\midrule
Average width & & 5.6 & 17.6 & & 4.8 & 15.7 & & 3.5 & 12.8 \\
\bottomrule
\end{tabular}
 
\end{center}
\begin{spacing}{1}
{\footnotesize
 \textbf{Data:} Opportunity Insights.\\
 \textbf{Note:} This table reports dominance sets for the 50 largest CZs conditional on trimming the lower and upper tails of the child income distribution. The baseline uses the full distribution of positive child income $y^c$. In the other columns, $y^c$ is trimmed at the 1st and 99th percentiles, or the 5th and 95th percentiles, of the CZ-specific distribution, respectively. Sampling weights are renormalized in each CZ. \textit{Rank} reports the ranking by $V_{0,u}$; rank 1 (50) corresponds to the highest (lowest) $V_{0,u}$.}
 \end{spacing}
 \end{table}
 
 \begin{table}[H]
\footnotesize
\begin{center}
\caption{Dominance sets for the 50 largest CZs under alternative functional forms}
     \label{tab:cz-ff}
\renewcommand{\arraystretch}{0.5}
\begin{tabular}{lc*{4}{cc}}
\toprule
  & & \multicolumn{2}{c}{Linear} & \multicolumn{2}{c}{Quadratic} & \multicolumn{2}{c}{Cubic} & \multicolumn{2}{c}{\shortstack{Non-\\parametric}} \\
 \cmidrule(lr){3-4} \cmidrule(lr){5-6} \cmidrule(lr){7-8} \cmidrule(lr){9-10}
Commuting zone & Rank & \(R^{[0,2]}\) & \(R^{[0,\infty)}\) & \(R^{[0,2]}\) & \(R^{[0,\infty)}\) & \(R^{[0,2]}\) & \(R^{[0,\infty)}\) & \(R^{[0,2]}\) & \(R^{[0,\infty)}\) \\
\midrule
Minneapolis & 1 & [1,2] & [1,8] & [1,2] & [1,19] & [1,2] & [1,21] & [1,2] & [1,22] \\
Salt Lake City & 2 & [1,2] & [1,2] & [1,2] & [1,3] & [1,2] & [1,3] & [1,2] & [1,3] \\
Boston & 3 & [3,3] & [3,9] & [3,3] & [2,11] & [3,3] & [2,12] & [3,3] & [2,11] \\
Newark & 4 & [4,7] & [4,17] & [4,6] & [3,15] & [4,6] & [3,15] & [4,6] & [3,15] \\
Pittsburgh & 5 & [4,6] & [2,12] & [5,7] & [4,19] & [5,7] & [4,19] & [5,7] & [4,20] \\
Washington DC & 6 & [5,7] & [4,15] & [4,7] & [2,12] & [4,7] & [2,12] & [4,7] & [2,13] \\
Bridgeport & 7 & [7,8] & [7,24] & [7,8] & [7,27] & [7,8] & [7,27] & [7,8] & [7,27] \\
Manchester & 8 & [4,8] & [2,10] & [4,8] & [2,12] & [4,8] & [2,12] & [4,8] & [2,13] \\
Milwaukee & 9 & [9,21] & [9,38] & [9,22] & [9,48] & [9,22] & [9,48] & [9,22] & [9,48] \\
Denver & 10 & [9,10] & [6,16] & [9,11] & [5,19] & [9,11] & [5,19] & [9,11] & [5,19] \\
Providence & 11 & [10,14] & [9,22] & [10,14] & [9,22] & [10,14] & [8,23] & [10,14] & [8,23] \\
Buffalo & 12 & [11,16] & [11,30] & [11,16] & [11,33] & [11,17] & [11,33] & [11,17] & [11,34] \\
Grand Rapids & 13 & [11,16] & [11,33] & [11,16] & [11,36] & [11,16] & [11,36] & [11,16] & [11,36] \\
Seattle & 14 & [10,15] & [5,18] & [10,15] & [5,20] & [10,15] & [5,21] & [10,15] & [5,21] \\
San Francisco & 15 & [11,16] & [3,17] & [11,16] & [2,17] & [11,16] & [2,17] & [11,16] & [2,17] \\
San Jose & 16 & [10,16] & [2,16] & [9,16] & [1,16] & [9,16] & [1,16] & [9,16] & [1,16] \\
Philadelphia & 17 & [17,22] & [16,35] & [16,22] & [15,36] & [16,22] & [15,35] & [16,22] & [15,35] \\
Kansas City & 18 & [16,22] & [15,32] & [16,22] & [15,35] & [16,22] & [15,35] & [16,22] & [14,35] \\
Fort Worth & 19 & [16,20] & [11,24] & [16,20] & [10,26] & [16,20] & [10,26] & [16,20] & [10,26] \\
Chicago & 20 & [20,26] & [19,36] & [20,26] & [17,36] & [20,27] & [18,36] & [20,27] & [18,36] \\
Houston & 21 & [17,23] & [9,24] & [17,23] & [5,24] & [17,23] & [6,25] & [17,23] & [5,25] \\
St. Louis & 22 & [22,30] & [22,43] & [22,31] & [21,47] & [22,31] & [21,47] & [22,31] & [21,47] \\
Baltimore & 23 & [23,31] & [22,41] & [22,31] & [21,41] & [22,31] & [21,39] & [22,31] & [21,39] \\
Portland & 24 & [16,24] & [12,29] & [16,24] & [10,29] & [15,24] & [10,29] & [15,24] & [10,29] \\
Cincinnati & 25 & [25,33] & [25,49] & [25,35] & [25,50] & [25,35] & [25,50] & [25,35] & [25,50] \\
Dallas & 26 & [22,27] & [17,36] & [22,27] & [16,35] & [22,27] & [16,35] & [22,27] & [16,35] \\
Indianapolis & 27 & [26,33] & [26,46] & [26,33] & [24,45] & [26,33] & [24,45] & [26,33] & [24,45] \\
San Diego & 28 & [18,28] & [3,28] & [18,28] & [3,28] & [18,28] & [3,28] & [18,28] & [3,28] \\
Austin & 29 & [24,30] & [18,37] & [24,30] & [18,37] & [24,30] & [18,37] & [24,30] & [18,37] \\
New York & 30 & [26,31] & [15,34] & [26,31] & [10,33] & [25,31] & [8,33] & [25,31] & [9,33] \\
Columbus & 31 & [31,36] & [30,47] & [31,37] & [29,49] & [31,37] & [30,50] & [31,37] & [30,50] \\
Cleveland & 32 & [32,37] & [28,47] & [32,38] & [28,48] & [32,38] & [28,49] & [32,38] & [28,49] \\
Dayton & 33 & [31,36] & [27,45] & [30,37] & [28,48] & [30,37] & [28,49] & [30,37] & [28,49] \\
Sacramento & 34 & [24,34] & [14,36] & [24,34] & [11,36] & [24,34] & [11,36] & [24,34] & [11,36] \\
Phoenix & 35 & [28,35] & [19,40] & [28,35] & [19,40] & [28,35] & [19,40] & [28,35] & [19,40] \\
Detroit & 36 & [36,38] & [32,47] & [34,38] & [28,46] & [34,38] & [29,47] & [34,38] & [29,47] \\
Raleigh & 37 & [37,42] & [36,48] & [37,42] & [33,49] & [37,42] & [31,49] & [37,42] & [31,49] \\
Nashville & 38 & [37,41] & [28,46] & [37,42] & [30,49] & [37,42] & [30,49] & [37,42] & [30,49] \\
Los Angeles & 39 & [29,39] & [6,39] & [29,39] & [5,39] & [29,39] & [5,39] & [29,39] & [5,39] \\
San Antonio & 40 & [38,41] & [20,42] & [38,41] & [20,42] & [38,41] & [18,42] & [38,41] & [17,42] \\
Las Vegas & 41 & [35,41] & [20,43] & [33,41] & [20,42] & [33,41] & [20,42] & [33,41] & [19,42] \\
Charlotte & 42 & [42,47] & [42,50] & [42,47] & [41,50] & [42,47] & [39,50] & [42,47] & [39,50] \\
Miami & 43 & [41,43] & [13,43] & [40,43] & [9,43] & [40,43] & [7,43] & [40,43] & [6,43] \\
Atlanta & 44 & [44,48] & [42,50] & [44,48] & [32,49] & [44,48] & [32,49] & [44,48] & [32,49] \\
Port St. Lucie & 45 & [43,46] & [29,46] & [43,46] & [27,46] & [43,46] & [27,46] & [43,46] & [27,46] \\
Fresno & 46 & [43,46] & [22,46] & [43,46] & [23,46] & [43,46] & [23,46] & [43,46] & [23,46] \\
Tampa & 47 & [45,47] & [31,47] & [45,47] & [29,47] & [45,47] & [29,47] & [45,47] & [29,47] \\
Orlando & 48 & [46,48] & [33,48] & [46,48] & [30,48] & [46,48] & [30,48] & [46,48] & [30,48] \\
Jacksonville & 49 & [49,49] & [37,49] & [49,49] & [36,49] & [49,49] & [37,49] & [49,49] & [37,49] \\
New Orleans & 50 & [50,50] & [41,50] & [50,50] & [40,50] & [50,50] & [39,50] & [50,50] & [39,50] \\
\midrule
Average width & & 5.6 & 17.6 & 6.0 & 19.9 & 6.1 & 20.2 & 6.1 & 20.3 \\
\bottomrule
\end{tabular}
 
\end{center}
\begin{spacing}{1}
{\footnotesize
 \textbf{Data:} Opportunity Insights.\\
 \textbf{Note:} This table reports dominance sets for the 50 largest CZs under different functional forms of Equation \eqref{eq:irs}. The linear specification replicates the baseline results of Figure~\ref{fig:rol_cz}. The quadratic and cubic specifications contain second- and third-order polynomials in $r^p$, respectively. The nonparametric specification estimates $U(\pi_s)$ as the weighted mean of $\ln y^c$ within types. \textit{Rank} reports the ranking by $V_{0,u}$; rank 1 (50) corresponds to the highest (lowest) $V_{0,u}$.}
 \end{spacing}
 \end{table}

 \begin{table}[H]
\footnotesize
\begin{center}
\caption{Dominance sets for the 50 largest CZs under alternative percentile-income schedules}
     \label{tab:cz-val}
\renewcommand{\arraystretch}{0.5}
\begin{tabular}{l*{2}{ccc}}
\toprule
  & \multicolumn{3}{c}{\shortstack{Baseline\\(Chetty et al., 2014)}} & \multicolumn{3}{c}{\shortstack{Alternative\\(Chetty et al., 2026)}} \\
 \cmidrule(lr){2-4} \cmidrule(lr){5-7}
Commuting zone & Rank & \(R^{[0,2]}\) & \(R^{[0,\infty)}\) & Rank & \(R^{[0,2]}\) & \(R^{[0,\infty)}\) \\
\midrule
Minneapolis & 1 & [1,2] & [1,8] & 1 & [1,2] & [1,8] \\
Salt Lake City & 2 & [1,2] & [1,2] & 2 & [1,2] & [1,2] \\
Boston & 3 & [3,3] & [3,9] & 3 & [3,3] & [3,9] \\
Newark & 4 & [4,7] & [4,17] & 4 & [4,7] & [4,17] \\
Pittsburgh & 5 & [4,6] & [2,12] & 5 & [4,7] & [2,12] \\
Washington DC & 6 & [5,7] & [4,15] & 6 & [4,7] & [4,15] \\
Bridgeport & 7 & [7,8] & [7,24] & 7 & [7,8] & [7,24] \\
Manchester & 8 & [4,8] & [2,10] & 8 & [4,8] & [2,10] \\
Milwaukee & 9 & [9,21] & [9,38] & 9 & [9,21] & [9,38] \\
Denver & 10 & [9,10] & [6,16] & 10 & [9,11] & [6,17] \\
Providence & 11 & [10,14] & [9,22] & 11 & [10,14] & [9,22] \\
Buffalo & 12 & [11,16] & [11,30] & 12 & [11,15] & [11,29] \\
Grand Rapids & 13 & [11,16] & [11,33] & 13 & [12,16] & [12,34] \\
Seattle & 14 & [10,15] & [5,18] & 14 & [10,15] & [5,18] \\
San Francisco & 15 & [11,16] & [3,17] & 15 & [11,16] & [3,18] \\
San Jose & 16 & [10,16] & [2,16] & 16 & [9,16] & [2,16] \\
Philadelphia & 17 & [17,22] & [16,35] & 17 & [17,22] & [16,35] \\
Kansas City & 18 & [16,22] & [15,32] & 18 & [16,22] & [15,33] \\
Fort Worth & 19 & [16,20] & [11,24] & 19 & [16,20] & [11,24] \\
Chicago & 20 & [20,26] & [19,36] & 20 & [20,29] & [19,36] \\
Houston & 21 & [17,23] & [9,24] & 21 & [17,23] & [9,24] \\
St. Louis & 22 & [22,30] & [22,43] & 22 & [22,30] & [22,43] \\
Baltimore & 23 & [23,31] & [22,41] & 23 & [23,31] & [22,41] \\
Portland & 24 & [16,24] & [12,29] & 24 & [16,24] & [12,29] \\
Cincinnati & 25 & [25,33] & [25,49] & 25 & [25,35] & [25,49] \\
Dallas & 26 & [22,27] & [17,36] & 26 & [22,28] & [17,36] \\
Indianapolis & 27 & [26,33] & [26,46] & 27 & [26,33] & [26,46] \\
San Diego & 28 & [18,28] & [3,28] & 28 & [18,28] & [3,28] \\
Austin & 29 & [24,30] & [18,37] & 29 & [24,30] & [18,37] \\
New York & 30 & [26,31] & [15,34] & 30 & [25,31] & [15,33] \\
Columbus & 31 & [31,36] & [30,47] & 31 & [31,36] & [30,47] \\
Cleveland & 32 & [32,37] & [28,47] & 32 & [32,39] & [28,47] \\
Dayton & 33 & [31,36] & [27,45] & 33 & [30,37] & [27,45] \\
Sacramento & 34 & [24,34] & [14,36] & 34 & [23,34] & [14,36] \\
Phoenix & 35 & [28,35] & [19,40] & 35 & [27,36] & [20,40] \\
Detroit & 36 & [36,38] & [32,47] & 36 & [35,39] & [32,47] \\
Raleigh & 37 & [37,42] & [36,48] & 37 & [37,42] & [36,49] \\
Nashville & 38 & [37,41] & [28,46] & 38 & [37,42] & [28,46] \\
Los Angeles & 39 & [29,39] & [6,39] & 39 & [27,39] & [5,39] \\
San Antonio & 40 & [38,41] & [20,42] & 40 & [36,41] & [20,42] \\
Las Vegas & 41 & [35,41] & [20,43] & 41 & [33,41] & [20,43] \\
Charlotte & 42 & [42,47] & [42,50] & 42 & [42,47] & [42,50] \\
Miami & 43 & [41,43] & [13,43] & 43 & [40,43] & [12,43] \\
Atlanta & 44 & [44,48] & [42,50] & 44 & [44,49] & [41,50] \\
Port St. Lucie & 45 & [43,46] & [29,46] & 45 & [43,46] & [28,46] \\
Fresno & 46 & [43,46] & [22,46] & 46 & [43,46] & [22,46] \\
Tampa & 47 & [45,47] & [31,47] & 47 & [45,47] & [31,47] \\
Orlando & 48 & [46,48] & [33,48] & 48 & [46,48] & [32,48] \\
Jacksonville & 49 & [49,49] & [37,49] & 49 & [48,49] & [37,49] \\
New Orleans & 50 & [50,50] & [41,50] & 50 & [50,50] & [41,50] \\
\midrule
Average width & & 5.6 & 17.6 & & 6.2 & 17.7 \\
\bottomrule
\end{tabular}
 
\end{center}
\begin{spacing}{1}
{\footnotesize
 \textbf{Data:} Opportunity Insights.\\
 \textbf{Note:} This table reports dominance sets for the 50 largest CZs using the different percentile-income schedules reported by \citet{Chetty2014} (baseline) and \citet{Chetty2026}. \textit{Rank} reports the ranking by $V_{0,u}$; rank 1 (50) corresponds to the highest (lowest) $V_{0,u}$.}
 \end{spacing}
 \end{table}

 \begin{table}[H]
 \scriptsize
 \begin{center}
 \caption{Dominance sets for the 50 largest CZs accounting for statistical uncertainty}
     \label{tab:cz-cs}
 \renewcommand{\arraystretch}{0.5}
 \setlength{\tabcolsep}{0.5pt}
\footnotesize
\begin{tabular}{lc cc *{3}{c} *{3}{c} cc cc}
\toprule
 & & \multicolumn{2}{c}{Normative} & \multicolumn{6}{c}{Statistical (95\%)} & \multicolumn{4}{c}{Combined (95\%)} \\
\cmidrule(lr){3-4} \cmidrule(lr){5-10} \cmidrule(lr){11-14}
 & & \multicolumn{2}{c}{} & \multicolumn{3}{c}{Marginal} & \multicolumn{3}{c}{Simultaneous} & \multicolumn{2}{c}{Marginal} & \multicolumn{2}{c}{Simultaneous} \\
\cmidrule(lr){5-7} \cmidrule(lr){8-10} \cmidrule(lr){11-12} \cmidrule(lr){13-14}
Commuting zone & Rank & \(R^{[0,2]}\) & \(R^{[0,\infty)}\) & \(\theta=0\) & \(\theta=2\) & \(\theta\to\infty\) & \(\theta=0\) & \(\theta=2\) & \(\theta\to\infty\) & \(R^{[0,2]}\) & \(R^{[0,\infty)}\) & \(R^{[0,2]}\) & \(R^{[0,\infty)}\) \\
\midrule
Minneapolis & 1 & [1,2] & [1,8] & [1,1] & [1,2] & [5,10] & [1,1] & [1,2] & [4,12] & [1,2] & [1,11] & [1,2] & [1,12] \\
Salt Lake City & 2 & [1,2] & [1,2] & [2,2] & [1,2] & [1,1] & [2,2] & [1,2] & [1,1] & [1,2] & [1,2] & [1,2] & [1,2] \\
Boston & 3 & [3,3] & [3,9] & [3,3] & [3,4] & [7,10] & [3,3] & [3,4] & [5,12] & [3,4] & [2,12] & [3,4] & [2,12] \\
Newark & 4 & [4,7] & [4,17] & [4,5] & [5,7] & [12,20] & [4,5] & [5,7] & [12,20] & [4,7] & [4,20] & [4,7] & [4,20] \\
Pittsburgh & 5 & [4,6] & [2,12] & [4,5] & [5,7] & [4,10] & [4,7] & [4,7] & [4,11] & [4,7] & [2,14] & [4,8] & [2,15] \\
Washington DC & 6 & [5,7] & [4,15] & [5,7] & [5,7] & [11,18] & [5,7] & [5,7] & [10,20] & [4,8] & [4,20] & [4,8] & [4,21] \\
Bridgeport & 7 & [7,8] & [7,24] & [6,7] & [8,8] & [22,28] & [5,8] & [8,10] & [21,28] & [5,9] & [5,28] & [5,10] & [5,28] \\
Manchester & 8 & [4,8] & [2,10] & [8,8] & [3,4] & [3,6] & [7,8] & [3,5] & [2,9] & [3,8] & [2,12] & [3,8] & [2,12] \\
Milwaukee & 9 & [9,21] & [9,38] & [9,13] & [20,22] & [36,43] & [9,13] & [19,22] & [36,43] & [9,22] & [9,43] & [9,22] & [9,44] \\
Denver & 10 & [9,10] & [6,16] & [9,13] & [9,12] & [10,18] & [9,13] & [8,12] & [8,19] & [8,16] & [4,22] & [8,16] & [3,24] \\
Providence & 11 & [10,14] & [9,22] & [9,14] & [13,13] & [13,21] & [9,14] & [13,13] & [12,21] & [9,16] & [6,25] & [9,16] & [6,26] \\
Buffalo & 12 & [11,16] & [11,30] & [9,14] & [14,17] & [24,28] & [9,15] & [14,17] & [22,28] & [9,18] & [8,30] & [9,18] & [8,30] \\
Grand Rapids & 13 & [11,16] & [11,33] & [9,14] & [14,17] & [29,34] & [9,16] & [14,17] & [29,34] & [9,18] & [9,35] & [9,18] & [9,35] \\
Seattle & 14 & [10,15] & [5,18] & [11,15] & [9,12] & [7,12] & [11,16] & [9,12] & [7,15] & [9,16] & [3,19] & [9,16] & [3,21] \\
San Francisco & 15 & [11,16] & [3,17] & [14,17] & [9,12] & [4,8] & [12,18] & [9,12] & [4,9] & [9,17] & [2,19] & [9,18] & [2,20] \\
San Jose & 16 & [10,16] & [2,16] & [15,18] & [9,12] & [2,3] & [13,18] & [8,12] & [2,4] & [9,18] & [2,19] & [8,18] & [2,19] \\
Philadelphia & 17 & [17,22] & [16,35] & [15,18] & [21,22] & [33,34] & [15,18] & [21,23] & [31,34] & [14,23] & [13,35] & [14,23] & [13,35] \\
Kansas City & 18 & [16,22] & [15,32] & [16,18] & [18,21] & [29,33] & [15,19] & [16,21] & [29,33] & [15,22] & [14,35] & [14,22] & [13,35] \\
Fort Worth & 19 & [16,20] & [11,24] & [19,21] & [14,20] & [11,20] & [18,22] & [14,20] & [10,20] & [14,23] & [9,27] & [14,24] & [8,28] \\
Chicago & 20 & [20,26] & [19,36] & [19,21] & [24,28] & [35,35] & [19,22] & [23,28] & [35,36] & [19,29] & [18,38] & [19,29] & [18,39] \\
Houston & 21 & [17,23] & [9,24] & [19,23] & [17,20] & [10,15] & [19,23] & [16,21] & [7,16] & [16,25] & [6,28] & [16,25] & [5,28] \\
St. Louis & 22 & [22,30] & [22,43] & [21,23] & [30,31] & [37,47] & [19,23] & [29,32] & [36,47] & [20,31] & [18,47] & [19,32] & [18,49] \\
Baltimore & 23 & [23,31] & [22,41] & [21,24] & [30,32] & [36,43] & [21,25] & [30,34] & [36,46] & [21,33] & [20,46] & [21,34] & [20,47] \\
Portland & 24 & [16,24] & [12,29] & [23,26] & [14,18] & [15,21] & [23,27] & [14,20] & [13,23] & [14,27] & [10,33] & [14,28] & [9,35] \\
Cincinnati & 25 & [25,33] & [25,49] & [24,27] & [32,36] & [46,50] & [23,27] & [31,37] & [45,50] & [24,36] & [24,50] & [23,38] & [22,50] \\
Dallas & 26 & [22,27] & [17,36] & [24,27] & [23,25] & [24,28] & [24,27] & [23,26] & [23,28] & [21,30] & [15,37] & [21,30] & [15,37] \\
Indianapolis & 27 & [26,33] & [26,46] & [25,29] & [31,36] & [41,48] & [24,30] & [30,36] & [39,49] & [24,36] & [23,49] & [22,36] & [21,49] \\
San Diego & 28 & [18,28] & [3,28] & [27,30] & [16,20] & [2,4] & [27,30] & [16,20] & [2,4] & [16,30] & [2,30] & [16,30] & [2,30] \\
Austin & 29 & [24,30] & [18,37] & [27,31] & [23,27] & [22,28] & [27,31] & [22,28] & [22,28] & [22,32] & [16,38] & [21,33] & [16,38] \\
New York & 30 & [26,31] & [15,34] & [28,30] & [24,28] & [15,20] & [27,31] & [24,28] & [14,20] & [23,33] & [12,36] & [23,33] & [10,36] \\
Columbus & 31 & [31,36] & [30,47] & [30,33] & [32,38] & [41,48] & [29,35] & [32,40] & [39,49] & [28,39] & [26,49] & [27,40] & [25,49] \\
Cleveland & 32 & [32,37] & [28,47] & [31,35] & [34,40] & [37,47] & [31,35] & [33,40] & [36,47] & [31,40] & [26,48] & [30,40] & [26,48] \\
Dayton & 33 & [31,36] & [27,45] & [31,36] & [32,36] & [36,43] & [31,37] & [31,38] & [36,46] & [29,38] & [26,47] & [28,39] & [25,48] \\
Sacramento & 34 & [24,34] & [14,36] & [32,37] & [23,26] & [12,20] & [31,37] & [23,27] & [11,20] & [23,37] & [10,39] & [23,37] & [10,39] \\
Phoenix & 35 & [28,35] & [19,40] & [32,37] & [26,29] & [22,28] & [31,37] & [25,29] & [22,28] & [24,38] & [17,42] & [24,38] & [17,42] \\
Detroit & 36 & [36,38] & [32,47] & [33,37] & [35,40] & [41,47] & [33,37] & [34,40] & [39,48] & [31,41] & [27,49] & [30,41] & [27,49] \\
Raleigh & 37 & [37,42] & [36,48] & [34,38] & [39,42] & [42,49] & [33,38] & [39,42] & [41,49] & [34,42] & [29,49] & [33,42] & [28,50] \\
Nashville & 38 & [37,41] & [28,46] & [37,38] & [37,42] & [36,41] & [37,39] & [35,42] & [35,43] & [34,42] & [28,48] & [34,42] & [27,48] \\
Los Angeles & 39 & [29,39] & [6,39] & [39,39] & [28,29] & [4,8] & [38,39] & [28,30] & [4,8] & [28,39] & [4,39] & [28,39] & [4,39] \\
San Antonio & 40 & [38,41] & [20,42] & [40,42] & [37,42] & [19,23] & [40,42] & [35,42] & [19,24] & [35,43] & [18,44] & [35,43] & [18,44] \\
Las Vegas & 41 & [35,41] & [20,43] & [40,42] & [32,38] & [21,26] & [40,42] & [32,40] & [20,28] & [32,42] & [18,44] & [32,42] & [17,44] \\
Charlotte & 42 & [42,47] & [42,50] & [40,42] & [45,49] & [49,50] & [40,42] & [45,49] & [49,50] & [40,49] & [39,50] & [40,49] & [38,50] \\
Miami & 43 & [41,43] & [13,43] & [43,43] & [39,42] & [11,18] & [43,44] & [39,42] & [10,18] & [39,43] & [11,44] & [39,44] & [10,44] \\
Atlanta & 44 & [44,48] & [42,50] & [44,46] & [47,49] & [45,49] & [44,46] & [47,49] & [44,49] & [43,49] & [38,50] & [43,49] & [37,50] \\
Port St. Lucie & 45 & [43,46] & [29,46] & [44,46] & [43,44] & [29,33] & [43,47] & [43,45] & [29,33] & [43,46] & [28,48] & [42,47] & [27,48] \\
Fresno & 46 & [43,46] & [22,46] & [44,47] & [43,44] & [21,26] & [44,49] & [43,44] & [20,27] & [43,47] & [21,48] & [43,49] & [20,49] \\
Tampa & 47 & [45,47] & [31,47] & [46,49] & [45,47] & [29,33] & [45,50] & [44,47] & [29,34] & [44,49] & [28,49] & [43,50] & [28,50] \\
Orlando & 48 & [46,48] & [33,48] & [47,50] & [45,47] & [29,33] & [46,50] & [45,47] & [29,34] & [44,50] & [28,50] & [44,50] & [28,50] \\
Jacksonville & 49 & [49,49] & [37,49] & [47,50] & [47,49] & [36,42] & [46,50] & [47,49] & [36,43] & [45,50] & [34,50] & [44,50] & [33,50] \\
New Orleans & 50 & [50,50] & [41,50] & [48,50] & [50,50] & [36,46] & [47,50] & [50,50] & [36,47] & [48,50] & [35,50] & [47,50] & [34,50] \\
\midrule
Average width & & 5.6 & 17.6 & 3.7 & 3.8 & 6.3 & 4.6 & 4.7 & 7.8 & 8.9 & 21.8 & 9.6 & 22.7 \\
\bottomrule
\end{tabular}
 
 \end{center}
 \begin{spacing}{1}
 {\footnotesize
 \textbf{Data:} Opportunity Insights.\\
 \textbf{Note:} This table reports normative, statistical, and combined sets for the 50 largest CZs. Details of the construction are discussed in Appendix~\ref{app: confsets}. The normative dominance sets $R^{[0,2]}$ and $R^{[0,\infty)}$ are the baseline results reported in Figure~\ref{fig:rol_cz}. Statistical confidence sets account for sampling variation in the CZ-specific frequency-weight matrices (Equation \ref{eq:fmap}). They are reported for the two limits of inequality-of-opportunity aversion $\theta$ and an intermediate level ($\theta=2$). Following \citet{MRSW24}, we report marginal confidence sets for the rankings of individual CZs and simultaneous confidence sets for the full ranking profile. Combined sets apply the normative dominance evaluation over both $\theta \in [0,2]$ and $\theta \in [0,\infty)$ to the marginal and simultaneous confidence sets. \textit{Rank} reports the ranking by $V_{0,u}$; rank 1 (50) corresponds to the highest (lowest) $V_{0,u}$.}
 \end{spacing}
 \end{table}

  \begin{table}[H]
 \footnotesize
 \begin{center}
 \caption{Dominance sets and sharp sets for the 50 largest CZs}
   \label{tab:cz-sharp}
 \renewcommand{\arraystretch}{0.5}
\begin{tabular}{lc*{2}{cc}}
\toprule
 & & \multicolumn{2}{c}{Dominance sets \(R_j\)} & \multicolumn{2}{c}{Sharp sets \(S_j\)} \\
\cmidrule(lr){3-4} \cmidrule(lr){5-6}
Commuting zone & Rank & \(R^{[0,2]}\) & \(R^{[0,\infty)}\) & \(S^{[0,2]}\) & \(S^{[0,\infty)}\) \\
\midrule
Minneapolis & 1 & [1,2] & [1,8] & [1,2] & [1,8] \\
Salt Lake City & 2 & [1,2] & [1,2] & [1,2] & [1,2] \\
Boston & 3 & [3,3] & [3,9] & [3,3] & [3,9] \\
Newark & 4 & [4,7] & [4,17] & [4,7] & [4,17] \\
Pittsburgh & 5 & [4,6] & [2,12] & [4,5] & [4,9] \\
Washington DC & 6 & [5,7] & [4,15] & [6,7] & [5,14] \\
Bridgeport & 7 & [7,8] & [7,24] & [7,8] & [7,24] \\
Manchester & 8 & [4,8] & [2,10] & [4,8] & [2,8] \\
Milwaukee & 9 & [9,21] & [9,38] & [9,21] & [9,38] \\
Denver & 10 & [9,10] & [6,16] & [9,10] & [9,12] \\
Providence & 11 & [10,14] & [9,22] & [10,13] & [10,19] \\
Buffalo & 12 & [11,16] & [11,30] & [12,15] & [12,28] \\
Grand Rapids & 13 & [11,16] & [11,33] & [13,14] & [13,32] \\
Seattle & 14 & [10,15] & [5,18] & [10,14] & [7,14] \\
San Francisco & 15 & [11,16] & [3,17] & [12,15] & [5,15] \\
San Jose & 16 & [10,16] & [2,16] & [10,16] & [2,16] \\
Philadelphia & 17 & [17,22] & [16,35] & [17,22] & [17,34] \\
Kansas City & 18 & [16,22] & [15,32] & [17,20] & [17,30] \\
Fort Worth & 19 & [16,20] & [11,24] & [16,19] & [14,19] \\
Chicago & 20 & [20,26] & [19,36] & [20,26] & [20,35] \\
Houston & 21 & [17,23] & [9,24] & [18,21] & [11,21] \\
St. Louis & 22 & [22,30] & [22,43] & [22,30] & [22,43] \\
Baltimore & 23 & [23,31] & [22,41] & [23,31] & [23,40] \\
Portland & 24 & [16,24] & [12,29] & [16,24] & [13,24] \\
Cincinnati & 25 & [25,33] & [25,49] & [25,33] & [25,49] \\
Dallas & 26 & [22,27] & [17,36] & [23,26] & [23,28] \\
Indianapolis & 27 & [26,33] & [26,46] & [27,32] & [27,45] \\
San Diego & 28 & [18,28] & [3,28] & [18,28] & [3,28] \\
Austin & 29 & [24,30] & [18,37] & [24,29] & [22,29] \\
New York & 30 & [26,31] & [15,34] & [27,30] & [18,30] \\
Columbus & 31 & [31,36] & [30,47] & [31,36] & [31,46] \\
Cleveland & 32 & [32,37] & [28,47] & [32,37] & [32,43] \\
Dayton & 33 & [31,36] & [27,45] & [32,34] & [32,39] \\
Sacramento & 34 & [24,34] & [14,36] & [24,34] & [16,34] \\
Phoenix & 35 & [28,35] & [19,40] & [28,35] & [22,35] \\
Detroit & 36 & [36,38] & [32,47] & [36,38] & [36,44] \\
Raleigh & 37 & [37,42] & [36,48] & [37,42] & [37,47] \\
Nashville & 38 & [37,41] & [28,46] & [38,40] & [36,41] \\
Los Angeles & 39 & [29,39] & [6,39] & [29,39] & [6,39] \\
San Antonio & 40 & [38,41] & [20,42] & [39,41] & [21,41] \\
Las Vegas & 41 & [35,41] & [20,43] & [35,41] & [21,41] \\
Charlotte & 42 & [42,47] & [42,50] & [42,47] & [42,50] \\
Miami & 43 & [41,43] & [13,43] & [41,43] & [13,43] \\
Atlanta & 44 & [44,48] & [42,50] & [44,48] & [44,49] \\
Port St. Lucie & 45 & [43,46] & [29,46] & [43,45] & [30,45] \\
Fresno & 46 & [43,46] & [22,46] & [43,46] & [22,46] \\
Tampa & 47 & [45,47] & [31,47] & [45,47] & [31,47] \\
Orlando & 48 & [46,48] & [33,48] & [46,48] & [33,48] \\
Jacksonville & 49 & [49,49] & [37,49] & [49,49] & [37,49] \\
New Orleans & 50 & [50,50] & [41,50] & [50,50] & [41,50] \\
\midrule
Average width & & 5.6 & 17.6 & 5.0 & 14.3 \\
\bottomrule
\end{tabular}

 \end{center}
   \begin{spacing}{1}
 {\footnotesize
 \textbf{Data:} Opportunity Insights.\\
 \textbf{Note:} This table compares the dominance sets $R^{[0,2]}$ and $R^{[0,\infty)}$ reported in Figure~\ref{fig:rol_cz} with the sharp attainable rank sets $S^{[0,2]}$ and $S^{[0,\infty)}$ for the 50 largest CZs. Dominance sets contain every rank a CZ could hold in any complete ranking consistent with the unambiguous pairwise dominance preorder over the respective range of $\theta$. Sharp sets span only the ranks a CZ actually attains across that range. \textit{Rank} reports the ranking by $V_{0,u}$; rank 1 (50) corresponds to the highest (lowest) $V_{0,u}$.}
 \end{spacing}
 \end{table}

 \begin{figure}[htbp]
  \caption{Comparison between \citet{Chetty2014} and \citet{Chetty2026}}
  \begin{center}
  \label{fig:vintages}
  \begin{subfigure}[t]{0.49\textwidth}
    \centering
\includegraphics[width=\textwidth]{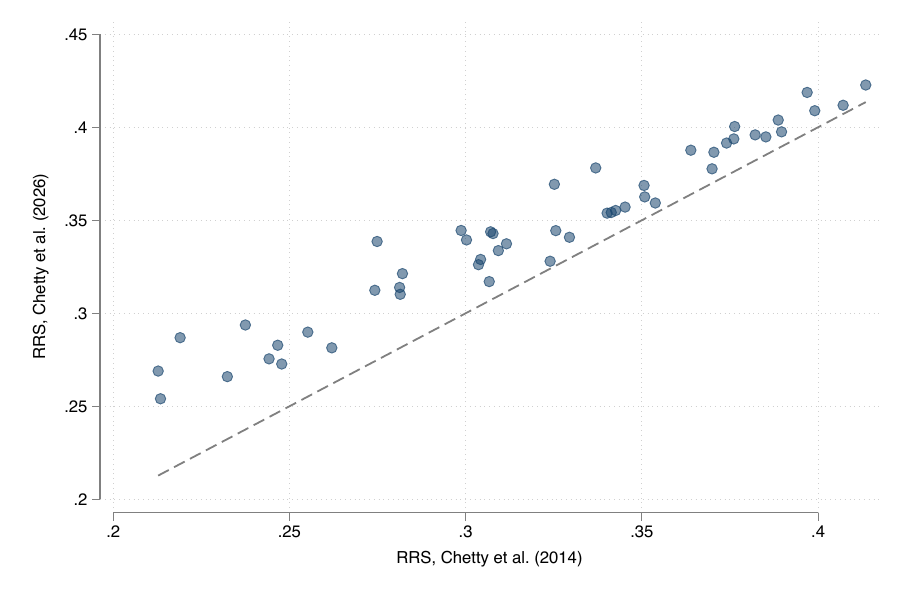}
  \caption{RRS}
  \end{subfigure}
  \begin{subfigure}[t]{0.49\textwidth}
    \centering
\includegraphics[width=\textwidth]{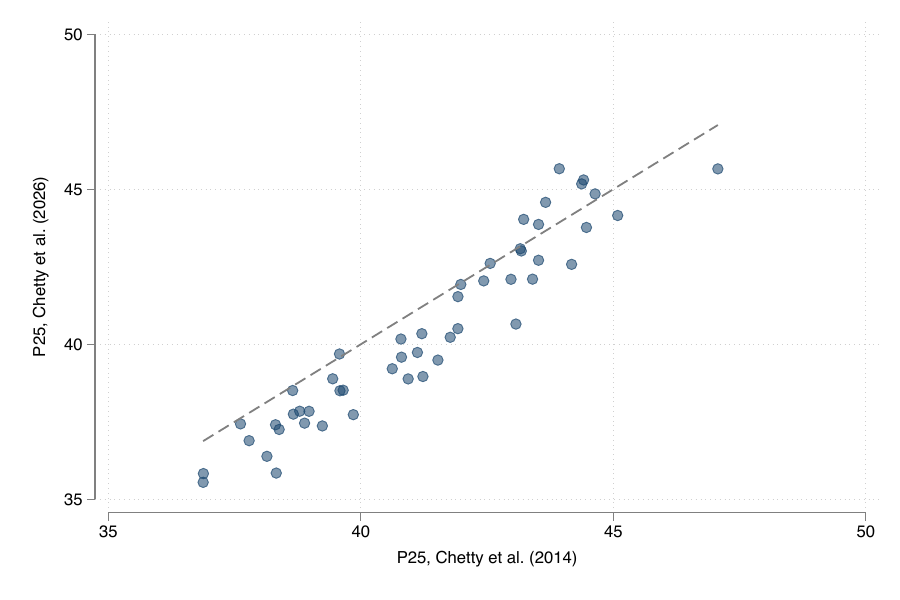}
  \caption{P25}
  \end{subfigure}
\end{center}
\begin{spacing}{0.9}
{\footnotesize
\textbf{Data:} Opportunity Insights.\\
\textbf{Note:} This figure compares the data of \citet{Chetty2014} and \citet{Chetty2026} for relative and absolute mobility. Panel (a) plots, for each of the 50 largest CZs, the rank-rank slope (RRS) reported by \citet{Chetty2014} against the RRS based on \citet{Chetty2026}, which is approximated by the slope of the secant through children's mean income rank at the 25th and 75th percentiles of the parental income distribution. Panel (b) plots the mean child rank at the 25th percentile of parental income (P25), based on \citet{Chetty2014} and \citet{Chetty2026}. Estimates of \citet{Chetty2014} are based on the 1980--1982 birth cohort.}
\end{spacing}
\end{figure}

\clearpage
\section{Normative weights and absolute mobility}
\label{app: normweights}
Normative weights $p^\star_{\theta, u}(f,q;s)$ redistribute attention away from the population average toward types with lower expected utility (Figure \ref{fig: theosum}). This provides a bridge to the intergenerational mobility literature, in which various studies calculate measures of absolute mobility by anchoring their analyses to specific parental income percentiles. For instance, following the influential work of \citet{Chetty2014, Chetty2026}, a common benchmark is the expected income of children whose parents sit at the 25th percentile.

To illustrate this bridge, recall the expression for normative weights (Equation~\eqref{opt-p}): 
\begin{equation}
    p^\star_{\theta, u}(f,q;s)=\frac{q(s) e^{-\theta U(\pi_s)}}{\sum_{s \in \states} q(s) \left(e^{- \theta U(\pi_s)} \right)}\,.
\end{equation}
Now, suppose that types $s\in\mathcal{S}$ are defined by parental income percentiles $r^p\in\{1,..., 100\}$, so that $q(s) = 1/100$ for all $s$, and that expected type utilities are estimated using the linear specification of the income--rank slope of Equation \eqref{eq:irs}. 

Then, we can express the normative focus of the social planner by calculating the average parental income rank considered in their evaluation:
\begin{equation}
    \tilde{r}^p=\sum_{s\in\states}p^\star_{\theta, u}(f,q;s)\times r_s^p=\sum_{s\in\states}\frac{e^{-\theta (\hat{\beta}^{{IRS}} \times r_s^p)}}{\sum_{s \in \states} \left(e^{- \theta (\hat{\beta}^{{IRS}} \times r_s^p)} \right)}\,\times r_s^p, \label{weightedrank}
\end{equation}
where the intercept $\hat\alpha$ cancels as it is uniform across $s$. Note that $\tilde{r}^p$ is strictly decreasing if and only if $\hat{\beta}^{IRS}>0$.

There are two polar cases. At $\theta = 0$ the weights are uniform, $p^\star_0(s) = q(s)$, so $\tilde{r}^p = 50.5$: The planner weights all parental percentiles equally. As $\theta \to \infty$ the weights concentrate on the type with the lowest utility, which is the lowest percentile if $\hat{\beta}^{IRS}>0$, so $\tilde{r}^p \to 1$: The planner focuses on the worst-off only. Between the limiting cases, $\tilde{r}^p$ traces out a monotonically decreasing curve based on Equation \eqref{weightedrank}.

Figure \ref{fig: normweights} illustrates the relationship between $\theta$ and $\tilde{r}^p$ using national US data. At the given level of intergenerational mobility in the US ($\hat{\beta}^{IRS} = 0.011$), an average normative focus on the 25th percentile corresponds to $\theta \approx 3.3$, whereas a focus on the 10th percentile requires $\theta \approx 9.3$. Importantly, while measures of absolute mobility focus on a single percentile, $\tilde{r}^p$ provides an average equivalent that allows for weights across neighboring parental income percentiles.

\begin{figure}[htbp]
\caption{Mapping parental income ranks and $\theta$}
\label{fig: normweights}
\begin{center}
\includegraphics[width=\textwidth]{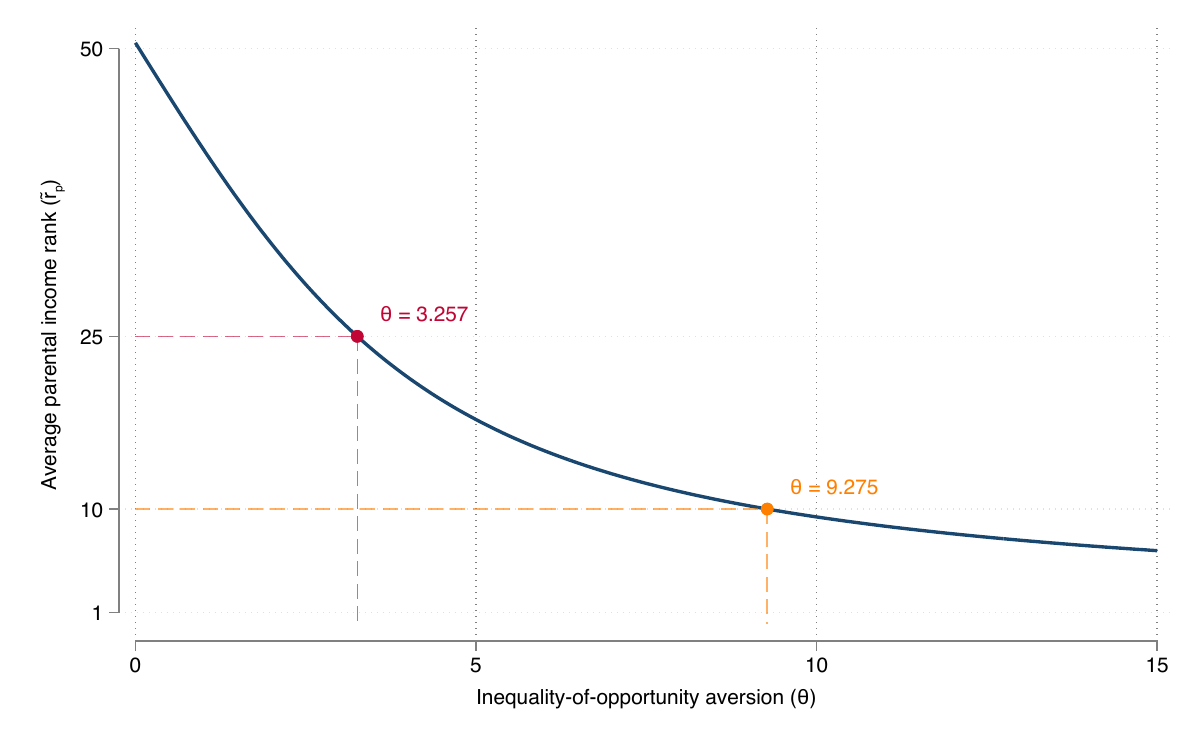}
\end{center}
\begin{spacing}{0.9}
{\footnotesize
\textbf{Data:} Opportunity Insights.\\
\textbf{Note:}  This figure shows the weighted average parental income rank as a function of $\theta$ calibrated on intergenerational mobility at the national level in the US based on our reconstructed dataset ($\hat{\beta}^{RRS} = 0.32$ and $\hat{\beta}^{IRS} = 0.011$). Dashed lines and markers indicate the values of $\theta$ at which the planner's average focus equals the 25th and 10th percentiles of parental income, respectively.}
\end{spacing}
\end{figure}

\end{document}